\documentclass[11pt,a4paper]{article}
\usepackage{amssymb,amsmath, amsfonts}
\usepackage{graphicx,graphics}
\usepackage{mathtools}
\usepackage{bbold}
\usepackage[english]{babel}
\usepackage{epsfig,url}
\usepackage{bbm,theorem}
\usepackage{a4wide}
\usepackage{color}
\usepackage{enumerate}
\usepackage{calrsfs}
\usepackage[bookmarks=true]{hyperref}
\usepackage{bookmark}
\usepackage{amsmath}
\usepackage{amsfonts}
\usepackage{amssymb}
\usepackage{verbatim}
\usepackage{graphicx}

\providecommand{\U}[1]{\protect\rule{.1in}{.1in}}
\DeclareMathAlphabet{\pazocal}{OMS}{zplm}{m}{n}
\newtheorem{theorem}{Theorem}[section]
\newtheorem{definition}[theorem]{Definition}
\newtheorem{corollary}[theorem]{Corollary}
\newtheorem{lemma}[theorem]{Lemma}
\newtheorem{proposition}[theorem]{Proposition}
\newtheorem{assumption}[theorem]{Assumption}

{\theorembodyfont{\upshape}
\newtheorem{remark}[theorem]{Remark}

}
\numberwithin{equation}{section}
\numberwithin{theorem}{section}
\newcommand{\qed}{\hfill$\Box$}
\newenvironment{proof}{\begin{trivlist}\item[]{\em Proof:}\/}{\qed\end{trivlist}}
\newenvironment{proofof}[1][Proof]{\noindent \textit{#1.} }{\ \qed}

\newcommand{\R}{{\mathbb R}}

\newcommand{\N}{{\mathbb N}}

\newcommand{\ep}{\varepsilon}

\newcommand{\rme}{{\rm e}}

\DeclareMathOperator*{\supp}{supp}

\newcommand{\eps}{{\varepsilon}}

\newcommand{\beq}{\begin{equation}}
\newcommand{\eeq}{\end{equation}}
\newcommand{\beqs}{\begin{eqnarray}}
\newcommand{\eeqs}{\end{eqnarray}}

\newcommand{\norm}[1]{\Vert #1\Vert}
\newcommand{\defset}[2]{ \left\{ #1\left|\,
#2\makebox[0cm]{$\displaystyle\phantom{#1}$}\right.\!\right\} }

\newcommand{\vep}{\varepsilon}

\newcounter{jlisti}

\begin{document}

\title{ Stationary non-equilibrium solutions for coagulation systems}
\author{Marina A. Ferreira, Jani Lukkarinen, Alessia Nota, Juan J. L.
Vel\'azquez}
\maketitle

\begin{abstract} 
We study coagulation equations under non-equilibrium conditions which are induced by the 
addition of a source term for small cluster sizes.
We consider both discrete and continuous coagulation equations, and allow for a large class of 
coagulation rate kernels, with the main restriction being 
boundedness from above and below by certain weight functions.   The weight functions depend on two power 
law parameters, and the assumptions cover, in particular, the 
commonly used free molecular and diffusion limited 
aggregation coagulation kernels.
Our main result shows that the two weight function parameters 
already determine whether there exists a stationary solution under the 
presence of a source term.  In particular, we find that the diffusive kernel allows for 
the existence of 
stationary solutions while there cannot be any such solutions for the free molecular kernel.   The argument to 
prove the non-existence of solutions relies on a novel power law lower bound, valid in the appropriate 
parameter regime, for the decay of stationary solutions with a constant flux. We obtain optimal lower 
and upper estimates of the solutions for large cluster sizes, and prove that the solutions of the discrete 
model behave asymptotically as solutions of the continuous model.
\end{abstract}

\bigskip


\tableofcontents

\bigskip 

\bigskip 

\section{Introduction}

\bigskip

Atmospheric cluster formation processes  \cite{Fried}, where certain species of the gas molecules (called {\em monomers\/}) can stick together 
and eventually produce macroscopic particles, are an important component in cloud formation and radiation scattering. 
The above cluster formation processes are modelled with the so-called General Dynamic Equation (GDE) \cite{Fried}. 
Under atmospheric conditions, the particle clusters are often aggregates of various molecular species and formed by collisions of several different monomer types, cf.\ \cite{Vehkam,Vehkamki2012} for more details and examples.  Accordingly, in the GDE one needs to label clusters not only by the total number of monomers in them but also by counting each monomer type.  This results in multicomponent labels for the concentration vector, with nonlinear interactions between the components.
Another feature of the GDE which has been largely absent 
from most of the previous mathematical work on coagulation equations,
is the presence of an external monomer source term.  Such sources are nevertheless important for atmospheric phenomena (for more details about the chemical and physical origin and relevance of the sources we refer for instance to \cite{Ehn2014, Kerminen2018}), although this problem has been barely considered in the mathematical literature. 
 
In this work, we focus on the effect the addition of a source term has on solutions of standard one-component coagulation equations.
This is by no means to imply that multicomponent  coagulation equations would not have interesting new mathematical features but these will be the focus of a separate work.  Here, we consider only one species of monomers, and we are interested in the distribution of the concentration of clusters formed out of these monomers.  Let $n_\alpha\ge 0$
denote the concentration of clusters with $\alpha\in \N$ monomers.

Considering the regime in which the precise spatial structure and loss of particles by deposition are not important, the GDE yields the following nonlinear evolution equation 
for the concentrations $n_\alpha$:
\begin{align}
& \partial_{t}n_{\alpha} \nonumber 
=\frac{1}{2}\sum_{0< \beta<\alpha}K_{\alpha-\beta,\beta
}n_{\alpha-\beta}n_{\beta}-n_{\alpha}\sum_{\beta>0}K_{\alpha,\beta}n_{\beta
}
\\ &\qquad +\sum_{\beta>0}\Gamma_{\alpha+\beta,\alpha}n_{\alpha+\beta}-\frac{1}{2}%
\sum_{0<\beta<\alpha}\Gamma_{\alpha,\beta}n_{\alpha}+
s_{\alpha}\, .
\label{B1}
\end{align}
The coefficients $K_{\alpha,\beta}$ describe the coagulation rate joining two clusters of sizes $\alpha$ and $\beta$ into a cluster of size $\alpha+\beta$, as dictated by mass conservation.  Analogously, the coefficients $\Gamma_{\alpha,\beta}$ 
describe the fragmentation rate of clusters of size $\alpha$
into two clusters which have sizes $\beta$ and $\alpha-\beta$.
We denote with $s_{\alpha}$ the (external) source of clusters of size $\alpha$.  In applications, typically only monomers or small clusters are being produced, so we make the assumption that the function $\alpha \mapsto s_\alpha$ has a bounded, non-empty support.  In the following, we make one further simplification and consider only cases where also fragmentation can be ignored, $\Gamma_{\alpha,\beta}=0$; the reasoning behind this choice is discussed later in Sec.~\ref{sec:rateintro}.
An overview of the currently available mathematical results for coagulation-fragmentation models can be found in \cite{C15, LM}.

Therefore, we are led to study the evolution equation
\begin{equation}
\partial_{t}n_{\alpha}=\frac{1}{2}\sum_{\beta<\alpha}K_{\alpha-\beta,\beta
}n_{\alpha-\beta}n_{\beta}-n_{\alpha}\sum_{\beta>0}K_{\alpha,\beta}n_{\beta} + s_{\alpha}\,. 
\label{B2}
\end{equation}
In this paper, we are concerned with the existence or nonexistence of steady 
state solutions to (\ref{B2}) for general coagulation rate kernels $K$, including in particular the physically relevant kernels discussed in Sec.~\ref{sec:rateintro}.
The source is here assumed to be localized on the ``left boundary'' of the system which have small cluster sizes.  Such source terms often lead to nontrivial stationary solutions towards which the time-dependent solutions evolve as time increases.  These stationary solutions are nonequilibrium steady states since they involve a steady flux of matter from the source into the system.    
The characterization of  nonequilibrium stationary states exhibiting transport phenomena is one of the central problems in statistical mechanics. 

The main result of this paper gives a contribution in this direction. More precisely, we address the question of existence of such stationary solutions to (\ref{B2}).
We prove that for a large class of kernels---including in particular the diffusion limited aggregation kernel given in (\ref%
{KernBrow})---stationary solutions to (\ref{B2}) yielding a constant flux of
monomers towards clusters with large sizes exist. On the contrary, for a different class of kernels---including the free molecular coagulation kernel with the form (\ref{eq:BalKer})---such a class of stationary solutions does not exist.

In the case of collision kernels for which stationary
nonequilibrium solutions to (\ref{B2}) exist, we can even compute the rate of
formation of macroscopic particles, which we identify here with infinitely large particles,
from an analysis of the properties of these stationary solutions, cf.
Section \ref{sec:types_sol}.
We find that in this case the main mechanism of transport of monomers to
large clusters corresponds to coagulation between clusters
with comparable sizes, cf.\ Lemma \ref{lem:J},
Section \ref{sec:estimates}.

{
The non-existence of such stationary solutions under a monomer source for a
general class of coagulation kernels yielding coagulation for arbitrary
cluster sizes is one of the novelties of our work. It has been pointed out in
Remark 8.1 of \cite{Dub} that for kernels $K_{\alpha,\beta}$ which vanish if
$\alpha>1$ or $\beta>1$, and sources $s_{\alpha}$ which are different from zero
for $\alpha\geq2$, stationary solutions of (\ref{B2}) cannot exist. Although the
example in \cite{Dub} refers to the continuous counterpart of (\ref{B2}) (c.f. (\ref{eq:time_evol})), 
the argument works similarly for discrete kernels. The example of non
existence of stationary solutions in \cite{Dub} relies on the fact that coagulation
does not take place for sufficiently large particles and therefore cannot
compensate for the addition of particles due to the source term $s_{\alpha}.$
In the class of kernels considered in this paper coagulation takes place for
all particle sizes and therefore the nonexistence of steady states must be due
to a different reason.
At first glance} this result might appear counterintuitive since this non-existence result includes kernels for which the dynamics seems to be well-posed.  Hence, one needs to explain what will happen at large times to the monomers injected into the system.
Our results suggest that for such kernels the aggregation of monomers with large clusters is so fast that it cannot be compensated by the constant addition of monomers described by the injection term $s_\alpha$.  Then the cluster concentration  $n_{\alpha}$ would converge to zero as $t\to \infty$ for bounded $\alpha$ even if  $n_{\alpha}=0$ is not a stationary solution to \eqref{B2} if $(s_{\beta})\neq 0$.

We remark that our non-existence result of stationary solutions includes in particular the so called free molecular kernel (cf. \eqref{KerBall-first} below) derived from kinetic theory which is commonly used for microscopic computations involving aerosols (cf. for instance \cite{Vehkam}). 

In this paper we consider, in addition to the stationary solutions of (\ref{B2}), also the stationary solutions of the continuous counterpart of (\ref{B2}),
\begin{equation}
\partial _{t}f(x,t)=\frac{1}{2}\int_{0}^{x}K\left( x-y,y\right) f\left(
x-y,t\right) f\left( y,t\right) dy-\int_{0}^{\infty }K\left( x,y\right)
f\left( x,t\right) f\left( y,t\right) dy+\eta \left( x\right) .
\label{eq:time_evol}
\end{equation}
In fact, we will allow $f$ and $\eta$ in this equation to be positive measures.  This  will make it possible to study the continuous and discrete equations simultaneously, using Dirac $\delta$-functions to connect $f(\xi)$ and $n_\alpha$ via the formula
$f(\xi) d\xi = \sum_{\alpha=1}^\infty n_\alpha \delta(\xi - \alpha)d\xi$.

In most of the mathematical studies of the
coagulation equation to date, it has been assumed that the injection terms $s_{\alpha}$ and $\eta\left(x\right) $ are absent. In the case of homogeneous kernels, i.e., kernels satisfying \begin{equation}
K(rx,ry)=r^\gamma K(x,y)\label{eq:homogeneity_cond}
\end{equation}
 for any $r>0$,
the long time asymptotics of the solutions of (\ref{eq:time_evol}) with $\eta\left(
x\right) =0$ might be expected to be self-similar for a large class of
initial data. This has been rigorously proved in \cite{MP04} for the
particular choices of kernels $K( x,y) =1$ and $K( x,y) =x+y$.  In the case of
 discrete problems, the distribution of clusters $%
n_{\alpha}$ has also been proved to behave in self-similar form for large times
and for a large class of initial data if the kernel is constant, $%
K_{\alpha,\beta}=1$, or additive, $K_{\alpha,\beta}=\alpha+\beta$ \cite{MP04}. For these kernels it is possible to find explicit
representation formulas for the solutions of (\ref{B2}), (\ref{eq:time_evol}) using
Laplace transforms.

For general homogeneous kernels construction of explicit self-similar 
solutions is no longer possible.  However, the existence of self-similar solutions of (\ref{eq:time_evol})
with $\eta=0$ has been proved for certain classes of homogeneous
kernels $K( x,y) $ using fixed point methods. These solutions
might have a finite monomer density (i.e., $\int_{0}^{\infty}xf\left(
x,t\right) dx<\infty$) as in \cite{EM05, FL05}, or infinite monomer density (i.e., $\int_{0}^{\infty}xf%
\left( x,t\right) dx=\infty$) as in \cite{BNV18, BNV19,NV13, NTV16}. Similar strategies can be applied to other kinetic equations \cite{GPV, JNV, NNTV}.

Problems like (\ref{B2}), (\ref{eq:time_evol}) with nonzero injection terms $s_{\alpha}$, $%
\eta\left( x\right) $ have been much less studied both in the physical and
mathematical literature. In \cite{DKW99} it has been observed using a combination of
asymptotic analysis arguments and numerical simulations that solutions of (\ref{B2}), (\ref{eq:time_evol}) with a finite  monomer density behave in
self-similar form for long times and for a class of homogeneous coagulation kernels, even considering source terms which depend on time following a power law $t^{\omega}$.  Coagulation equations with sources have also been considered in \cite{LK02} using  Renormalization Group methods and leading to predictions of analogous self-similar behaviour. For what concerns the rigorous mathematical literature, in \cite{Dub}, the existence of stationary solutions has been obtained in the case of bounded kernels.   
Well-posedness of the time-dependent problem for a class of homogeneous coagulation kernels with homogeneity $\gamma \in [0,2]$ has been proven in \cite{EM06}. 
 For the constant kernel, the stability of the corresponding solutions has
been proven using Laplace transform methods (cf. \cite{Dub}).
Convergence to equilibrium for  a class of coagulation equations containing also growth terms as well as sources has been studied in \cite{G,GZ}.
Analogous
stability results for coagulation equations with the form of \eqref{B1} but
containing an additional removal term on the right-hand side with the form
$-r_{\alpha}n_{\alpha},\ r_{\alpha}>0$ have been obtained in \cite{KT19}.

In this paper we study  the stationary solutions of (\ref{B2}), (\ref{eq:time_evol}) for
 coagulation kernels satisfying 
\begin{align}\label{eq:cond_kernel}
&c_{1} w(x,y)\leq K( x,y) \leq
c_{2}w(x,y)\,, \qquad 
w(x,y) := x^{\gamma +\lambda }y^{-\lambda }+ y ^{\gamma +\lambda }x^{-\lambda }
\,,
\end{align}%
for some $c_1,c_2>0$ and for all $x,y$.
The weight function $w$ depends on two real parameters: the homogeneity parameter $\gamma$ and the ``off-diagonal rate'' 
parameter $\lambda$.   The parameter $\gamma$ yields the behaviour of kernel $K$ under the scaling of the particle size while 
the parameter $\lambda$ measures how relevant the coagulation events between particles of different sizes are.
However, let us stress that we do not assume the kernel $K$ itself to be homogeneous, even though the weight functions are.

The main result of this paper is the following. Given $\eta$ compactly supported there exists at least one nontrivial stationary solution to the problem (\ref{eq:time_evol}) if and only if $|\gamma +2\lambda|<1$. In particular, if $|\gamma +2\lambda|\ge 1$ no such
stationary solutions can exist. Note that the parameters $\gamma$ and $\lambda$ are arbitrary real numbers and they may be negative or greater than one here.   
Therefore, these results do not depend on having global well-posedness of
mass-preserving solutions for the time-dependent problem \eqref{eq:time_evol}. In particular,
our theorems cover ranges of parameters for which the solutions to the
time-depend problem \eqref{eq:time_evol} can exhibit gelation in finite or zero time. A
detailed description of the current state of the art concerning wellposedness
and gelation results can be found in \cite{BLL}. At a first glance, the fact that the
existence of stationary solutions of \eqref{eq:time_evol} does not depend on having or not
solutions for the time dependent problem might appear surprising. However, the
reason for this becomes  clearer if we notice that the homogeneity of the
kernel is one of the main factors determining the wellposedness of the time
dependent problem \eqref{eq:time_evol}. On the other hand, the homogeneity of the kernel $K$
is not an essential property of the stationary solution problem as it can be
seen (cf. \cite{DP}) noticing that if $f$ is a stationary solution of \eqref{eq:time_evol}, then $x^{\theta
}f\left(  x\right)  $ is a stationary solution of \eqref{eq:time_evol} with kernel
$\frac{K\left(  x,y\right)  }{\left(  xy\right)  ^{\theta}}$ and the same
source $\eta$. This new kernel satisfies \eqref{eq:cond_kernel} with $\gamma$ and $\lambda$
replaced by $\left(  \gamma-2\theta\right)  $ and $\lambda+\theta$
respectively. In particular, we can so obtain kernels with arbitrary homogeneity
and having basically the same steady states, up to the product by a power
law.

We also prove in this paper the analog of these existence and non-existence results for the discrete coagulation equation \eqref{B2}. Moreover, we derive upper and lower estimates, as well as regularity results, for the stationary solutions to \eqref{B2}, \eqref{eq:time_evol} for the range of parameters for which these solutions exist, i.e.  $|\gamma +2\lambda|<1$. Finally, we also describe the asymptotics for large cluster sizes of these stationary solutions.

\subsection{On the choice of coagulation and fragmentation rate functions}\label{sec:rateintro}

Although we do not keep track of any 
spatial structure, the coagulation rates $K_{\alpha,\beta}$ do depend on the specific
mechanism which is responsible for the aggregation of the
clusters. These coefficients need to be computed for example using kinetic theory 
and the result will depend on what is assumed about the 
particle sizes and the processes yielding the motion of the clusters.

For instance, in the case of electrically neutral
particles with a size much smaller than the mean free path between two
collisions between clusters, 
the coagulation kernel is (cf.\ \cite{Fried})%
\begin{equation}
K_{\alpha,\beta}=\left( \frac{3}{4\pi}\right) ^{\frac{1}{6}}\sqrt{6k_{B}T}%
\left( \frac{1}{m( \alpha) }+\frac{1}{m( \beta) }%
\right) ^{\frac{1}{2}}\left( V( \alpha)^{\frac{1}{3}}
+V( \beta) ^{\frac{1}{3}}\right) ^{2} 
\label{KerBall-first}
\end{equation}
where $V( \alpha) $ and $m( \alpha) $ are respectively
the volume and the mass of the cluster characterized by the composition $%
\alpha$.  We denote as $k_{B}$ the Boltzmann constant, as $T$ the absolute temperature, and if $m_1$ is the mass of one monomer, we have above $m(\alpha)= m_1 \alpha$.  In the derivation, one also assumes a spherical shape of the clusters.  If the particles are distributed inside the sphere with a uniform mass density $\rho$, assumed to be independent of the cluster size, we also have $V(\alpha)=\frac{m_1}{\rho} \alpha$.  Changing the time-scale we can set all  the physical constants to one. Finally, it is possible to define  a continuum function $K(x,y)$ by setting 
$\alpha=x$, $\beta=y$ in the above formula. We call this function the {\em free molecular coagulation kernel}, given explicitly by 
\begin{equation}
K( x,y) =\big(x^{\frac{1}{3}}+y^{\frac{1}{3}}\big)^{2}\big(%
x^{-1}+y^{-1}\big)^{\frac{1}{2}}.  \label{eq:BalKer}
\end{equation} 
It is now straightforward to check that with the parameter choice $\gamma=\frac{1}{6}$, $\lambda=\frac{1}{2}$ 
there are $c_1,c_2>0$ such that  \eqref{eq:cond_kernel} holds for all $x,y>0$.
Since here $\gamma + 2 \lambda = \frac{7}{6}>1$, the free molecular kernel belongs to the second category which has no stationary state.

Another often encountered example is diffusion limited aggregation which was studied already in the original work by Smoluchowski  \cite{S16}.  Suppose that there is a background of non-aggregating neutral particles producing cluster paths resembling Brownian motion between their collisions.  Then one arrives at the coagulation kernel 
\begin{equation}
K_{\alpha,\beta}=\frac{2k_{B}T}{3\mu}\left( \frac{1}{ V(
\alpha) ^{\frac{1}{3}}}+\frac{1}{ V( \beta)^{\frac{1}{3}}}\right) 
\left( V( \alpha)^{\frac{1}{3}}
+ V( \beta) ^{\frac{1}{3}}\right) 
\label{KernBrow-first}
\end{equation}
where $\mu>0$ is the viscosity of the gas in which the clusters move.

As before, we then set $V(\alpha)=\frac{m_1}{\rho} \alpha$ and 
define a continuum function $K(x,y)$ by setting 
$\alpha=x$, $\beta=y$ on the right hand side of (\ref{KernBrow-first}).
The constants may then be collected together and after rescaling time one may use the following 
kernel function
\begin{equation}
K( x,y)= 
\left( x^{-\frac{1}{3}} + y^{-\frac{1}{3}} \right)
\left( x^{\frac{1}{3}}+ y^{\frac{1}{3}}\right) \,,
\label{KernBrow}
\end{equation}
which we call here {\em diffusive coagulation kernel\/} or  {\em Brownian kernel\/}.
In this case, for the parameter choice $\gamma=0$, $\lambda=\frac{1}{3}$ 
there are $c_1,c_2>0$ such that for all $x,y>0$ \eqref{eq:cond_kernel} holds.
Since here $0<\gamma + 2 \lambda = \frac{2}{3}<1$, the diffusive kernel belongs to the first category which will have some stationary solutions.

Several other coagulation kernels can be found in the physical and chemical literature. For instance, the derivation of  the free molecular kernel \eqref{KerBall-first}  and the Brownian kernel \eqref{KernBrow-first} is discussed in \cite{Fried}. The derivation of coagulation describing the aggregation between charged and neutral particles can be found in \cite{SB73}.
Applications of these three kernels to specific problems in chemistry can be found for instance in 
\cite{Vehkam}. 

Concerning the fragmentation coefficients $\Gamma_{\alpha,\beta}$, it is
commonly  assumed in the physics and chemistry literature  that these coefficients are related to the
coagulation coefficients by means of the following detailed balance condition (cf.\ for instance \cite{Vehkam})
\begin{equation}
\Gamma_{\alpha+\beta,\beta}=\frac{P_{\text{ref}}}{k_{B}T}\exp\!\left( \frac{\Delta
G_{\text{ref},\alpha+\beta}-\Delta G_{\text{ref},\alpha}-\Delta G_{\text{ref},\beta}}{k_{B}T}%
\right) K_{\alpha,\beta}   \label{DetailedBalance}
\end{equation}
where $\Delta G_{\text{ref},\alpha}$ is the Gibbs energy of formation of the
cluster $\alpha$ and $P_{\text{ref}}$ is the reference pressure at which these
energies of formation are calculated.  
Since we assume the coagulation kernel to be symmetric,  
$
K_{\alpha,\beta}=K_{\beta,\alpha}$, the fragmentation coefficients then satisfy a symmetry requirement
$\Gamma_{\alpha+\beta,\alpha}=\Gamma_{\alpha+\beta,\beta}$ for all $\alpha,\beta\in\mathbb{N}^{d}$.

In the processes of particle aggregation, usually the formation of larger
particles is energetically favourable, which means that%
\begin{equation*}
\Delta G_{\text{ref},\alpha+\beta}\ll\Delta G_{\text{ref},\alpha}+\Delta G_{\text{ref},\beta}\,.
\end{equation*}
Under this assumption, it follows from (\ref{DetailedBalance}) that%
\begin{equation*}
\Gamma_{\alpha+\beta,\beta}\ll K_{\alpha,\beta}\,,
\end{equation*}
and then we might expect to approximate the solutions of \eqref{B1} by means of the solutions of \eqref{B2}. 
The description of the precise conditions on the Gibbs free energy $\Delta
G_{\text{ref},\alpha }$ which would allow to make this approximation rigorous is an interesting mathematical problem that we do not address in the present paper.
Therefore,  we restrict our analysis here to the coagulation equations \eqref{B2} and \eqref{eq:time_evol}.

\subsection{Notations and plan of the paper}

Let $I$ be any interval such that $I \subset \R_+=[0,\infty)$. 
We reserve the notation $\R_*$ for the case $I=(0,\infty)$.
We will denote by $C_{c}(I )$ the space of compactly supported continuous
functions on $I$
and by $C_b(I)$ the space of functions that are continuous and bounded on $I$.
Unless mentioned otherwise, we endow both spaces with the standard supremum norm.  Then $C_b(I)$ is a Banach space and $C_{c}(I )$ is its subspace.
We denote the completion of $C_{c}(I)$ in $C_b(I)$ by $C_0(I)$ which naturally results in a Banach space.  For example, then $C_0(\R_+)$ is the space of continuous functions vanishing at infinity and 
$C_{0}(I )=C_{c}(I )=C_b(I)$ if $I$ is a finite, closed interval.

Moreover, we denote  by $ \mathcal{M}_{+}(I) $ the space of nonnegative  Radon measures on $I$.  
Since $I$ is locally compact, $ \mathcal{M}_{+}(I) $ can be identified with the space of positive linear functionals on $C_{c}(I )$ via Riesz--Markov--Kakutani theorem.
For measures $\mu\in \mathcal{M}_{+}(I) $, we denote its total variation norm by $\norm{\mu}$
and recall that since the measure is positive, we have $\norm{\mu}=\mu(I)$.
Unless $I$ is a closed finite interval, not all of these measures need to be bounded. 
The collection of bounded, positive measures is denoted by $\mathcal{M}_{+,b}(I):=\{\mu\in \mathcal{M}_{+}(I) \,|\,\mu(I)<\infty\}$ and we denote the collection of bounded complex Radon measures by  $\mathcal{M}_{b}(I)$.
We recall that the total variation indeed defines a norm in $\mathcal{M}_{b}(I)$, and this space is a Banach space which can be identified with the dual space $C_{0}(I )^*=C_c(T)^*$.
In addition, $\mathcal{M}_{+,b}(I)$ is a norm-closed subset of
$\mathcal{M}_{b}(I)$. 
Alternatively, we can endow both $\mathcal{M}_{b}(I)$ and  
$\mathcal{M}_{+,b}(I)$ with the $\ast$--weak topology which is generated by the functionals $\langle  \varphi,\mu\rangle=\int_I \varphi(x) \mu(d x)$ with $\varphi\in C_c(I)$.

We will use indistinctly $\eta(x) dx$ and $\eta(dx)$ to denote elements of these measure spaces.
The notation $\eta( dx) $ will be preferred when performing integrations or when we want to emphasize that the measure might not be absolutely continuous with respect to the Lebesgue measure. 

For the sake of notational simplicity,  in some of the proofs we will resort to a generic constant $C$ which  may change from line to line. 

\medskip

The paper is structured as follows. In Section \ref{sec:setting} we discuss the types of solutions considered here and we state the main results. In Section \ref{sec:existence} we prove the existence of steady states for the coagulation equation with source in the continuum case \eqref{eq:time_evol} assuming $|\gamma+2\lambda|<1$. We prove the complementary nonexistence of stationary solutions to \eqref{eq:time_evol} for $|\gamma+2\lambda|\geq 1$ in Section \ref{sec:nonexistence}. 
The analogous existence and nonexistence results for the discrete model \eqref{B2} are collected into 
Section \ref{sec:discr1d}.  In Section \ref{sec:estimates} we derive several further estimates for the solutions of both continuous and discrete models, including also estimates for moments of the solutions. These estimates imply in particular that the only relevant collisions are those between particles of comparable sizes.  
Finally, in Section \ref{sec:discr_cont} we prove that the stationary solutions of the discrete model \eqref{B2} behave as the solutions of the continuous model \eqref{eq:time_evol} for large cluster sizes.

\bigskip 

\section{Setting and main results}\label{sec:setting}

\subsection{Different types of stationary solutions for coagulation equations.}\label{sec:types_sol}

\bigskip

The stationary solutions to the discrete equation \eqref{B2} satisfy: 
\begin{equation}
0=\frac{1}{2}\sum_{\beta<\alpha}K_{\alpha-\beta,\beta}n_{\alpha-\beta}%
n_{\beta}-n_{\alpha}\sum_{\beta>0}K_{\alpha,\beta}n_{\beta}+s_\alpha \label{B2Stat}%
\end{equation}
where $\alpha\in\mathbb{N}$ and $s_\alpha$ is supported on a finite set of integers. Analogously, in the continuous case, the stationary solutions to \eqref{eq:time_evol} satisfy
\begin{equation}
0=\frac{1}{2}\int_{0}^{x}K\left(  x-y,y\right)  f\left(  x-y\right)
f\left(  y\right)  dy-\int_{0}^{\infty}K\left(  x,y\right)  f\left(
x\right)  f\left(  y\right)  dy+\eta\left(  x\right), \label{eq:contStat}%
\end{equation}
where the source term $\eta\left(  x\right)  $ is compactly supported in $[1, \infty)$.  Although we write the equation using a notation where $f$ and $\eta$ are given as functions, the equation can be extended in a natural manner to allow for measures.  The details of the construction are discussed in Sec.\ \ref{sec:existence} and the explicit weak formulation may be found in (\ref{eq:stationary_eq}).

We remark that  equation \eqref{B2Stat}  can be written  as 
\begin{equation}
J_\alpha\left( n\right) - J_{\alpha-1}\left( n\right) = \alpha s_\alpha \,, \qquad \text{ for } \alpha \geq 1\,, \label{fluxDStat}%
\end{equation}
where we define $J_0(n)=0$ and, for $\alpha \geq 1$, we set 
\begin{equation*}
J_\alpha(n) = \sum\limits_{\beta=1}^\alpha \sum\limits_{\gamma=\alpha-\beta+1}^\infty K(\beta,\gamma) \beta n_\beta n_\gamma \,.
\end{equation*}
Notice that we will use indistintly the notation $ K_{\beta,\gamma}$ or $ K(\beta,\gamma)$.  
On the other hand, for sufficiently regular functions $f$ equation \eqref{eq:contStat} can similarly be written as
\begin{equation}
\partial_{x}J\left(  x;f\right)  =x\eta\left(  x\right)  \label{fluxStat}%
\end{equation}
where 
\begin{equation}\label{eq:flux}
J\left(  x;f\right)  =\int_{0}^{x}dy\int_{x-y}^{\infty}dz K(y,z)
y f\left(  y\right)  f\left(  z\right).
\end{equation}
This implies that the fluxes $J_\alpha(n)$ and $J(x;f)$ are constant for $\alpha$ and $x$ sufficiently large due to the fact that  $s$ is supported in a finite set and $\eta$ is compactly supported, and we prove in Lemma \ref{lem:flux} that this property continues to hold even when $f$ is a measure.
If $s_\alpha$ or $\eta(x)$ decay sufficiently fast for large values of $\alpha $ or $x$ then  $J_\alpha(n)$ or $J(x;f)$ converges to a positive constant as $\alpha$ or $x$ tend to infinity.

Given that other concepts of stationary solutions are found in the physics literature, we will call the  solutions of \eqref{B2Stat} and \eqref{eq:contStat} {\it stationary injection solutions}.
In this paper we will be mainly concerned with these solutions.  The physical meaning of these solutions, when they exist, is that  it is possible to  transport monomers 
towards  large clusters at the same rate at which  the monomers are added into the system.

For comparison,  let us also discuss briefly other concepts of stationary solutions and the relation with the stationary injection solutions.  One case often considered in the physics literature are {\it constant flux solutions} (cf. \cite{T96}). 
These are solutions of \eqref{eq:contStat} with $\eta \equiv 0$  satisfying 
\begin{equation}
J(x;f)=J_0\,, \quad \text{for } x>0\,,\label{eq:flux_J0}
\end{equation}
where $J_0 \in \R_+$ and $J(x;f)$ is defined in \eqref{eq:flux}. 
Explicit stationary solutions for coagulation equations have been obtained and discussed in \cite{DG92, DS95, DS92_95, L99, S83}. In these papers the collision kernel $K$ under consideration is not homogeneous. In the case of homogeneous kernels $K$ there is an explicit method to obtain power solutions of  \eqref{eq:contStat} by means of some transformations of the domain of integration that were introduced by Zakharov in order to study the solutions of the Weak Turbulence kinetic equations (cf.\ \cite{Zakharov1, Zakharov2}). Zakharov's method has been applied to coagulation equations in \cite{Connaughton}.

Alternatively, we can obtain power law solutions of \eqref{eq:flux_J0} using the homogeneity $\gamma$ of the kernel (cf.\ \eqref{eq:homogeneity_cond}).
Indeed, suppose that 
 $f\left(  x\right)  =c_{s}\left(  x\right)  ^{-\alpha}$ for some $c_{s}$ positive and $\alpha \in \R$. 
Using the homogeneity of the kernel $K$ we obtain
\begin{align*}
J(x;f)=G\left(
\alpha\right)  \left(  c_{s}\right)  ^{2}\left(  x\right)  ^{3+\gamma-2\alpha}
\end{align*}
under the assumption that%
\begin{equation}
G\left(  \alpha\right)  =\int_{0}^{1}dy\int_{1-y}^{\infty}dzK\left(
y,z\right)  \left(  y\right)  ^{1-\alpha}\left(  z\right)  ^{-\alpha}%
<\infty\,.\label{PowLawCond}%
\end{equation}
Using \eqref{eq:flux_J0}, we then obtain $\alpha = (3+\gamma)/2$ and  $c_{s}=\sqrt{\frac{J_{0}}{G\left(  \alpha\right)  }}.$ Therefore, \eqref{PowLawCond} holds if and only if $|\gamma+2\lambda| <1$. Notice that
(\ref{PowLawCond}) yields a necessary and sufficient condition to have a power
law solution of (\ref{eq:flux_J0}). 
However, one should not assume that all solutions of \eqref{eq:flux_J0} are given by a power law;
indeed, we have preliminary evidence that there exist smooth homogeneous kernels satisfying \eqref{eq:cond_kernel} for which there are non- power law solutions to \eqref{eq:flux_J0}.

Finally, let us mention one more type of solutions  associated with the discrete coagulation equation \eqref{B2Stat} that have some physical interest.   
This is the {\it boundary value problem\/} in which the
concentration of monomers is given and the coagulation equation \eqref{B2Stat} is satisfied for clusters containing two or more monomers ($\alpha  \geq 2$). The problem then becomes
\begin{eqnarray}\label{eq:bvp}
0&=&\frac{1}{2}\sum_{\beta<\alpha}K_{\alpha-\beta,\beta}n_{\alpha-\beta}%
n_{\beta}-n_{\alpha}\sum_{\beta>0}K_{\alpha,\beta}n_{\beta}\,, \quad\text{for
}\alpha\geq2\, \nonumber \\ 
 n_{1} &=& c_{1}
 \end{eqnarray}
 where $c_1 >0$ is given.

Notice that if we can solve the injection problem \eqref{B2Stat} for some source $s=s_1\delta_{\alpha,1} $ with $s_1 >0$, then we can solve the  boundary value problem \eqref{eq:bvp} for any $c_1>0$. Indeed, let us denote by $N_\alpha(s_1)$, $\alpha\in \N$,  the solution to \eqref{B2Stat} with source $s=s_1\delta_{\alpha,1} $. Then equation \eqref{B2Stat} for $\alpha=1$ reduces to 
$$  N_1(s_1) \sum_{\beta>0} K_{1\beta}N_\beta(s_1) = s_1\, .$$
This implies that $0<N_1(s_1) <\infty$.
Then the solution to \eqref{eq:bvp} is given by
 $$n_\alpha = c_1\frac{N_\alpha(s_1)}{N_1(s_1)}.$$
Moreover, if we can solve  \eqref{B2Stat} for some $s_1 > 0$, then we can solve \eqref{B2Stat} for arbitrary values of $s_1$ due to the fact that if $n$ is a solution of \eqref{B2Stat} with source $s$ then for any $\Lambda>0,$ $\sqrt{\Lambda}n$ solves \eqref{B2Stat} with source $\Lambda n.$ 

In this paper we will consider the problems \eqref{B2Stat} and \eqref{eq:contStat} in Sections \ref{sec:setting} to \ref{sec:estimates}. In Section \ref{sec:discr_cont} we prove that  a rescaled version of the solutions to \eqref{B2Stat} and \eqref{eq:contStat} behaves for large cluster sizes as a solution to  \eqref{eq:flux_J0}.   We will not discuss solutions to \eqref{eq:bvp} in this paper.

\medskip

{
In this paper we will study the solutions of \eqref{B2Stat}, \eqref{eq:contStat} for kernels
$K\left(  x,y\right)  $ which behave for large clusters as $x^{\gamma+\lambda}%
y^{-\lambda}+y^{\gamma+\lambda}x^{-\lambda}$ for suitable coefficients
$\gamma,\lambda\in\mathbb{R}$ in the case of the equation  \eqref{eq:contStat}  as well as
their discrete counterpart in the case of \eqref{B2Stat}. (See next Subsection for the
precise assumptions on the kernels, in particular \eqref{eq:cond_kernel2}, \eqref{eq:cond_kernel3}.) The main
result that we prove in this paper is that the equations \eqref{B2Stat},  \eqref{eq:contStat} with
nonvanishing source terms $s_{\alpha},\ \eta$, respectively, have a solution if $\left\vert
\gamma+2\lambda\right\vert <1$ and they do not have solutions at all if
$\left\vert \gamma+2\lambda\right\vert \geq1.$ The heuristic idea behind this
result is easy to grasp. We will describe it in the case of the equation
 \eqref{eq:contStat}, since the main ideas are similar for \eqref{B2Stat}. }

{The equation \eqref{eq:contStat} can be reformulated as \eqref{fluxStat}, \eqref{eq:flux}. Since $\eta$ is
compactly supported we obtain that $J\left(  x;f\right)  $ is a constant
$J_{0}>0$ for $x$ sufficiently large. The homogeneity of the kernel $K\left(
x,y\right)  =x^{\gamma+\lambda}y^{-\lambda}+y^{\gamma+\lambda}x^{-\lambda}$
suggests that the solutions of the equation $J\left(  x;f\right)  =J_{0}$
should behave as $f\left(  x\right)  \approx Cx^{-\frac{\gamma+3}{2}}$ for
large $x,\ $with $C>0.$ Actually this statement holds in a suitable sense that
will be made precise later. However, this asymptotic behaviour for $f\left(
x\right)  $ cannot take place if $\left\vert \gamma+2\lambda\right\vert \geq1$
because the integral in \eqref{eq:flux} would be divergent. Therefore, solutions to
 \eqref{eq:contStat} can only exist for $\left\vert \gamma+2\lambda\right\vert <1.$ }

\subsection{Definition of solution and main results}

We restrict our analysis to the kernels satisfying \eqref{eq:cond_kernel}, or at least one of the inequalities there.  To account for all the relevant cases, let us summarize the assumptions on the kernel slightly differently here.  We always assume that 
\begin{equation}
K:\mathbb{R}_*\times \mathbb{R}_*\rightarrow \mathbb{R}_{+}\, ,\quad K%
\text{ is continuous} \,, \label{ContinAssumpt}
\end{equation}
and
for all $x,y$, 
\begin{equation}
K(x,y)\geq 0\, ,\qquad K(x,y)=K(y,x)\, .  \label{eq:cond_kernel1}
\end{equation}%
We also only consider kernels for which one may find $\gamma,\lambda\in \R$ such that at least one of the following holds:
there is $c_1>0$ such that for all $(x,y)\in \mathbb{R}_*^{2}$
\begin{equation}
K\left( x,y\right) \geq c_{1}\left( x^{\gamma +\lambda }y^{-\lambda
}+y^{\gamma +\lambda }x^{-\lambda }\right)\,,  \label{eq:cond_kernel2}
\end{equation}%
and/or there is $c_2>0$ such that for all $(x,y)\in \mathbb{R}_*^{2}$
\begin{equation}
K\left( x,y\right) \leq c_{2}\left( x^{\gamma +\lambda }y^{-\lambda
}+y^{\gamma +\lambda }x^{-\lambda }\right) \,. \label{eq:cond_kernel3}
\end{equation}%

The class of kernels satisfying all of the above assumptions includes many of the most commonly encountered coagulation kernels. It includes in particular the Smoluchowski (or
Brownian) kernel (cf.\ \eqref{KernBrow})
and the free molecular kernel (cf.\ \eqref{eq:BalKer}).

The source rate is assumed to be given by $\eta \in \mathcal{M}_{+}\left( \mathbb{R}%
_*\right) $ and to satisfy 
\begin{equation}
\supp\left( \eta \right) \subset \left[ 1,L_{\eta }\right] \text{ for some }%
L_{\eta }\geq 1\, .  \label{eq:cond_eta}
\end{equation}
Note that then always $\eta \left( \mathbb{R}%
_*\right) <\infty$, i.e., the measure $\eta$ is bounded.

We study the existence of stationary injection solutions to equation~%
\eqref{eq:time_evol} in the following precise sense:
\begin{definition}
\label{DefFluxSol} Assume that $K:{\R}_*^{2}\rightarrow {\mathbb{R}%
}_{+}$ is a continuous function satisfying \eqref{eq:cond_kernel1} and the upper bound %
\eqref{eq:cond_kernel3}. Assume further that $\eta \in \mathcal{M}_{+}\left( \mathbb{%
R}_*\right) $ satisfies \eqref{eq:cond_eta}. We will say that $f\in 
\mathcal{M}_{+}\left( \mathbb{R}_*\right) ,$ satisfying $f\left( \left(
0,1\right) \right) =0$ and
\begin{equation}
\int_{\R_* }x^{\gamma +\lambda }f\left( dx\right) + \int_{\R_* }x^{-\lambda }f\left( dx\right) <\infty\,,
\label{eq:moment_cond}
\end{equation}
is a stationary injection solution of \eqref{eq:time_evol} if the following
identity holds for any test function 
$\varphi \in C_{c}({\R}_*)$: 
\begin{equation}
\frac{1}{2}\int_{\R_*}\int_{\R_*} K\left( x,y\right) \left[
\varphi \left( x+y\right) -\varphi \left( x\right) -\varphi \left( y\right) %
\right] f\left( dx\right) f\left( dy\right) +\int_{\R_*}\varphi
\left( x\right) \eta \left( dx\right) =0\,.  \label{eq:stationary_eq}
\end{equation}
\end{definition}

\begin{remark}\label{rem:defFluxSol}
Definition \ref{DefFluxSol}, or a discrete version of it, will be used throughout most of the paper (cf.\ Sections \ref{sec:setting} to \ref{sec:estimates}). 
In Section \ref{sec:discr_cont}, we will use a more general notion of a stationary injection solution to \eqref{eq:time_evol}, considering source terms  $\eta$ which satisfy $ \supp \eta \subset [a, b]$ for some given constants $a$ and $b$ such that $0<a<b$.  
Then we require that $f\in \mathcal{M}_{+}\left( \mathbb{R}_*\right)$ and $f((0, a))=0$, in addition to (\ref{eq:moment_cond}).  Note that for such measures we have $\int_{\R_*} f(dx) = \int_{[a,\infty)} f(dx)$.  The generalized case is straightforwardly reduced to the above setup by rescaling space via the change of variables $x'=x/a$.
\end{remark}

The condition $f\left( \left(0,1\right) \right) =0$ is a natural requirement for stationary solutions of the
coagulation equation, given that $\eta \left( \left( 0,1\right) \right) =0$.
As we show next, 
the second integrability condition (\ref{eq:moment_cond}) is the minimal one needed to have well defined integrals
in the coagulation operator.

First, note that all the integrals appearing in \eqref{eq:stationary_eq} are well
defined for any $\varphi \in C_{c}\left( \mathbb{R}_*\right)$ with $\supp%
\varphi \subset ( 0,L]$, because we can then restrict the domain of
integration to the set $\left\{ \left( x,y\right) \in \left[ 1,L\right]
\times \left[ 1,\infty \right)\right\} $ in the term containing $\varphi \left(
x\right)$, and to the set $\left\{ \left( x,y\right) \in \left[ 1,L\right]
^{2}\right\} $ in the term containing $\varphi \left( x+y\right)$. In addition, \eqref{eq:cond_kernel3} implies  that $K\left(
x,y\right) \leq \tilde{C}_{L}[y^{\gamma + \lambda}+y^{-\lambda}]$ for $\left( x,y\right) \in \left[ 1,L\right]
\times \left[ 1,\infty \right)$.
Therefore,
\begin{align*}
& \int_{\mathbb{R}_*}\int_{\mathbb{R}_*}K\left( x,y\right) \left\vert
\varphi \left( x+y\right) \right\vert f\left( dx\right) f\left( dy\right) 
\leq C_{L}\left( \int_{\left[ 1,L\right] }f\left( dx\right) \right) ^{\! 2}\,, \\
& \int_{\mathbb{R}_*}\int_{\mathbb{R}_*}K\left( x,y\right) \left\vert
\varphi \left( x\right) \right\vert f\left( dx\right) f\left( dy\right)
\leq C_{L}\left(
\int_{\mathbb{R}_*}y^{\gamma +\lambda }f\left( dy\right)+\int_{\mathbb{R}_*}y^{-\lambda }f\left( dy\right) \right)\int_{\left[ 1,L\right] }f\left( dx\right) \,,
\end{align*}%
where $C_{L}$ depends on $\varphi$, $\gamma$, and $\lambda$.  Then, the assumption~\eqref{eq:moment_cond} in the
Definition \ref{DefFluxSol} implies that all the integrals appearing in %
\eqref{eq:stationary_eq} are convergent.

We now state the main results of this paper.
\begin{theorem}
\label{thm:existence} Assume that $K$ satisfies~\eqref{ContinAssumpt}--%
\eqref{eq:cond_kernel3} and $| \gamma +2\lambda | <1.$ Let $\eta \neq 0  $ satisfy \eqref{eq:cond_eta}. Then, there exists a stationary injection solution $f\in 
\mathcal{M}_{+}\left( \mathbb{R}_*\right) $, $f\neq 0$, to~%
\eqref{eq:time_evol} in the sense of Definition~\ref{DefFluxSol}. 
\end{theorem}

\begin{theorem}
\label{thm:NonExistence} Suppose that $K\left( x,y\right) $ satisfies~%
\eqref{ContinAssumpt}--\eqref{eq:cond_kernel3} as well as $| \gamma +2\lambda | \geq 1.$ Let
us assume also that $\eta \neq 0$  satisfies \eqref{eq:cond_eta}. Then, there is no solution of~\eqref{eq:time_evol} in the sense of Definition \ref{DefFluxSol}.
\end{theorem}

\begin{remark}
Notice that the free molecular kernel defined as in \eqref{eq:BalKer} satisfies %
\eqref{eq:cond_kernel1}--\eqref{eq:cond_kernel3} with $\gamma =\frac{1}{6},\
\lambda =\frac{1}{2}$. Then, since $\gamma +2\lambda >1$, we are in the
Hypothesis of Theorem \ref{thm:NonExistence} which implies that there are
no solutions of \eqref{eq:time_evol} in the sense of the Definition \ref%
{DefFluxSol} for the kernel \eqref{eq:BalKer} and some $\eta\ne 0$. On the other hand, in the
case of the Brownian kernel defined in (\ref{KernBrow}) with $\gamma = 0 $ and $\lambda = \frac{1}{3}$ the assumptions of Theorem %
\ref{thm:existence} hold and nontrivial stationary injection solutions in the
sense of Definition \ref{DefFluxSol} exist for each $\eta$ satisfying
(\ref{eq:cond_eta}).
\end{remark}

\begin{remark}
We observe that if $\eta =0$, there is a trivial stationary solution to~%
\eqref{eq:time_evol} given by $f=0$. On the other hand, if $\eta \ne 0$, then $f=0$ cannot be a solution.
\end{remark}

\begin{remark} Assumption \eqref{eq:cond_eta} is motivated by specific problems in chemistry \cite{Vehkam} which have a source of monomers $s_{\alpha}=s\delta_{\alpha,1}$ only. 
However, in all the results of this paper this assumption could be replaced by the most general condition
\begin{equation}
\operatorname*{supp}\left(  \eta\right)  \subset\left[  1,\infty\right)
,\ \ \int_{\left[  1,\infty\right)  }x\eta\left(  dx\right)  <\infty
\label{GeneralAssumpt}%
\end{equation}
and in the discrete case, the analogous condition \eqref{eq:Dcond_s} could be replaced by 
$\sum_{\alpha=1}^{\infty}\alpha s_{\alpha}<\infty$. 
Indeed, it is easily seen  that the only property of the source term $\eta$ that  is used in the arguments of the proofs, both in the existence and
non-existence results, is that:
\begin{equation}
\frac{J}{2}\leq\int_{\left[  1,L_{\eta}\right]  }x\eta\left(  dx\right)  \leq
J\ \ ,\ \ \text{where }\int_{\left[  1,\infty\right)  }x\eta\left(  dx\right)
=J\in\left(  0,\infty\right)  \label{WeakCondA1}%
\end{equation}
for some $L_{\eta}$ sufficiently large, or an analogous condition in the
discrete case, which follows immediately from (\ref{GeneralAssumpt}).
Moreover, it seems feasible to extend the support of $\eta$ to all positive real numbers $\mathbb{R}_{\ast}$, by assuming suitable smallness conditions for $f$ and $\eta$ near the origin (for instance in the form of a
bounded moment) in order to avoid fluxes of volume of particles coming from $x=0$. This would lead us to consider issues different from the
main ones considered in this paper, therefore we decided to not further consider this case here.
\end{remark}

The flux of mass from small to large particles at the stationary state is computed in the next lemma for the above measure-valued solutions.  In comparison to (\ref{eq:flux}), then one needs to refine the definition by using a right-closed interval for the first integration and an open interval for the second integration, as stated in (\ref{eq:flux_lem}) below.
\begin{lemma}\label{lem:flux}
Suppose that the assumptions of Theorem~\ref{thm:existence} hold. Let $f$ be a stationary injection solution in the sense of Definition~\ref{DefFluxSol}. Then $f$ satisfies for any $R>0$ 
\begin{equation}\label{eq:flux_lem}
\int_{ (0,R]}\int_{(R-x,\infty)} K(x,y)x f(dx) f(dy) = \int_{(0,R]} x \eta(dx) \, .
\end{equation}
\end{lemma}
\begin{remark}
Note that if $R\geq L_\eta$, the right-hand side of \eqref{eq:flux_lem} is always equal to $J=\int_{[1,L_\eta]} x \eta(dx)>0$.  Therefore, the flux is constant in regions involving only large cluster sizes. 
\end{remark}
\begin{proof}  If $R<1$, both sides of (\ref{eq:flux_lem}) are zero, and the equality holds.
Consider then some $R\ge 1$ and for all $\varepsilon$ with $0<\varepsilon<R$
choose some $\chi_\varepsilon \in C_c^\infty(\R_*)$ such that $0\le \chi_\varepsilon\le 1$, 
$\chi_\varepsilon (x) = 1$, for $1\le x \leq R$, and $\chi_\varepsilon(x) = 0$, for $x \geq R+\varepsilon $.  Then for each $\varepsilon$ we may define
$\varphi(x) = x \chi_\varepsilon(x) $ and thus obtain a valid test function $\varphi\in C_{c}({\R}_*)$.  Since then \eqref{eq:stationary_eq} holds, we find that for all $\varepsilon$
\begin{equation}\label{eq:fluxChi}
\frac{1}{2}\int_{\R_* }\int_{\R_* }K\left( x,y\right) \left[(x+y) \chi_\varepsilon(x+y) -x \chi_\varepsilon(x) - y \chi_\varepsilon(y) \right] f\left( dx\right) f\left( dy\right) + \int_{\R_*}x\chi_\varepsilon(x)\eta(dx)=0\,. 
\end{equation}
The first term can be rewritten as follows
\begin{eqnarray*}
&& \frac{1}{2}\iint_{ \{(x,y)| x+y> R\}} K\left( x,y\right) \left[(x+y) \chi_\varepsilon(x+y) -x \chi_\varepsilon(x) - y \chi_\varepsilon(y) \right] f\left( dx\right) f\left( dy\right) \\
&=& \frac{1}{2}\iint_{ \{(x,y)| x+y> R,\ x\leq R,\ y\leq R\}} K\left( x,y\right) \left[(x+y) \chi_\varepsilon(x+y) -x  - y  \right] f\left( dx\right) f\left( dy\right)\\ 
&& +\frac{1}{2}\iint_{ \{(x,y)| x> R,\ y\leq R \}} K\left( x,y\right) \left[(x+y) \chi_\varepsilon(x+y) -x \chi_\varepsilon(x) - y  \right] f\left( dx\right) f\left( dy\right) \\
&& +\frac{1}{2}\iint_{ \{(x,y)| y> R,\ x \leq R \}} K\left( x,y\right) \left[(x+y) \chi_\varepsilon(x+y) -x - y \chi_\varepsilon(y) \right] f\left( dx\right) f\left( dy\right) \\
&&+\frac{1}{2}\iint_{ \{(x,y)| y> R,\ x>R \}} K\left( x,y\right) \left[-x \chi_\varepsilon(x) - y \chi_\varepsilon(y) \right] f\left( dx\right) f\left( dy\right). 
\end{eqnarray*}
We readily see that the terms involving $\chi_\varepsilon$ on the right hand side 
tend to zero as $\varepsilon$ tends to zero due to the fact that for Radon measures $\mu$ the integrals $\int_{[a-\varepsilon, a)}d\mu $ and $\int_{(a,a+\varepsilon]} d\mu $ converge to $0$ as $\varepsilon $ tends to zero.
Then we obtain  from \eqref{eq:fluxChi}
\begin{eqnarray*}
&   \frac{1}{2}\iint\limits_{ \{(x,y)| x+y> R,\ x \leq R,\ y \leq R \}} K\left( x,y\right) (x+y)f\left( dx\right) f\left( dy\right)
+ \frac{1}{2}\iint\limits_{ \{(x,y)| x> R,\ y \leq R \}} K\left( x,y\right)  y f\left( dx\right) f\left( dy\right)\\
&  +\frac{1}{2}\iint\limits_{ \{(x,y)| y> R,\ x \leq R \}} K\left( x,y\right)  x f\left( dx\right) f\left( dy\right)  =  \int_{(0,R]}x\eta(dx)\,.  
\end{eqnarray*}
Rearranging the terms we obtain
\begin{align*} 
  & \frac{1}{2}\iint\limits_{ \{(x,y)| x+y> R,\ x \leq R \}} K\left( x,y\right) x f\left( dx\right) f\left( dy\right)
+   \frac{1}{2}\iint\limits_{ \{(x,y)| x+y> R,\ y \leq R \}} K\left( x,y\right) y f\left( dx\right) f\left( dy\right) 
\\ &\quad 
=\int_{(0,R]}x\eta(dx)\,
\end{align*}
which implies (\ref{eq:flux_lem}) using a symmetrization argument.
\end{proof}

{
The following Lemma will be used several times throughout the paper to convert bounds for certain ``running averages'' into uniform bounds of integrals. The function $\varphi$  below is included mainly for later convenience.

\begin{lemma}\label{lem:bound}
Suppose $a>0$ and $b \in (0,1)$, and assume that 
$R\in (0,\infty]$ is such that 
$R\ge a$.  Consider some $f \in \mathcal{M_+}(\R_*)$ and $\varphi \in C(\R_*)$,
with $\varphi\ge 0$.
\begin{enumerate}
 \item\label{it:Rlessinf} Suppose $R<\infty$, and assume that there is 
 $g \in L^1([a,R])$ such that $g\ge 0$ and 
\begin{equation}\label{eq:avg}
\frac{1}{z}\int_{[bz,z]} \varphi(x) f(dx) \leq g(z)\,, \quad \text{for } z \in [a,R] \,.
\end{equation}
Then
\begin{equation}\label{eq:bound}
\int_{[a,R]} \varphi(x) f(dx) \leq \frac{\int_{[a,R]}g(z)dz}{\vert \ln b\,\vert }
+ R g(R)\,.
\end{equation}
 \item\label{it:polcase}
 Consider some $r\in (0,1)$, and assume that $a/r\le R<\infty$.
 Suppose that 
(\ref{eq:avg}) holds for $g(z)=c_0 z^q$, with $q\in \R$ and $c_0\ge 0$.  Then there is a constant $C>0$, which depends only on $r$, $b$ and $q$, such that
\begin{equation}\label{eq:polbound}
\int_{[a,R]} \varphi(x) f(dx) \leq C c_0 \int_{[a,R]} z^q dz \,.
\end{equation}
\item\label{it:Risinf} If $R=\infty$ and there is 
 $g \in L^1([a,\infty))$ such that $g\ge 0$ and 
\begin{equation}\label{eq:avgRinfty}
\frac{1}{z}\int_{[bz,z]} \varphi(x) f(dx) \leq g(z)\,, \quad \text{for } z \ge a \,,
\end{equation}
then
\begin{equation}\label{eq:boundRinfty}
\int_{[a,\infty)} \varphi(x) f(dx) \leq \frac{\int_{[a,\infty)}g(z)dz}{\vert \ln b\,\vert }\,.
\end{equation}
\end{enumerate}
\end{lemma}
\begin{proof}
We first prove the general case in item \ref{it:Rlessinf}. Assume thus that $R<\infty$ and that $g\ge 0$ is such that (\ref{eq:avg}) holds.
We recall that then $0<a\le R$.  If $a\ge b R$, 
we can estimate
\begin{equation*}
\int_{[a,R]} \varphi(x) f(dx) \leq
\int_{[bR,R]} \varphi(x) f(dx)\leq Rg(R)\,,
\end{equation*}
using the assumption (\ref{eq:avg}) with $z=R$.  Thus (\ref{eq:bound}) holds in this case since $g\ge 0$.

Otherwise, we have $0<a<b R$.
By assumption, the constant $C_1 := \int_{[a,R]} g(z) dz\geq 0$ is finite.
Integrating~\eqref{eq:avg} over $z$ from $a$ to $R$, we obtain
$$\int_{[a,R]}\int_{[bz,z]}\frac{1}{z} \varphi(x)f(dx) dz\leq C_1\, .$$
The iterated integral satisfies the assumptions of Fubini's theorem, and thus 
it can be written as an integral over the set 
\begin{eqnarray}
&& \{(z,x)\ |\ a \leq z\leq R, \ bz \leq x \leq z \}\nonumber \\
&=& \{(z,x)\ |\ ba\leq x\leq R, \ \max\{a,x\} \leq z \leq \min\{\frac{1}{b}x,R\} \}\nonumber \\
&\supset &  \{(z,x)\ |\ a\leq x\leq bR, \ x \leq z \leq \frac{1}{b}x \}\,.\nonumber
\end{eqnarray}
Therefore, after using Fubini's theorem to obtain an integral where $z$-integration comes first, we find that
$$ \int_{[a,bR]} \int_{[x, x/b]}\frac{1}{z}  \varphi(x) f(dx ) dz \ \leq\ C_1.$$
The integral over $z$ yields $\ln(x/(bx))=\vert \ln b\,\vert $, and thus $\int_{[a,bR]}  \varphi(x) f(dx ) \ \leq\ C_1/ \vert \ln b\,\vert$.  To get an estimate for the integral over $[bR,R]$,
we use (\ref{eq:avg}) for $z=R$.
Hence, \eqref{eq:bound} follows also in this case which completes the proof of the first item.

For item \ref{it:polcase}, let us assume that $0<r<1$, $a\le r R<\infty$, and that
(\ref{eq:avg}) holds for $g(z)=c_0 z^q$, with $q\in \R$ and $c_0\ge 0$.  Since then $g\in L^1([a,R])$ and $g\ge 0$, we can conclude from the first item that
that \eqref{eq:bound} holds.  Thus we only need find a suitable bound for the second 
term therein, for $R g(R)=c_0 R^{q+1}$.  By changing the integration variable from $z$ to $y=z/R$,
we find 
\[
 \int_{[a,R]} z^q dz= R^{q+1}\int_{[a/R,1]} y^q dy
 \ge R^{q+1}\int_{[r,1]} y^q dy\,.
\]
Here, $C':=\int_{[r,1]} y^q dy$ satisfies $0<C'<\infty$
for any choice of $q\in \R$.  Therefore,
we can now conclude that (\ref{eq:polbound}) holds for $C=|\ln b\,|^{-1} + 1/C'$ which depends only on $q$, $b$, and $r$.

For item \ref{it:Risinf}, let us suppose that $R=\infty$
and consider an arbitrary finite $R'>a/b$.  Then item \ref{it:Rlessinf}
holds for $R'$.  Moreover, we can employ the estimate derived in the proof of the item \ref{it:Rlessinf} above, and conclude that
$\int_{[a,b R']}  \varphi(x) f(dx ) \leq C_1/ \vert \ln b\,\vert$
where $C_1=\int_{[a,R']} g(z) dz\le \int_{[a,\infty)} g(z) dz$.  Taking $R'\to \infty$ and using the monotone convergence theorem, proves that (\ref{eq:boundRinfty}) holds.  This completes the proof of the Lemma.
\end{proof}
}

\section{Existence results: Continuous model}\label{sec:existence}

Our first goal is to prove the existence of a stationary injection solution (cf.\ Theorem~%
\ref{thm:existence}) under the assumption $|\gamma + 2\lambda|<1$. 
This will be accomplished in three steps: We first prove in Proposition~\ref%
{thm:existence_evolution} existence and uniqueness of time-dependent
solutions for a particular class of compactly supported continuous kernels.
Considering  these solutions at large times allows us to prove in Proposition \ref%
{thm:existence_truncated} existence of stationary injection solutions for this
class of kernels using a fixed point argument.  We then 
extend the existence result to general unbounded kernels supported in ${%
\mathbb{R}}_*^{2}$ and satisfying \eqref{eq:cond_kernel1}--%
\eqref{eq:cond_kernel3} with $| \gamma +2\lambda | <1$.

Compactly supported continuous kernels are automatically bounded from above but, 
for the first two results, we will also assume that the kernel has a uniform lower bound
on the support of the source.  To pass to the limit 
including the more general kernel functions, it will be necessary to control the 
dependence of the solutions on both of the bounds and on the size of support of the kernel.
To fix the notations, let us first choose an upper bound $L_\eta$ for the support of the source,
i.e., a constant satisfying \eqref{eq:cond_eta}.  In the first two Propositions, we will consider
kernel functions which are continuous, non-negative, have a compact support, and for which we may find 
$R_{\ast }\geq L_{\eta }$ and $a_1$, $a_2$ such that  $0<a_{1}<a_{2}$ and $%
K(x,y)\in \lbrack a_{1},a_{2}],$ for $(x,y)\in \lbrack 1,2R_{\ast }]^{2}$.
This allows us to prove first that the time-evolution is well-defined, Proposition \ref%
{thm:existence_evolution}, and then in Proposition \ref%
{thm:existence_truncated} the existence of stationary injection solutions for this
class of kernels using a fixed point argument.  The proofs include sufficient control of the dependence 
of the solutions on the cut-off parameters to remove 
the restrictions and obtain the result in Theorem \ref{thm:existence}.

In fact, not only we regularize the kernel, but we also introduce a cut-off for the coagulation 
gain term.  This will guarantee that the equation is well-posed and has solutions 
whose support never extends beyond the interval $[1,2R_{\ast }]$.
To this end, let us choose $\zeta _{R_{\ast }}\in C\left( \mathbb{R}_*\right)$ 
such that $0\le \zeta _{R_{\ast }}\le 1$, 
$\zeta_{R_{\ast }}\left( x\right) =1$ for $0\leq x\leq R_{\ast }$, and
$\zeta _{R_{\ast}}\left( x\right) =0$ for $x\geq 2R_{\ast }$.
We then regularize the time evolution equation \eqref{eq:time_evol} as 
\begin{equation}
\partial _{t}f(x,t)=\frac{\zeta _{R_{\ast }\!}(x) }{2}\int_{\left(
0,x\right] }K( x-y,y) f( x-y,t) f( y,t)
dy - \int_{{{\R}_*}}\! K( x,y) f( x,t) f(
y,t) dy+\eta( x) \,. \label{evolEqTrunc}
\end{equation}%
As we show later, this will result in a well-posedness theory 
such that any solution of \eqref{evolEqTrunc} has the following property: 
$f\left( \cdot ,t\right) $ is supported on
the interval $\left[ 1,2R_{\ast }\right] $ for each $t\geq 0$. 
Let us also point out that  since we are interested in solutions $f$ such that $f\left( %
\left( 0,1\right) ,t\right) =0$, the above integral $\int_{\left( 0,x%
\right] }\left( \cdot \cdot \cdot \right) $ can be replaced by $\int_{\left[
1,x-1\right] }\left( \cdot \cdot \cdot \right) $ if $x \geq 1$.

\begin{assumption}\label{Assumptions} 
Consider a fixed source term $\eta \in 
\mathcal{M}_{+}\left( {\R}_*\right)$ and assume that $L_\eta\ge 1$ satisfies \eqref{eq:cond_eta}.
Suppose $R_{\ast }$, $a_1$, $a_2$, and $T$ are constants for which 
$R_{\ast }>L_\eta$, $0<a_{1}<a_{2}$, and $T>0$.  Suppose $K:{%
\mathbb{R}}_*^{2}\rightarrow {\mathbb{R}}_{+}$ is a continuous, non-negative, symmetric function
such that $K(x,y)\le a_2$ for all $x,y$, and we also have 
$K(x,y)\in \lbrack a_{1},a_{2}]$ for $(x,y)\in \lbrack
1,2R_{\ast }]^{2}$, and
$K(x,y)=0$, if $x\geq 4R_{\ast}$ or $y\geq 4R_{\ast}$. 
Moreover, we assume that there is given a function $\zeta _{R_{\ast }}$ such that $\zeta _{R_{\ast }}\in C\left( \mathbb{R}_*\right)$, $0\le \zeta _{R_{\ast }}\le 1$,  $\zeta
_{R_{\ast }}\left( x\right) =1$ for $0\leq x\leq R_{\ast }$, and $\zeta _{R_{\ast
}}\left( x\right) =0$ for $x\geq 2R_{\ast }$.
\end{assumption}

We will now study measure-valued solutions of the regularized problem  
\eqref{evolEqTrunc} in an integrated form.  To this end, we use a fairly strong notion
of continuous differentiability although uniqueness of the regularized problem might hold in a larger class.  
However, since we cannot prove uniqueness after the regularization has been removed,
it is not a central issue here.

\begin{definition}\label{defC1space}
 Suppose $Y$ is a normed space, $S\subset Y$, and $T>0$.  We use the notation $C^1([0,T],S;Y)$ for the collection 
 of maps $f:\left[ 0,T\right]\to S$ such that $f$ is continuous 
 and there is 
 $\dot{f}\in C([0,T],Y)$
 for which the Fr\'echet derivative of $f$ at any $t\in \left( 0,T\right)$ is given by $\dot{f}(t)$.
 
 We also drop the normed space $Y$ from the notation if it is obvious from the context, in particular, 
 if $S=\mathcal{M}_{+,b}(I)$ and $Y=C_0(I)^*$ or $Y=S$.
\end{definition}
Clearly, if $f\in C^1([0,T],S;Y)$, the function $\dot{f}$ above is unique and it can be found by requiring
that for all $t\in \left( 0,T\right)$
\[
 \lim_{\vep\to 0} \frac{\norm{f(t+\vep)-f(t)-\vep \dot{f}(t)}_Y}{|\vep|} = 0\,,
\]
and then taking the left and right limits to obtain the values $\dot{f}(0)$ and $\dot{f}(T)$. 
What is sometimes relaxed in similar notations is the existence of the left and right limits.

\begin{definition}\label{DefTimeSol}  Suppose that Assumption \ref{Assumptions} holds. 
Consider some initial data $f_{0}\in \mathcal{M}_{+}({\mathbb{R}}_*)$ for which $f_{0}\left( \left( 0,1\right) \cup \left( 2R_{\ast
},\infty \right) \right) =0$. Then $f_0\in \mathcal{M}_{+,b}({\mathbb{R}}_*)$.

We will say that $f\in {C^{1}(\left[ 0,T%
\right] ,\mathcal{M}_{+,b}({\mathbb{R}}_*))}$ satisfying $f\left( \cdot
,0\right) =f_{0}\left( \cdot \right) $ is a time-dependent solution of (\ref{evolEqTrunc}) if the following
identity holds for any test function $\varphi \in C^{1}(\left[ 0,T\right]
,C_{c}\left( {\mathbb{R}}_*\right) )$ and all $0<t<T$,
\begin{eqnarray}
&&\frac{d}{dt}\int_{{{\mathbb{R}}_*}}\varphi \left( x,t\right) f\left(
dx,t\right) -\int_{{{\mathbb{R}}_*}}\dot{\varphi}\left( x,t\right) f\left(
dx,t\right)  \notag \\
&=&\frac{1}{2}\int_{{{\mathbb{R}}_*}}\int_{{{\mathbb{R}}_*}}K\left(
x,y\right) \left[ \varphi \left( x+y,t\right) \zeta _{R_{\ast }}\left(
x+y\right) -\varphi \left( x,t\right) -\varphi \left( y,t\right) \right]
f\left( dx,t\right) f\left( dy,t\right) \hspace{1cm}  \notag \\
&&+\int_{{{\mathbb{R}}_*}}\varphi \left( x,t\right) \eta \left( dx\right) .
\label{eq:evol_eqWeak}
\end{eqnarray}
\end{definition}

\begin{remark}
Note that for any such solution $f$, automatically by continuity and compactness of $[0,T]$ one has
\begin{equation}
\sup\limits_{t\in \lbrack 0,T]}\left( \int_{{{\mathbb{R}}_*}}f\left(
dx,t\right) \right) <\infty\,,    \label{eq:moment_cond_evol}
\end{equation}%
since $\norm{f}=f(\R_*)$.  Let us also point out that whenever $\varphi \in C^{1}(\left[ 0,T\right],C_{c}\left( {\mathbb{R}}_*\right) )$ and 
$f\in {C^{1}(\left[ 0,T\right] ,\mathcal{M}_{+,b}({\mathbb{R}}_*))}$,  
the map $t\mapsto \int_{{{\mathbb{R}}_*}}\varphi \left( x,t\right) f\left(dx,t\right)$ 
indeed belongs to $C^{1}(\left[ 0,T\right] ,\mathbb{R}_{*})$. 
Thus the derivative on the left hand side of \eqref{eq:evol_eqWeak} is defined in the usual sense and, in fact, it is equal to $\int_{{{\mathbb{R}}_*}}\varphi \left( x,t\right) \dot{f}\left(dx,t\right)$.
In addition, there is sufficient regularity that after
integrating (\ref{eq:evol_eqWeak}) over the interval $\left[ 0,t%
\right] $ we obtain%
\begin{eqnarray}
&&\int_{{{\mathbb{R}}_*}}\varphi \left( x,t\right) f\left( dx,t\right)
-\int_{{{\mathbb{R}}_*}}\varphi \left( x,0\right) f_{0}\left( dx\right) 
\notag \\
&=&\int_{0}^{t}ds\int_{{{\mathbb{R}}_*}}\dot{\varphi}\left( x,s\right)
f\left( dx,s\right)  \notag \\
&&+\frac{1}{2}\int_{0}^{t}ds\int_{{{\mathbb{R}}_*}}\int_{{{\mathbb{R}}_*}%
}K\left( x,y\right) \left[ \varphi \left( x+y,s\right) \zeta _{R_{\ast
}}\left( x+y\right) -\varphi \left( x,s\right) -\varphi \left( y,s\right) %
\right] f\left( dx,s\right) f\left( dy,s\right)  \notag \\
&&+\int_{0}^{t}ds\int_{{{\mathbb{R}}_*}}\varphi \left( x,s\right) \eta
\left( dx\right) .  \label{eq:EvolWeakInt}
\end{eqnarray}
\end{remark}

We can define also weak stationary solutions of (\ref{evolEqTrunc}).  It is
straightforward to check that if $F(dx)$ is a stationary solution, then  $f(d x ,t)=F(dx)$ is a solution to \eqref{eq:EvolWeakInt} with initial condition $f_0(dx)=F(dx)$.

\begin{definition}
\label{DefStatSol} Suppose that Assumption \ref{Assumptions} holds.  
We will
say that $f\in {\mathcal{M}_{+}({\mathbb{R}}_*)},$ satisfying $f((0,1) \cup (2R_*,\infty))=0$  
is a stationary injection solution of (\ref{evolEqTrunc}) if the following
identity holds for any test function $\varphi \in C_{c}\left( {\mathbb{R}}%
_*\right) $: 
\begin{eqnarray}
&&  \notag \\
0 &=&\frac{1}{2}\int_{{\mathbb{R}}_*}\int_{{\mathbb{R}}_*}K\left(
x,y\right) \left[ \varphi \left( x+y\right) \zeta _{R_{\ast }}\left(
x+y\right) -\varphi \left( x\right) -\varphi \left( y\right) \right] f\left(
dx\right) f\left( dy\right)  \notag \\
&&+\int_{{\mathbb{R}}_*}\varphi \left( x\right) \eta \left( dx\right) .
\label{eq:evol_eqWeakSt}
\end{eqnarray}
\end{definition}

\begin{proposition}
\label{thm:existence_evolution} 
Suppose that Assumption \ref{Assumptions} holds. Then, for any initial
condition $f_{0}$ satisfying $f_{0}\in \mathcal{M}_{+}({\mathbb{R}}_*)$, $%
f_{0}\left( \left( 0,1\right) \cup \left( 2R_{\ast },\infty \right) \right)
=0$ there exists a unique time-dependent solution 
$f\in {C^{1}(\left[ 0,T\right] ,\mathcal{M}_{+,b}({\mathbb{R}}_*))}$
to (\ref{evolEqTrunc}) which solves it in the classical
sense. Moreover, we have%
\begin{equation}
f\left( \left( 0,1\right) \cup \left( 2R_{\ast },\infty \right) ,t\right)
=0\, ,\quad \text{for }0\leq t \le T\,, \label{fVanish}
\end{equation}%
and the following estimate holds
\begin{equation}
\int_{{\mathbb{R}}_*}f(dx,t)\leq  \int_{{\mathbb{R}}_*}f_0(dx)+ C t \, ,\quad  0 \leq t \leq T   \,, \label{fEstim}
\end{equation}%
for $C = \int_{{\mathbb{R}}_*}\eta(dx)\ge 0$ which is independent of $f_{0}$, $t$, and $T$.
\end{proposition}

\begin{remark}
We remark that the lower estimate $K(x,y)\geq a_{1}>0$ will not be used in
the proof of Proposition \ref{thm:existence_evolution}. However, this
assumption will be used later in the proof of the existence of stationary injection
solutions.
\end{remark}

\begin{proof}
In this proof we skip some standard computations which may be found in \cite[Section 5]{Vuoksenmaa20}. 
 We define $\mathcal{X}_{R_{\ast }}=\left\{ f\in \mathcal{M}_{+}({\mathbb{R}}%
_*):f\left( \left( 0,1\right) \cup \left( 2R_{\ast },\infty \right)
\right) =0\right\}$.  Since $[1,2R_{\ast }]$ is compact, for any $f\in \mathcal{X}_{R_{\ast }}$
we have $f(\R_*)<\infty$, and thus $\mathcal{X}_{R_{\ast }}\subset \mathcal{M}_{+,b}(\mathbb{R}_*)$.
For $f\in \mathcal{M}_{+,b}(\mathbb{R}_*)$, we clearly have $f\in \mathcal{X}_{R_{\ast }}$
if and only if $\int \varphi(x) f(d x)=0$ for all $\varphi\in C_0(\R_*)$ whose support lies in $\left( 0,1\right) \cup \left( 2R_{\ast },\infty \right)$.  Therefore, $\mathcal{X}_{R_{\ast }}$ is a closed
subset both in the $\ast -$weak and norm topology of $C_0(\R_*)^*=\mathcal{M}_{b}({\mathbb{R}}_*)$.  

For the rest of this proof, we endow $\mathcal{X}_{R_{\ast }}$ with the norm topology which makes it into a complete metric space.
We look for solutions $f$ in the subset $X:=C( \left[ 0,T\right], \mathcal{X}_{R_{\ast }})$ of the
Banach space $C\left( \left[ 0,T\right], \mathcal{M}_{b}({\mathbb{R}}_*)\right)$.  
{The space $X$ is
endowed with  the norm }
\begin{equation*}
\left\Vert f\right\Vert _{T}=\sup_{0\leq t\leq T}\left\Vert f\left( \cdot
,t\right) \right\Vert\,.
\end{equation*}%
By the uniform limit theorem, also $X$ is then a complete metric space.

We now reformulate (\ref{evolEqTrunc}) as the following integral equation acting on $\mathcal{X}_{R_{\ast }}$: we define for $0\le t\le T$, $x\in \R_*$, and $f\in X$ first a function
\begin{equation}
a\left[ f\right] \left( x,t\right) =\int_{{\mathbb{R}}_*}K\left(
x,y\right) {f\left( dy,t\right)} \,, \label{aFunct}
\end{equation}
and using this we obtain a measure, written for convenience using the function notation,
\begin{align}
&\mathcal{T}\left[ f\right] \left( x,t\right) := f_{0}\left( x\right) e^{-\int_{0}^{t}a\left[ f\right] \left(x,
s\right) ds} 
+\eta \left( x\right) \int_{0}^{t}e^{-\int_{s}^{t}a\left[ f\right] \left(
x,\xi \right) d\xi }ds
\nonumber \\ & \qquad
+\frac{\zeta _{R_{\ast }}(x) }{2}\int_{0}^{t}e^{-\int_{s}^{t}a\left[ f\right] \left( x,\xi \right) d\xi
}\int_{0}^{x}K\left( x-y,y\right) f\left( x-y,s\right) f\left( y,s\right)
dy ds  \,.
  \label{fIntForm} 
\end{align}%
Notice that the definition \eqref{aFunct} indeed is pointwise well defined and yields a function 
$(x,s)\mapsto a\left[ f\right] \left( x,s\right) $ which is  continuous and non-negative for any $f\in X$. Moreover, we claim that if $f\in X$, then 
(\ref{fIntForm}) defines a measure in $\mathcal{M}_{+}({\mathbb{R}}%
_*)$ for each $t\in [0,T]$, and we have in addition $\mathcal{T}\left[ f\right] \in X$. The only
non-obvious term is the term on the right-hand side containing $\int_{0}^{x}K\left( x-y,y\right) f\left(
x-y,s\right) f\left( y,s\right) dy$. We first explain how this term defines a continuous linear functional on $C_c\left( {\mathbb{R}}_*\right)$. 
Define $g(x,s)=\frac{\zeta _{R_{\ast }}(x) }{2}e^{-\int_{s}^{t}a\left[ f\right] \left( x,\xi \right) d\xi
}$ which is a jointly continuous function with $g(x,s)=0$ if $x\ge 2 R_{\ast }$.
Given $\varphi \in C_{c}\left( {\mathbb{R}}_*\right) $ we then set 
\begin{align}
& \left\langle \varphi,\int_{0}^{t} g(\cdot,s)\int_{0}^{\cdot}K\left( \cdot-y,y\right) f\left( \cdot-y,s\right) f\left(
y,s\right) dy ds \right\rangle \nonumber\\
&\quad = \int_{0}^{t} \int_{{\mathbb{R}}_*} \left[ \int_{{\mathbb{R}}_*}K\left(
x,y\right) g(x+y,s) \varphi \left( x+y\right) f\left( dx,s\right)\right]f\left( dy,s\right)ds \, . \label{ConvDef}
\end{align}%
Here the right-hand side of (\ref{ConvDef}) is well defined since $%
f\left( \cdot ,s\right) \in \mathcal{X}_{R_{\ast }}$ for each $s\in \left[
0,t\right] .$ Moreover, this operator defines a continuous linear functional
from $C_c\left( {\mathbb{R}}_*\right) $ to $\mathbb{R}$, and thus is associated with a unique positive Radon measure.
Finally, if $\varphi(x)=0$ for $1\le x\le 2 R_{\ast }$, then $g(x+y,s) \varphi \left( x+y\right)=0$ for $x+y<1$, 
which implies that the right hand side of \eqref{ConvDef} is zero.  Therefore, the measure belongs to $\mathcal{X}_{R_{\ast }}$
for all $t$.  Continuity in $t$ follows straightforwardly.

The operator $\mathcal{T}\left[ \cdot \right] $ defined in (\ref{fIntForm})
is thus a mapping from $C([0,T],\mathcal{X}_{R_{\ast }})$ to $C([0,T],\mathcal{X}%
_{R_{\ast }})$ for each $T>0.$  
We now claim that it is a
contractive mapping from the complete metric space
\begin{equation*}
{X}_{T} :=\left\{ f\in X \, |\, \left\Vert f-f_{0}\right\Vert _{T}\leq 1\right\}
\end{equation*}%
to itself if $T$ is sufficiently small. This follows by means of
standard computations using the assumption
$K\left( x,y\right) \leq a_{2}$, as
well as the inequality $\left\vert e^{-x_{1}}-e^{-x_{2}}\right\vert \leq
\left\vert x_{1}-x_{2}\right\vert $ valid for $x_{1}\geq 0,\ x_{2}\geq 0.$ 

Therefore, there exists a unique solution of $f=\mathcal{T}\left[ f\right] $  in ${X}_{T}$
assuming that $T$ is sufficiently small. Notice that $f\geq 0$ by
construction.

In order to show that the obtained solution can be extended to arbitrarily
long times we first notice that if $f=\mathcal{T}\left[ f\right] $, then $f\in
C^{1}(\left[ 0,T\right] ,\mathcal{X}_{R_{\ast }})$ and the definition in 
(\ref{fIntForm}) implies that $f$ satisfies (\ref{evolEqTrunc}). 
Integrating this equation with respect to the $x$ variable,
we obtain the following estimates:
\begin{align}
& \partial _{t}\left( \int_{{\mathbb{R}}_* }f\left( dx,t\right) \right) 
\nonumber \\ & \quad
\leq \frac{1}{2}\int_{{\mathbb{R}}_* }f\left( dy,t\right) \int_{{%
\mathbb{R}}_*}K\left( x,y\right) f\left( dx,t\right) -\int_{{\mathbb{R}}_* }f\left( dy,t\right) \int_{{\mathbb{R}}_*
}K\left( x,y\right) f\left( dx,t\right) + \int_{{\mathbb{R}}_*}\eta \left(
dx\right) 
\nonumber \\ & \quad
= -\frac{1}{2}\int_{{\mathbb{R}}_* }f\left( dy,t\right) \int_{{%
\mathbb{R}}_*}K\left( x,y\right) f\left( dx,t\right) +\int_{{\mathbb{R}}%
_* }\eta \left( dx\right) 
\nonumber \\ & \quad
\leq \int_{{\mathbb{R}}_*}\eta \left( dx\right)
\label{eq:zeromomevolreg}
\end{align}
whence (\ref{fEstim}) follows. We can then extend the solution to
arbitrarily long times $T>0$ using standard arguments.
After this, the uniqueness of the solution in $C^{1}(\left[ 0,T\right] ,\mathcal{M}_{+,b}({\mathbb{R}}_*))$
  follows by a standard Gr\"onwall estimate.
\end{proof}

\begin{remark}\label{th:boundforzeromomreg}
Notice that using the inequality $K(x,y)\geq a_{1}>0$ we can strengthen (\ref{eq:zeromomevolreg})
into the estimate%
\begin{equation*}
\partial _{t}\left( \int_{{\mathbb{R}}_*}f\left( dx,t\right) \right) \leq -%
\frac{a_{1}}{2}\left( \int_{{\mathbb{R}}_* }f\left( dx,t\right) \right)
^{2}+\int_{{\mathbb{R}}_*}\eta \left( dx\right). 
\end{equation*}%
Inspecting the sign of the right hand side this implies an estimate stronger than (\ref{fEstim}), namely,%
\begin{equation*}
\int_{{\mathbb{R}}_* }f\left( dx,t\right) \leq \max \left\{ \int_{ {%
\mathbb{R}}_*}f_{0}\left( dx\right) ,\left(\frac{2}{a_{1}}\int_{{\mathbb{R}}%
_*}\eta \left( dx\right)\right)^{\frac 1 2} \right\}.
\end{equation*}
\end{remark}

We now prove that solutions obtained in Proposition \ref%
{thm:existence_evolution} are weak solutions in the sense of Definition \ref%
{DefTimeSol}.

\begin{proposition}
\label{WeakSolBounded} Suppose that the assumptions in Proposition \ref%
{thm:existence_evolution} hold. Then, the solution $f$ obtained  is a
weak solution of (\ref{evolEqTrunc}) in the sense of Definition \ref%
{DefTimeSol}.
\end{proposition}

\begin{proof}
Multiplying (\ref{evolEqTrunc}) by a continuous test function $\varphi \in
C^{1}\left( \left[ 0,T\right] ,C\left( {\mathbb{R}}_*\right) \right) $
with $T>0$ we obtain, using the action of the convolution on a test function
in (\ref{ConvDef}): 
\begin{eqnarray}
&&\int_{{\mathbb{R}}_*}\varphi \left( x,t\right) \dot{f}\left(
dx,t\right)  \notag \\
&=&\frac{1}{2}\int_{{\mathbb{R}}_*}\int_{{\mathbb{R}}_*}K\left(
x,y\right) \left[ \varphi \left( x+y,t\right) \zeta _{R_{\ast }}\left(
x+y\right) -\varphi \left( x,t\right) -\varphi \left( y,t\right) \right]
f\left( dx,t\right) f\left( dy,t\right) \hspace{1cm}  \notag \\
&&+\int_{{\mathbb{R}}_*} {\varphi \left( x,t\right)} \eta \left( dx\right) .
\label{IntStrForm}
\end{eqnarray}%
As mentioned earlier, the left-hand side can be rewritten as 
\begin{equation*}
\int_{{\mathbb{R}}_*}\varphi \left( x,t\right) \dot{f}\left(
dx,t\right) =\frac{d}{dt}\int_{{\mathbb{R}}_*}\varphi \left( x,t\right)
f\left( dx,t\right) -\int_{{\mathbb{R}}_*} {\dot{\varphi}} \left(
x,t\right) f\left( dx,t\right) .
\end{equation*}%
Therefore, $f$ satisfies (\ref{eq:evol_eqWeak}) in Definition \ref{DefTimeSol}. 
\end{proof}

We will use in the following the dynamical
system notation $S\left( t\right) $ for the map%
\begin{equation}
S\left( t\right) f_{0}=f\left( \cdot ,t\right)  \label{SemDef}
\end{equation}%
where $f$ is the solution of (\ref{evolEqTrunc}) obtained in Proposition \ref%
{thm:existence_evolution}. Note that by uniqueness $S\left( t\right) $ has the following
semigroup property:%
\begin{equation}
S\left( t_{1}+t_{2}\right) =S\left( t_{1}\right) S\left( t_{2}\right) \text{
for each }t_{1},t_{2}\in \mathbb{R}_+ . \label{SemProp}
\end{equation}
The operators $S\left( t\right) $ define a mapping%
\begin{equation}
S\left( t\right) :\mathcal{X}_{R_{\ast }}\rightarrow \mathcal{X}_{R_{\ast
}}\ \text{for each }t\geq 0  \label{SemMapping}
\end{equation}
where $\mathcal{X}_{R_{\ast }}=\left\{ f\in \mathcal{M}_{+}({\mathbb{R}}%
_*):f\left( \left( 0,1\right) \cup \left( 2R_{\ast },\infty \right)
\right) =0\right\}$, as before.

We can now prove the following result:
\begin{proposition}
\label{thm:existence_truncated} Under the assumptions of Proposition~\ref%
{thm:existence_evolution}, there exists a stationary injection solution $\hat{f}%
\in \mathcal{M}_{+}({\mathbb{R}}_*)$ to (\ref{evolEqTrunc}) as defined in
Definition~\ref{DefStatSol}.
\end{proposition}

\begin{proof}
We provide below a proof of the statement but skip over some standard technical computations. Further details about these technical estimates can be found from \cite[Section 5]{Vuoksenmaa20}.

We first construct an invariant region for the evolution equation %
\eqref{evolEqTrunc}.  Let $f_0 \in \chi_{R_*} $ and set $f(t)=S(t)f_0$. 
In particular, $f$  satisfies (\ref{eq:evol_eqWeak}).
Let us then choose a time-independent 
test function such that $\varphi (x)=1$ when $1\le x\leq 2R_{\ast}$.
Similarly to (\ref{eq:zeromomevolreg}) and using the fact that
$f(\cdot,t)$ has support in $[1,2 R_*]$, the lower bound for $K$
implies an estimate 
\begin{equation*}
\frac{d}{dt}\int_{[1,2R_{\ast }]}f\left( dx,t\right) \leq -\frac{a_{1}}{2}%
\left( \int_{[1,2R_{\ast }]}f\left( dx,t\right) \right) ^{2}+c_{0}
\end{equation*}%
where $c_{0}=\int_{\R_*}\eta \left( dx\right) $. 
As in Remark \ref{th:boundforzeromomreg},
inspecting the sign of the right hand side we then find that
if we choose any $M\geq \sqrt{\frac{2c_{0}}{a_{1}}}$, then the following set is invariant under the time-evolution:
\begin{equation}
\mathcal{U}_{M}=\left\{ f\in \mathcal{X}_{R_{\ast }}:\int_{[1,2R_{\ast
}]}f(dx)\leq M\right\} \, . \label{eq:invariant_region}
\end{equation} 
Moreover, $\mathcal{U}_{M}$ is
compact in the $\ast -$weak topology due to Banach-Alaoglu's Theorem (cf.\ 
\cite{Brezis}), since it is an intersection of a $\ast -$weak closed set $\mathcal{X}_{R_{\ast }}$ and the closed ball 
$\norm{f}\le M$.

Consider the operator $S(t):\mathcal{X}_{R_{\ast }}\rightarrow \mathcal{X}%
_{R_{\ast }}$ defined in (\ref{SemDef}). We now endow $\mathcal{X}_{R_{\ast }}$
with the $\ast -$weak topology and prove that $S(t)$ is
continuous. Due to Proposition \ref{WeakSolBounded} we have that $f(\cdot
,t)=S(t)f_{0}$ satisfies (\ref{eq:EvolWeakInt}) for any test function $%
\varphi \in C^{1}\left( \left[ 0,T\right] ,C_c\left( \R_*
\right) \right) $,\ $0\leq t\leq T$ with $T>0$ arbitrary. Let $f_{0},\hat{f}%
_{0}\in \mathcal{X}_{R_{\ast }}$. We write $f(\cdot ,t)=S(t)f_{0}$ and $\hat{%
f}(\cdot ,t)=S(t)\hat{f}_{0}.$  
Using (\ref{eq:EvolWeakInt}) and  subtracting the
corresponding equations for $f$ and $\hat{f}$, we obtain
\begin{eqnarray}
&&\int_{{\mathbb{R}}_*}\varphi \left( x,t\right) (f\left( dx,t\right) -%
\hat{f}\left( dx,t\right) )-\int_{{\mathbb{R}}_*}\varphi \left( x,0\right)
(f_{0}\left( dx\right) -\hat{f}_{0}\left( dx\right) )  \notag \\
&=&\int_{0}^{t}ds\int_{{\mathbb{R}}_*}(f\left( dx,t\right) -\hat{f}\left(
dx,t\right) )\left( \dot{\varphi} \left( x,s\right) +\mathcal{L}\left[
\varphi \right] \left( x,s\right) \right)  \label{SymmEq}
\end{eqnarray}
where%
\begin{equation*}
\mathcal{L}\left[ \varphi \right] \left( x,s\right) =\frac{1}{2}\int_{{%
\mathbb{R}}_*}K\left( x,y\right) \left[ \varphi \left( x+y,s\right) \zeta
_{R_{\ast }}\left( x+y\right) -\varphi \left( x,s\right) -\varphi \left(
y,s\right) \right] (f\left( dy,s\right) +\hat{f}\left( dy,s\right) ) \,.
\end{equation*}
For the derivation of (\ref{SymmEq}), we have used symmetry properties under
the transformation  $x\leftrightarrow y$: clearly, 
$K\left( x,y\right) %
\left[ \varphi \left( x+y,s\right) \chi _{\left\{ x+y\leq R_{\ast }\right\}
}\left( x,y\right) -\varphi \left( x,s\right) -\varphi \left( y,s\right) %
\right] $ is then symmetric and  $\left[ f\left( dx,s\right) \hat{f}\left( dy,s\right) -f\left(
dy,s\right) \hat{f}\left( dx,s\right) \right] $ is antisymmetric, and hence their product integrates to zero.

Consider then an arbitrary $\psi \in C_c\left( {\mathbb{R}}_*\right)$. 
We claim that there is a
test function $\varphi \in C^{1}\left( \left[ 0,t\right] ,C_{c}\left( {%
\mathbb{R}}_*\right) \right) $ such that%
\begin{equation}
\dot{\varphi} \left( x,s\right) +\mathcal{L}\left[ \varphi \right]
\left( x,s\right) =0\ \ \text{for }0\leq s\leq t\, , \ { x\ge 1}\,, \quad \text{with\ \ }\varphi
\left( \cdot ,t\right) =\psi \left( \cdot \right) .  \label{AdjEqu}
\end{equation}
Given such a function $\varphi $, since $f$ and $\hat{f}$ have no support on $(0,1)$, equation (\ref{SymmEq})
implies
\begin{equation}\label{eq:adjointsemigest}
\int_{{\mathbb{R}}_*}\psi \left( x\right) (f\left( dx,t\right) -\hat{f}%
\left( dx,t\right) )=\int_{{\mathbb{R}}_*}\varphi \left( x,0\right)
(f_{0}\left( dx\right) -\hat{f}_{0}\left( dx\right) ) \,.
\end{equation}
Therefore, if such a function $\varphi $ exists for any $%
\psi \in C_c\left( {\mathbb{R}}_*\right)$, we would find that the estimate at time $t$, 
$\left\vert \int_{{\mathbb{R}}_*}\psi \left(
x\right) (f\left( dx,t\right) -\hat{f}\left( dx,t\right) )\right\vert $, will become 
arbitrarily small if the estimate at time $0$,
$\left\vert \int_{[1, {2R_{\ast }}]}\varphi \left(
x,0\right) (f_{0}\left( dx\right) -\hat{f}_{0}\left( dx\right) )\right\vert $,
is made sufficiently small. In particular, this property can be used to prove that 
for every $f(t)=S(t)f_0$ in a $\ast -$weak open set $U$
one can find a $\ast -$weak open neighbourhood $V$ of $f_0$ such that for any 
$\hat{f}_{0}\in V$ one has $S(t)\hat{f}_{0} \in U$.  
Hence, the $\ast -$weak continuity of $%
S\left( t\right) $ would then follow.

{

In order to conclude the proof of the continuity of $S(t)$ in the $\ast -$%
weak topology it only remains to prove the existence of $\varphi \in
C^{1}\left( \left[ 0,t\right] ,C_{c}\left( {\mathbb{R}}_*\right) \right) $
satisfying (\ref{AdjEqu}) for a fixed $%
\psi \in C_c\left( {\mathbb{R}}_*\right)$. 
First, let us choose $a\in (0,1)$ and $b\ge 4R_{\ast }$
so that the support of $\psi$ is contained in $I_0:=[a,b]$.  
We now construct $\varphi$ as a solution to an evolution equation in 
the Banach space $Y:=\defset{h\in C(\R_*)}{h(x)=0\text{ if }x\le a\text{ or }x\ge b}$ which is a closed subspace of $C_0(\R_*)$.

More precisely, we now look for solutions $\varphi\in \tilde{X}:= C([0,t],Y)$,
endowed with the weighted norm $\norm{\varphi}_M := \sup_{x\in \R_*,\, s\in [0,t]}|\varphi(x,s)| \rme^{M (s-t)}$.  The parameter $M>0$ is chosen sufficiently large, 
as explained later.

Clearly, $\psi\in Y$. 
To regularize the small values of $x$, we choose a function 
$\phi_a\in C(\R_*)$ such that $0\le \phi_a\le 1$,
$\phi_a(x)=0$ if $x\le a$, and $\phi_a(x)=1$ if $x\ge 1$, and then define
\[
 \tilde{\mathcal{L}}[\varphi](x,s) := \phi_a(x) \mathcal{L}[\varphi](x,s)\,,\quad
 x>0\,,\ 0\le s\le t\,.
\]
Now, if $\varphi\in\tilde{X}$, we have $\tilde{\mathcal{L}}[\varphi](x,s)=0$
both if $x\le a$ (due to the factor $\phi_a$) and if $x\ge b$, since $K(x,y)=0$ if $x\ge b \ge 4 R_*$.  In addition, the assumptions guarantee that $x\mapsto \mathcal{L}[\varphi](x,s)$ is continuous, so we find that $ \tilde{\mathcal{L}}[\varphi](x,s)\in Y$ for any fixed $s$.

We look for solutions $\varphi$ as fixed points satisfying $\varphi=\mathcal{A}[\varphi]$, where
\[
 \mathcal{A}[\varphi](x,s) := \psi(x) + \int_s^t \tilde{\mathcal{L}}[\varphi](x,s') ds'\,, \quad 
 x>0\,,\ 0\le s\le t\,.
\]
A straightforward computation, using the uniform bounds of total variation norms of $f$ and $\hat{f}$, shows that $\mathcal{A}$ is a map from $\tilde{X}$
to itself.  In addition, since
\[
 |\mathcal{A}[\varphi_1](x,s)-\mathcal{A}[\varphi_2](x,s)|
 \le  \frac{3a_2}{2} \left(\norm{f}_t+\norm{\hat{f}}_t \right) \norm{\varphi_1-\varphi_2}_{M}
 \int_s^t \rme^{-M(s'-t)} ds'\,,
\]
we find that $\mathcal{A}$ is also a contraction if we fix $M$ so that 
$M> \frac{3a_2}{2} \left(\norm{f}_t+\norm{\hat{f}}_t \right)$.
Thus by the Banach fixed point theorem, there is a unique $\varphi\in \tilde{X}$ such that $\varphi=\mathcal{A}[\varphi]$.  This choice
satisfies (\ref{AdjEqu}), at least for $x\ge 1$, and hence completes the proof of continuity of $S(t)$.

}

We next prove that also $t\mapsto S\left( t\right) f_{0}$ is
continuous in the $\ast -$weak topology. Let $t_{1},t_{2}\in \left[ 0,T%
\right] $ with $t_{1}<t_{2}.$ Let $\varphi \in C_{c}\left( {\mathbb{R}}%
_*\right) .$ Using (\ref{eq:EvolWeakInt}) we obtain:%
\begin{eqnarray*}
&&\int_{{{\mathbb{R}}_*}}\varphi \left( x\right) \left[ f\left(
dx,t_{2}\right) -f\left( dx,t_{1}\right) \right] \\
&=&\frac{1}{2}\int_{t_{1}}^{t_{2}}ds\int_{{{\mathbb{R}}_*}}\int_{{{\mathbb{%
R}}_*}}K\left( x,y\right) \left[ \varphi \left( x+y\right) \zeta _{R_{\ast
}}\left( x+y\right) -\varphi \left( x\right) -\varphi \left( y\right) \right]
f\left( dx,s\right) f\left( dy,s\right) \\
&&+\int_{t_{1}}^{t_{2}}ds\int_{{{\mathbb{R}}_*}}\varphi \left( x\right)
\eta \left( dx\right) .
\end{eqnarray*}%
Thus using the bound $\left\Vert f\right\Vert _{T}<\infty$ we
obtain%
\begin{equation}\label{eq:timeweakstar}
\left\vert \int_{{{\mathbb{R}}_*}}\varphi \left( x\right) \left[ f\left(
dx,t_{2}\right) -f\left( dx,t_{1}\right) \right] \right\vert \leq C \left( t_{2}-t_{1}\right)\norm{\varphi}\,,
\end{equation}
where the constant $C$ does not depend on $t_1$, $t_2$ or $\varphi$.
Therefore, the mapping $t\mapsto S\left( t\right) f_{0}$ is continuous
in the $\ast -$weak topology.

We can now conclude the proof of Proposition \ref{thm:existence_truncated}.
As proven above, for any fixed $t$, the operator $S(t):\mathcal{U}_{M}\rightarrow \mathcal{U}_{M}$ is continuous and $\mathcal{U}_{M}$ is convex and compact when 
endowed with the $\ast -$weak topology. 
Using Schauder fixed point theorem, for all $%
\delta >0$, there exists a fixed point $\hat{f}_{\delta }$ of $S(\delta )$
in $\mathcal{U}_{M}$. 
{
In addition, $\mathcal{U}_{M}$ is metrizable and hence sequentially compact.
As shown in \cite[Theorem 1.2]{EM05}, these properties imply
that there is $\hat{f}$ such that 
$S(t)\hat{f} = \hat{f}$ for all $t$.  Thus $\hat{f}$ is a stationary injection solution to 
\eqref{evolEqTrunc}.
}
\end{proof}

\bigskip

We now prove Theorem \ref{thm:existence}.

\begin{proofof}[Proof of Theorem~\ref{thm:existence} (existence)] 
Given a kernel $K(x,y)$ satisfying \eqref{eq:cond_kernel2}, \eqref{eq:cond_kernel3}, it can be rewritten as  
\begin{equation}
K\left(  x,y\right)  =\left(  x+y\right)  ^{\gamma}\Phi\left(  \frac{x}%
{x+y},x\right)  \label{B1in3}%
\end{equation}
where
\begin{equation}\label{eq:B1bound}
 \frac{C_{1}}{s^{p} \left(  1-s\right)  ^{p}}  \leq
\Phi\left(  s,x\right)  \leq \frac{C_{2}}{s^{p}\left(
1-s\right)  ^{p}},\quad {(s,x)\in(0,1)\times \R_{\ast}}
\end{equation}
with $p=\max\left\{  \lambda,-\left(  \gamma+\lambda\right)
\right\} $ and the constants $C_1>0$, $C_2<\infty$ independent of $x$.
Notice that the dependence of the function $\Phi$ on $x$ is due to the fact that we are not assuming the kernel $K(x,y)$ to be an homogeneous function.

By definition of $p$, we have $\gamma + 2 p =|\gamma+2\lambda|\ge 0$, and thus always $p \geq -\frac{\gamma}{2}$.
On the other hand, 
by assumption, $|\gamma+2\lambda|<1$, and thus also $p<\frac{1-\gamma}{2}$.
Reciprocally, we observe that kernels with the form \eqref{B1in3} satisfying \eqref{eq:B1bound} with $p \geq -\frac{\gamma}{2}$ satisfy also \eqref{eq:cond_kernel2}, \eqref{eq:cond_kernel3} .

We use two levels of truncations. First, given $\varepsilon$ with $0<\varepsilon<1$ we define %
\begin{equation}
K_{\varepsilon}\left(  x,y\right)  =\min\left\{  \left(  x+y\right)  ^{\gamma
},\frac{1}{\varepsilon}\right\}  \Phi_{\varepsilon}\left(  \frac{x}%
{x+y}, x\right)  +\varepsilon\,,     \label{eq:1levTrunc}
\end{equation}
where 
$\Phi_\varepsilon$ is smooth, non-negative, and bounded by $\frac{A}{\varepsilon^{\sigma}}$ everywhere, and satisfies
\begin{equation}
\Phi_{\varepsilon}\left(  s, x\right)  =
\begin{cases}
\Phi\left(  s, x\right) \,, & \text{if\ \ }\Phi\left(  s, x\right)  \leq\frac
{A}{\varepsilon^{\sigma}}\,,\\
0\,, & \text{if\ \ }\Phi\left(  s, x\right)  \geq\frac{2A}{\varepsilon^{\sigma}}\,.%
\end{cases}
 \label{B3in3}%
\end{equation}
Here
$A$ is a large constant independent of $\varepsilon$; we take $A=1$ when $\Phi$ is unbounded, and assume it sufficiently large in a way that will be seen in the proof if $\Phi$ is bounded. 
Concerning $\sigma$ we take $\sigma=0$ if $p\leq 0$ for any $\gamma$, $\sigma>0$ arbitrary small if $p>0$ and $\gamma\leq 0$ and $0< \sigma <\frac p \gamma$ if $p>0$ and $\gamma >0$.
We then have%
\begin{equation}
0\leq\Phi_{\varepsilon}\left(  s, x\right)  \leq C_{2}\min\left\{  \frac
{1}{s^{\lambda}}\frac{1}{\left(  1-s\right)  ^{\lambda}}+s^{\gamma+\lambda
}\left(  1-s\right)  ^{\gamma+\lambda},\frac{A}{C_{2}\varepsilon^{\sigma}}\right\}\,.
\label{B2in3}%
\end{equation}

The second level of truncation is to define
\begin{equation}\label{eq:2levTrunc}
K_{\varepsilon,R_{\ast}}\left(  x,y\right)  =K_{\varepsilon}\left(
x,y\right)  \omega_{R_{\ast}}\left(  x,y\right), 
\end{equation}
where $\omega_{R_{\ast}}\in C^{\infty}_{0}(\R^2_{+})$, $0\le \omega_{R_*}\le 1$, and
\begin{equation*}
\omega_{R_{\ast}}(x,y)  =
\begin{cases}
1\,, &\text{if\ \ } (x,y)\in [0,2R_{\ast} ]^2\,, \\ 
0\,, &\text{if\ \ } x\geq 4 R_{\ast} \text{ or } y\geq 4R_{\ast} \,.
\end{cases}
\end{equation*}
Notice that, if $\gamma\leq0$ the truncation in  $\min\left\{  \left(x+y\right)  ^{\gamma},\frac{1}{\varepsilon}\right\}  $ in \eqref{eq:1levTrunc}
does not have any
effect, because we are only interested in the region where $x\geq1$ and $y\geq1$, due to the fact that the solutions we construct satisfy $f((0,1))=0$.

From
Proposition~\ref{thm:existence_truncated}, to every $\varepsilon$ and $R_*$, there exists a stationary injection
solution $f_{\varepsilon ,R_{\ast }}$ satisfying 
\begin{align}
&\frac{1}{2}\int_{{{\mathbb{R}}_{*}}}\int_{{{\mathbb{R}}_{*}}}K_{\varepsilon
,R_{\ast }}\left( x,y\right) \left[ \varphi \left( x+y\right) \zeta
_{R_{\ast }}\left( x+y\right) -\varphi \left( x\right) -\varphi \left(
y\right) \right] f_{\varepsilon ,R_{\ast }}\left( dx\right) f_{\varepsilon
,R_{\ast }}\left( dy\right)  \notag  \label{WeakTruncForm} \\
&+\int_{{{\mathbb{R}}_{*}}}\varphi \left( x\right) \eta \left( dx\right) =0\, ,
\end{align}%
for any test function $\varphi \in {C}_{c}{({\mathbb{R}}_{*})}$.
As in the proof of Lemma \ref{lem:flux}
consider any $z,\delta>0$ and take $\chi_\delta \in C^\infty(\R_+)$ satisfying $\chi_\delta (x) = 1$ if $x \leq z$, and $\chi_\delta(x) = 0$ if $x \geq z+\delta $. 
Then $\varphi(x) = x \chi_\delta(x) $ is a valid non-negative test function. Since $\zeta_{R_*}\le 1$, we may employ
the inequality $\varphi \left( x+y\right) \zeta
_{R_{\ast }}\left( x+y\right)\le \varphi \left( x+y\right)$ in (\ref{WeakTruncForm}), and conclude that for these test functions
\begin{align*}
&\frac{1}{2}\int_{{{\mathbb{R}}_{*}}}\int_{{{\mathbb{R}}_{*}}}K_{\varepsilon
,R_{\ast }}\left( x,y\right) \left[ \varphi \left( x+y\right)  -\varphi \left( x\right) -\varphi \left(
y\right) \right] f_{\varepsilon ,R_{\ast }}\left( dx\right) f_{\varepsilon
,R_{\ast }}\left( dy\right) 
+\int_{{{\mathbb{R}}_{*}}}\varphi \left( x\right) \eta \left( dx\right) \ge 0\,.
\end{align*}%
Using the equalities derived in the proof of Lemma \ref{lem:flux}
and taking $\delta\to 0$ then proves that 
\begin{equation}
\int_{\left( 0,z\right] }\int_{( z-x,\infty ) }K_{\varepsilon
,R_{\ast }}\left( x,y\right) xf_{\varepsilon ,R_{\ast }}\left( dx\right)
f_{\varepsilon ,R_{\ast }}\left( dy\right) \le \int_{(0,z]}x\eta \left(
dx\right)\,, \quad \text{for }z>0\, .
\label{eq:CurrentboundRstar}
\end{equation} 

A lower bound for the left hand-side and an upper bound for the right
hand-side of (\ref{eq:CurrentboundRstar}), both independent of $R_{\ast }$, are computed next. Since $\supp%
\eta \subset \lbrack 1,L_\eta]$ and $\int \eta $ is{\ bounded}, then 
\begin{equation}
\int_{(0,z]} x\eta \left( dx\right) \leq \int_{[1,L_\eta]} x\eta \left( dx\right)
=:c,  \label{eq:bound_eta}
\end{equation}%
where $c$ is a constant independent of $R_{\ast }$ and $c$ is bounded by $L_{\eta} \vert|\eta\vert|$. On the other hand we
have $K_{\varepsilon ,R_{\ast }}\left( x,y\right) \geq \varepsilon >0$
for $\left( x,y\right) \in \left[ 1,2 R_{\ast }\right] ^{2}.$ Then,%
\begin{equation*}
\varepsilon \int_{\left( 0,z\right] }\int_{(z-x,2 R_{\ast }] }xf_{\varepsilon ,R_{\ast }}\left( dx\right)
f_{\varepsilon ,R_{\ast }}\left( dy\right) \leq c\quad \text{if }0<z\le 2 R_*\, .
\end{equation*}
Using that here
\begin{equation*}
[2z/3,z]^{2}\subset \left\{ \left( x,y\right) \in \mathbb{R}_{+}^{2}:0<
x\leq z,\ z-x< y\leq 2 R_* \right\} \, ,
\end{equation*}
we obtain
\begin{equation*}
\varepsilon \iint_{[2z/3,z]^{2}}xf_{\varepsilon ,R_{\ast
}}(dx)f_{\varepsilon ,R_{\ast }}(dy)\leq c\quad \text{if }0<z\le 2 R_*\, .
\end{equation*}%
Since $x\geq 2z/3$ in the domain of integration, we obtain 
\begin{equation*}
  2z/3 \left( \int_{[2z/3,z]}f_{\varepsilon ,R_{\ast
}}(dx)\right) ^{2}\leq \frac{c}{\varepsilon},
\end{equation*}%
which implies that 
\begin{equation}
\frac{1}{z}\int_{[2z/3,z]}f_{\varepsilon ,R_{\ast }}(dx)\leq \frac{C_{\varepsilon }}{z^{3/2}},\ \ \ 0< z\leq 2 R_{\ast },
\end{equation}
where $C_{\varepsilon }$ is a numerical constant depending on $\varepsilon $
but independent of $R_{\ast }$.
Since the right hand side is integrable on $[1,2R_{\ast }]$, 
Lemma~\ref{lem:bound} may be employed to obtain a bound
\begin{equation}
\int_{[1,2 R_*]}f_{\varepsilon ,R_{\ast }}(dx)\leq \frac{2 C_\varepsilon}{\ln(3/2)} + C_\varepsilon \frac{1}{\sqrt{2 R_{\ast }}}\,.
\label{EstTruncFunctionsmallz}
\end{equation}
Since the support of $f_{\varepsilon ,R_{\ast }}$ lies in $[1,2 R_{\ast }]$ we find that 
for all $R_{\ast }\ge 1$
\begin{equation}
\int_{\R_*}f_{\varepsilon ,R_{\ast }}(dx)\leq \bar C_{\varepsilon }\,,
\label{EstTruncFunction}
\end{equation}
where $\bar C_{\varepsilon }$ is a constant independent of $R_{\ast }$. 
Following the same argument for arbitrary lower limit $y\ge 1$, we also obtain a decay bound
\begin{equation}
\int_{[y,\infty)}f_{\varepsilon ,R_{\ast }}(dx)\leq   \bar C_{\varepsilon } y^{-\frac{1}{2}}\,.
\label{EstTruncFunctiondecay}
\end{equation}

Thus, estimate (\ref{EstTruncFunction})
implies that for each $\varepsilon$ the family of solutions $\{f_{\varepsilon ,R_{\ast }}\}_{R_{\ast }\ge 1}$  is contained in a closed unit ball of $\mathcal{M}_{+,b}\left( \mathbb{R}_*\right) $.  This is a sequentially compact set in the 
$\ast -$weak topology, and thus by taking a subsequence if needed, we can find $f_{\varepsilon
}\in \mathcal{M}_{+,b}\left( \mathbb{R}_*\right) $ such that $f_{\varepsilon
}\left( \left( 0,1\right) \right) =0$ and%
\begin{equation}
f_{\varepsilon ,R_{\ast }^{n}}\rightharpoonup f_{\varepsilon }\text{ as }%
n\rightarrow \infty \text{ in the }\ast -\text{weak topology}
\label{fWeakLimit}
\end{equation}%
with $R_{\ast }^{n}\rightarrow \infty $ as $n\rightarrow \infty$.
Note that then we can use the earlier ``step-like'' test-functions and the bounds
(\ref{EstTruncFunction}) and
(\ref{EstTruncFunctiondecay}) to conclude that also the limit functions satisfy similar estimates, namely,
\begin{equation}
\int_{\left( 0,\infty \right) }f_{\varepsilon }(dx)\leq \bar C_{\varepsilon }\,, \qquad
\int_{\left[ y,\infty \right) }f_{\varepsilon }(dx)\leq  \bar C_{\varepsilon }y^{-\frac{1}{2}}\,, \quad \text{if } y\ge 1\,.
\label{fBound}
\end{equation}%


Consider next a fixed test function $\varphi \in C_{c}(\R_*)$.  Now for all large enough values of $n$, we have  $\varphi \left( x+y\right) \zeta
_{R^n_{\ast }}\left( x+y\right) =\varphi \left( x+y\right)$ everywhere, since the support of $\varphi$ is bounded.  We claim that as $n\to \infty$, the limit of (\ref{WeakTruncForm}) is given by 
\begin{equation}
\frac 1 2\int_{\R_*^{2}}K_{\varepsilon }\left( x,y\right) [\varphi
(x+y)-\varphi (x)-\varphi (y)]f_{\varepsilon }\left( dx\right)
f_{\varepsilon }\left( dy\right) +\int_{\R_*}\varphi (x)\eta \left(
dx\right) =0.  \label{WeakFormEps}
\end{equation}
Since $f_{\varepsilon ,R_{\ast }^{n}}$ has support in $[1,2 R_{\ast }]$, it follows that we may always replace $K_{\varepsilon,R_{\ast }}(x,y)$ in (\ref{WeakTruncForm}) by 
$K_{\varepsilon }(x,y)$ without altering the value of the integral.
By the above observations,  it suffices to show that 
\begin{equation}
\lim_{n\to \infty}\int_{\R_*^{2}}\phi( x,y)
\mu_n(dx)\mu_n(dy) =
\int_{\R_*^{2}}\phi( x,y) f_{\varepsilon }\left( dx\right)
f_{\varepsilon }\left( dy\right)\, ,  \label{eq:stongereps}
\end{equation}
for $\mu_n(dx):=f_{\varepsilon ,R^n_{\ast }}\left( dx\right)$ and  
\[
 \phi( x,y) := K_{\varepsilon }\left( x,y\right) [\varphi
(x+y)-\varphi (x)-\varphi (y)]\,.
\]
Note that although $\phi\in C_b(\R_*^2)$, it typically would not have compact support.  However, the earlier tail estimates suffice to control the large values of $x,y$, as we show in detail next.

We prove (\ref{eq:stongereps}) by showing that every subsequence has a subsequence such that the limit holds.
For notational convenience, let $\mu_n$ denote the first subsequence and consider an arbitrary $\vep'>0$.
We first regularize the support of $\phi$ by choosing a function  
$g:\R_+\to [0,1]$ which is continuous and for which $g(r)=1$, for $r\le 1$, and $g(r)=0$, for $r\ge 2$.
We set $\phi_M(x,y) := g\left(\frac{x}{M}\right) g\left(\frac{y}{M}\right)g\left(\frac{1}{M x}\right) g\left(\frac{1}{M y}\right)\phi(x,y)$.
Then for every $M$, we have $\phi_M\in C_c(\R_*^2)$ and thus it is uniformly continuous.
By (\ref{fBound}), we may use dominated convergence theorem to conclude that 
$\int \phi_M(x,y)f_{\varepsilon }\left( dx\right) f_{\varepsilon }\left( dy\right)\to
\int \phi(x,y)f_{\varepsilon }\left( dx\right) f_{\varepsilon }\left( dy\right)$ as $M\to \infty$.
Thus for all sufficiently large $M$, we have
$\left|\int \phi_M(x,y)f_{\varepsilon }\left( dx\right) f_{\varepsilon }\left( dy\right)-
\int \phi(x,y)f_{\varepsilon }\left( dx\right) f_{\varepsilon }\left( dy\right)\right|<\vep'$.
On the other hand, by the decay bound in (\ref{EstTruncFunctiondecay}) we can find a constant
$C$ which does not depend on $R_\ast$ and for which 
$\left|\int \phi_M(x,y)\mu_n(dx)\mu_n(dy)-\int \phi(x,y)\mu_n(dx)\mu_n(dy)\right|\le C M^{-\frac{1}{2}}$.
We fix $M=M(\vep')$ to be a value such that also this second bound is less than $\vep'$
for all $n$.

{
In order to conclude the proof of (\ref{eq:stongereps}) it only remains to show that
 $\int\phi_{M}\left(  x,y\right)  \mu_{n}\left(  dx\right)  \mu_{n}\left(
dy\right)  $ converges to $\int\phi_{M}\left(  x,y\right)  f_{\varepsilon
}\left(  dx\right)  f_{\varepsilon}\left(  dy\right)  $ as $n\rightarrow
\infty.$ This is just a consequence of the fact that the convergence $\mu
_{n}\left(  dx\right)  \rightharpoonup f_{\varepsilon}\left(  dx\right)  $ as
$n\rightarrow\infty$ in the $\ast-$weak topology of the space $\mathcal{M}%
_{+,b}\left(  \left[  \frac{1}{2M},2M\right]  \right)  $ implies the
convergence $\mu_{n}\left(  dx\right)  \mu_{n}\left(  dy\right)
\rightharpoonup f_{\varepsilon}\left(  dx\right)  f_{\varepsilon}\left(
dy\right)  $ as $n\rightarrow\infty$ in the $\ast-$weak topology of the space
$\mathcal{M}_{+,b}\left(  \left[  \frac{1}{2M},2M\right]  ^{2}\right)  .$ This
result can be found for instance in \cite{Billings}, Theorem 3.2 for
probability measures, which implies the result for arbitrary measures using
simple rescaling arguments.
}

Since $f_\varepsilon$ is then a stationary solution to (\ref{eq:time_evol}) with $K=K_{\ep}$, 
we can apply Lemma \ref{lem:flux} directly, and conclude  that
\begin{equation}
\int_{\left(0,z\right] }xf_{\varepsilon }\left( dx\right) \int_{( 
z-x,\infty ) }K_{\varepsilon }\left( x,y\right) f_{\varepsilon
}\left( dy\right) \leq c\quad \text{if }z>0\, ,
\label{FluxEstEps}
\end{equation}%
where $c$ is defined in~\eqref{eq:bound_eta} and is independent 
of $\varepsilon$.
We now observe that (\ref{eq:B1bound}) and (\ref{eq:1levTrunc})-(\ref{B3in3}) imply for all sufficiently small $\varepsilon$ 
\begin{equation*}
K_{\varepsilon }(x,y)\geq \varepsilon+C_{0}\min\{z^{\gamma},\frac{1}{\varepsilon}\} \quad \text{for}\ \ (x,y)\in \left[ \frac{z}{2},z\right]^2
\end{equation*}
where $C_{0}>0$ is independent of $\varepsilon$ and we used that $\frac{x}{x+y}\in 
\left[\frac 1 3, \frac 2 3\right]$.
Combining this estimate with (\ref{FluxEstEps}) as well as the fact that 
\begin{equation*}
[2z/3,z]^{2}\subset \left\{ \left( x,y\right) \in \mathbb{R}_{+}^{2}:0<
x\leq z,\ z-x< y<\infty \right\} \ 
\end{equation*}
we obtain
\begin{equation*}	
\left( \varepsilon+C_{0}\min\{z^{\gamma},\frac{1}{\varepsilon}\} \right) \frac{2}{3} z \left( \int_{[2z/3,z]}f_{\varepsilon }(dx)\right) ^{2}\leq c\quad
\text{for all }z\in (0,\infty ).
\end{equation*} 
Therefore, we obtain the following estimates for the measures $f_{\varepsilon}\left(  dx\right)  $:
\begin{equation}
\frac{1}{z}\int_{\left[  \frac{2z}{3},z\right]  }f_{\varepsilon}\left(
dx\right)  \leq\frac{\tilde{C}}{z^{\frac{3}{2}}}\left(  \frac{1}{\min\left(
z^{\gamma},\frac{1}{\varepsilon}\right)  }\right)  ^{\frac{1}{2}}, \label{A1} 
\end{equation}
\begin{equation}
\frac{1}{z}\int_{\left[  \frac{2z}{3},z\right]  }f_{\varepsilon}\left(
dx\right)  \leq\frac{\tilde{C}}{z^{\frac{3}{2}}\sqrt{\varepsilon}}, \label{A2}
\end{equation}
where $\tilde{C}$ is independent of $\varepsilon$.

Consider first the case $\gamma\le 0$ and 
recall that then $p\ge 0$ and $z^\gamma \le 1$ for $z\ge 1$.  Since $f_\varepsilon((0,1))=0$, 
then the bound (\ref{A1}) implies that for all $z\ge 1$ we have
\begin{equation}
\frac{1}{z}\int_{\left[  \frac{2z}{3},z\right]  } x^{\gamma+p}f_{\varepsilon}\left(
dx\right)  \leq C z^{\frac{\gamma+2p-3}{2}}\,. \label{A1upgrade} 
\end{equation}
Since $\gamma+2 p <1$, Lemma \ref{lem:bound} implies then that
for all $y\ge 1$,
\begin{equation}
\int_{[y,\infty)  } x^{\gamma+p}f_{\varepsilon}\left(
dx\right)  \leq C y^{-\frac{1-\gamma-2p}{2}}\, \label{A1upbound} 
\end{equation}
where the constant $C$ does not depend on $\varepsilon$.
In particular, then the measures $x^{\gamma+p} f_{\varepsilon}\left(
dx\right)$ belong to a $\ast$-weak compact set, and there exist $F\in \mathcal{M}_{+,b}\left( 
\mathbb{R}_{*}\right) $ such that%
\begin{equation}
x^{\gamma+p} f_{\varepsilon_n}\left(dx\right)
\rightharpoonup F\left(dx\right)\text{ as }n\rightarrow \infty \text{
in the }\ast -\text{weak topology}  \label{fepsWeakLimitneggamma}
\end{equation}%
for some sequence $( \varepsilon _{n}) _{n\in \mathbb{N}}$
with $\lim_{n\rightarrow \infty }\varepsilon _{n}=0$.
We denote $f\left(dx\right)= x^{-\gamma-p} F\left(dx\right)$, and then $f
\in \mathcal{M}_{+}\left( 
\mathbb{R}_{*}\right) $.  In addition, $f((0,1))=0$ and it satisfies the tail estimate
\begin{equation}
\int_{[y,\infty)  } x^{\gamma+p}f\left(dx\right)
= \int_{[y,\infty)  } F\left(dx\right) 
\leq C y^{-\frac{1-\gamma-2p}{2}}\,, \quad y\ge 1\,. \label{ftailboundneggamma} 
\end{equation}

It remains to consider the case $\gamma >0$.  
Then (\ref{A1}) implies that  
\begin{align}
& \frac{1}{z}\int_{\left[  \frac{2z}{3},z\right]  }f_{\varepsilon}\left(
dx\right)  \leq  
\tilde{C} z^{-\frac{(\gamma+3)}{2}} \,, \quad 1\le z\le \varepsilon^{-\frac{1}{\gamma}}\,, \label{eq:fvepbounda} \\
&\frac{1}{z}\int_{\left[  \frac{2z}{3},z\right]  }f_{\varepsilon}\left(
dx\right)  \leq  \frac{\tilde{C} \sqrt{\varepsilon}}{z^{\frac{3}{2}}}\,,\quad z> \varepsilon^{-\frac{1}{\gamma}}\,.
\end{align}
Using these bounds in item \ref{it:Risinf} of Lemma \ref{lem:bound} implies then that
for all $y\ge 1$,
\begin{equation}
\int_{[y,\infty)  } f_{\varepsilon}\left(
dx\right)  \leq C \left( y^{-\frac{1+\gamma}{2}}+
 \left(\frac{\varepsilon}{y}\right)^{\frac{1}{2}}\right)\, \label{A1upgradeposgamma} 
\end{equation}
where the constant $C$ does not depend on $\varepsilon$.
Hence, in this case the family of measures $\left\{
f_{\varepsilon }\right\} _{\varepsilon >0}$ 
is contained in a $\ast$-weak compact set 
in $\mathcal{M}_{+,b}\left( 
\mathbb{R}_{*}\right)$. Therefore, there exists $f\in \mathcal{M}_{+,b}\left( 
\mathbb{R}_{*}\right) $ such that%
\begin{equation}
f_{\varepsilon _{n}}\rightharpoonup f\text{ as }n\rightarrow \infty \text{
in the }\ast -\text{weak topology}  \label{fepsWeakLimit}
\end{equation}%
for some sequence $\left( \varepsilon _{n}\right) _{n\in \mathbb{N}}$
with $\lim_{n\rightarrow \infty }\varepsilon _{n}=0.$

{To obtain better tail bounds for the limit measure,
let us first observe that by (\ref{eq:fvepbounda}), there is a constant $C$ such that for all $\vep$
\[
 \frac{1}{z}\int_{\left[  \frac{2z}{3},z\right]  } x^{\gamma+p}f_{\varepsilon}\left(
dx\right)  \leq   C z^{\frac{\gamma+2p-1}{2}-1}\,, \quad 1\le z\le \varepsilon^{-\frac{1}{\gamma}}\,.
\]
Therefore, applying item \ref{it:polcase} of Lemma \ref{lem:bound} with $r=\frac{1}{2}$, and using the assumption $\gamma+2 p <1$,
we can adjust the constant $C$ so that 
\begin{align}\label{eq:fvepafinbound}
 \int_{[a,\varepsilon^{-\frac{1}{\gamma}}]  } x^{\gamma+p}f_{\varepsilon}\left(
dx\right)  \leq   C a^{-\frac{1-(\gamma+2p)}{2}}\,, \quad 1\le a\le \frac{1}{2}\varepsilon^{-\frac{1}{\gamma}}\,.
\end{align}

Let then $y,R\ge 1$ be such that $y<R$ but they are otherwise arbitrary.  We choose a test function $\varphi\in C_c(\R_*)$,
such that $0\leq \varphi \leq 1$, $\varphi(x)=1$ for $y\leq x\leq R$, and $\varphi(x)=0$ for $x\geq 2R$ and for $x\le \frac{1}{2}
y$.
Then, if also $\vep\le (2R)^{-\gamma}$, we have
\[
 \int_{\R_*}\varphi(x)  x^{\gamma+p}f_{\varepsilon}\left(
dx\right) \le  \int_{[\frac{1}{2}y,2R]} x^{\gamma+p}f_{\varepsilon}\left(dx\right)
\le C 2^{\frac{1-(\gamma+2p)}{2}}y^{-\frac{1-(\gamma+2p)}{2}} \,,
\]
where for values $y\le 2$ the estimate follows by using $f_{\varepsilon}((0,1))=0$ and then $a=1$ in (\ref{eq:fvepafinbound}).  Applying this with $\vep=\vep_n$ and then taking $n\to\infty$ proves that 
\[
  \int_{[y,R]} x^{\gamma+p}f\left(dx\right)
\le
 \int_{\R_*}\varphi(x)  x^{\gamma+p}f\left(
dx\right) \le C 2^{\frac{1-(\gamma+2p)}{2}}y^{-\frac{1-(\gamma+2p)}{2}} \,.
\]
Here we may take $R\to \infty$, and using monotone convergence theorem we can conclude that $f$ satisfies a tail estimate identical to the earlier case with $\gamma\le 0$, namely, also for $\gamma>0$ we can find a constant $C$ such that
\begin{equation}
\int_{[y,\infty)  } x^{\gamma+p}f\left(dx\right)
\leq C y^{-\frac{1-\gamma-2p}{2}}\,, \quad y\ge 1\,. \label{ftailboundposgamma} 
\end{equation}

}

It only remains to take the limit $\varepsilon _{n}\rightarrow 0$ in (\ref%
{WeakFormEps}). Suppose that $\varphi \in C_{c}\left( \mathbb{R}_{*}\right)
. $   Then, in the term containing $\varphi (x+y)$ we have that the integrand
is different from zero only in a bounded region. Using then that for any $q\in \R$ we have $
\lim_{\varepsilon \rightarrow 0} (x y)^{q} K_{\varepsilon }\left( x,y\right) = (x y)^{q} K\left(
x,y\right) $ uniformly in compact subsets of $\R_*$, as well as \eqref{fepsWeakLimitneggamma} and (\ref{fepsWeakLimit}), we
obtain that the limit of that term is%
\begin{equation*}
\int_{(0,\infty )^{2}}K\left( x,y\right) \varphi (x+y)f\left(
dx\right) f\left( dy\right).
\end{equation*}

The terms containing $\varphi \left( x\right) $ or $\varphi \left( y\right) $
can be treated analogously due to the symmetry under the transformation $%
x\leftrightarrow y.$ We then consider the limit of the term containing $%
\varphi \left( x\right) $ where $\varphi \in C_{c}\left( \mathbb{R}_{*}\right) .$  Our goal is to show that the contribution to the integral due to regions $\{y\geq M\}$ where $M$ is very large, can be made arbitrarily small as $M\to\infty$, uniformly in $\varepsilon.$ 
Suppose that $M$ is chosen sufficiently large, so that the support of $\varphi$ is contained in $\left(  0,M\right) $.  We then have
the following identity:
\begin{align*}
& \int_{\mathbb{R}_{*}^{2}\cap\left\{  y\geq M\right\}  }K_{\varepsilon
}\left(  x,y\right)  \varphi\left(  x\right)  f_{\varepsilon}\left(
dx\right)  f_{\varepsilon}\left(  dy\right)  \\
& =\int_{\mathbb{R}_{*}^{2}\cap\left\{  y\geq M\right\}  }\left[  \min\left\{
\left(  x+y\right)  ^{\gamma},\frac{1}{\varepsilon}\right\}  \Phi
_{\varepsilon}\left(  \frac{x}{x+y}, x\right)  +\varepsilon\right]
\varphi\left(  x\right)  f_{\varepsilon}\left(  dx\right)  f_{\varepsilon
}\left(  dy\right) \, .
\end{align*}

Given that only values with $x\geq 1$ may contribute, and $x$ is in a bounded region contained in $\left[
1,M\right]  $, we obtain, using \eqref{A1}, an estimate
\begin{align}
&  \int_{\mathbb{R}_{*}^{2}\cap\left\{  y\geq M\right\}  }K_{\varepsilon
}\left(  x,y\right)  \varphi\left(  x\right)  f_{\varepsilon}\left(
dx\right)  f_{\varepsilon}\left(  dy\right) \nonumber \\
&  \quad \leq C \sup_{x\in \supp \varphi}
\int_{\left\{  y\geq M\right\}  }
\min\left\{
\left(  x+y\right)  ^{\gamma},\frac{1}{\varepsilon}\right\}  \Phi
_{\varepsilon}\left(  \frac{x}{x+y}, x\right) 
f_{\varepsilon
}\left(  dy\right)
+C\varepsilon\int_{\left\{
y\geq M\right\}  }f_{\varepsilon}\left(  dy\right)\,.  \label{eq:EstA3}
\end{align}
Using (\ref{A2}) we can bound the second term uniformly,
\[
C\varepsilon\int_{\left\{  y\geq M\right\}  }f_{\varepsilon}\left(  dy\right)
\leq C\sqrt{\varepsilon}\int_{\left\{  y\geq M\right\}  }\frac
{dy}{y^{\frac{3}{2}}}\leq\frac{C\sqrt{\varepsilon}}{M^{\frac{1}{2}}}\,,%
\]
where the constant $C$ is always independent of $\varepsilon$, although it might need to be adjusted at each inequality.
Therefore, the second term on the right hand side of \eqref{eq:EstA3} tends to zero as $\varepsilon\rightarrow0.$

In order to estimate the first term we need to consider separately different ranges of the values of the exponents $p$
 and $\gamma$. 
 We claim that for $1\leq x  \leq  C_0 \leq M,$ $y\geq M$, where $\supp \varphi \subset [0,C_0]$, the following estimates hold for some constants $C$, $C_*>0$ which do not depend on $\varepsilon$: 
\begin{enumerate}
\item If $\gamma\leq0$ and $p\leq0$ we have 
\begin{equation}
\min\left\{  \left(  x+y\right)  ^{\gamma},\frac{1}{\varepsilon}\right\}
\Phi_{\varepsilon}\left(  \frac{x}{x+y}, x\right)  \leq C  
\, . 
\label{C1}%
\end{equation}
\item  If $\gamma>0$ and $p\leq0$ we have 
\begin{align}
& \min\left\{  \left(  x+y\right)  ^{\gamma},\frac{1}{\varepsilon}\right\}
\Phi_{\varepsilon}\left(  \frac{x}{x+y}, x\right)
\nonumber \\ & \quad 
\leq C\left(  y^{\gamma
+\lambda}+y^{-\lambda}\right)  \chi_{\left\{  y\le \left(  \frac
{1}{\varepsilon}\right)  ^{\frac{1}{\gamma}}\right\}  }+\frac{C}{\varepsilon
}\left(  y^{\lambda}+y^{-\gamma-\lambda}\right)  \chi_{\left\{  y>\left(
\frac{1}{\varepsilon}\right)  ^{\frac{1}{\gamma}}\right\}  }\, ,
\label{C2}%
\end{align}
where $\chi_{U}$ is the characteristic function of the set $U.$
\item If $\gamma\leq0$ and $p>0$ we have
\begin{equation}
\min\left\{  \left(  x+y\right)  ^{\gamma},\frac{1}{\varepsilon}\right\}
\Phi_{\varepsilon}\left(  \frac{x}{x+y}, x\right)  \leq C\left(  y^{\gamma
+\lambda}+y^{-\lambda}\right)  \chi_{\left\{  y\leq C_{\ast}
\varepsilon  ^{-\frac{\sigma}{p}}\right\}  }\,.
\label{C3}%
\end{equation}
\item If $\gamma>0$ and $p>0$ we have 
\begin{equation}
\min\left\{  \left(  x+y\right)  ^{\gamma},\frac{1}{\varepsilon}\right\}
\Phi_{\varepsilon}\left(  \frac{x}{x+y}, x\right)  \leq C\left(  y^{\gamma
+\lambda}+y^{-\lambda}\right)  \chi_{\left\{  y\leq C_{\ast}
\varepsilon  ^{-\frac{\sigma}{p}}\right\}  }\, .
\label{C4}%
\end{equation}
\end{enumerate}

\bigskip

{
In the case 1 we use the fact that, since $p\leq0,$ we have $\sigma=0.$ Then \eqref{B2in3} implies that $\Phi_{\varepsilon}\left(  s,x\right)  \leq C.$ On the other hand, using that $\gamma\leq0$ and $x\geq1,\ y\geq M$ we have $\min\left\{  \left(  x+y\right)  ^{\gamma},\frac{1}{\varepsilon}\right\}\leq1$ whence \eqref{C1} follows.

In the case 2, we use the fact that since $p\leq0$ we have $\sigma=0.$
Moreover, since $x\leq M$ and $y\geq M$ we have that $s=\frac{x}{x+y}\leq
\frac{1}{2}.$ Given that $p=\max\left\{  \lambda,-\left(  \gamma
+\lambda\right)  \right\}  \leq0$ we have $\lambda\leq0,$ $\left(
\gamma+\lambda\right)  \geq0.$ Then \eqref{B2in3} implies that $\Phi_{\varepsilon
}\left(  s,x\right)  \leq C\left(  s^{\gamma+\lambda}+s^{-\lambda}\right)  $.
Using then that $y\leq\left(  x+y\right)  \leq2y$ as well as $1\leq x\leq
C_{0}\leq M\leq y$ we obtain $\Phi_{\varepsilon}\left(  s,x\right)  \leq
C\left(  \frac{x^{\gamma+\lambda}}{\left(  y\right)  ^{\gamma+\lambda}}%
+\frac{x^{-\lambda}}{\left(  y\right)  ^{-\lambda}}\right)  \leq C\left(
y^{\lambda}+y^{-\gamma-\lambda}\right)  .$ On the other hand, in order to
estimate $\min\left\{  \left(  x+y\right)  ^{\gamma},\frac{1}{\varepsilon
}\right\}  $ we use the fact that since $1\leq x\leq C_{0}\leq M\leq y$ we
have $\min\left\{  \left(  x+y\right)  ^{\gamma},\frac{1}{\varepsilon
}\right\}  \leq C\min\left\{  y^{\gamma},\frac{1}{\varepsilon}\right\}  .$
Considering separately the cases $y\leq\left(  \frac{1}{\varepsilon}\right)
^{\frac{1}{\gamma}}$ and $y>\left(  \frac{1}{\varepsilon}\right)  ^{\frac
{1}{\gamma}}$ we obtain 
$$\min\left\{  y^{\gamma},\frac{1}{\varepsilon
}\right\}  \leq C\left(  y^{\gamma}\chi_{\left\{  y\leq\left(  \frac
{1}{\varepsilon}\right)  ^{\frac{1}{\gamma}}\right\}  }+\frac{1}{\varepsilon
}\chi_{\left\{  y>\left(  \frac{1}{\varepsilon}\right)  ^{\frac{1}{\gamma}
}\right\}  }\right).$$
 Multiplying the estimates obtained for $\Phi
_{\varepsilon}\left(  s,x\right)  $ and $\min\left\{  \left(  x+y\right)
^{\gamma},\frac{1}{\varepsilon}\right\}  $ we derive \eqref{C2}.

In the cases 3 and 4 we use the fact that, since $p>0$ we have $0<\sigma
<\frac{p}{\gamma}.$ Using \eqref{B3in3} we obtain that $\Phi_{\varepsilon}\left(
s,x\right)  =0$ if $s\leq C \varepsilon  ^{\frac{\sigma}{p}}.$
Using that $\frac{x}{2y}\leq s=\frac{x}{x+y}\leq\frac{x}{y}$ and that $1\leq
x\leq C_{0}$ it follows that $\Phi_{\varepsilon}\left(  s,x\right)  =0$ for
$y>C_{\ast} \left( \frac{1}{  \varepsilon}\right) ^{\frac{\sigma}{p}}$
for some
$C_{\ast}>0.$ On the other hand \eqref{B3in3} as well as the fact that $s\leq\frac
{1}{2}$ implies also that $\Phi_{\varepsilon}\left(  s,x\right)  \leq\frac
{C}{s^{p}}.$ We then have%
\begin{equation}
\Phi_{\varepsilon}\left(  s,x\right)  \leq\frac{C}{s^{p}}\chi_{\left\{
y\leq C_{\ast} \left( \frac{1}{  \varepsilon}\right) ^{\frac{\sigma}{p}}\right\}\ .
}\label{IntStep}%
\end{equation}

We now remark that in the case 3, since $\gamma\leq 0$ and $y\geq1$ we have
$\left(  x+y\right)  ^{\gamma}\leq\frac{1}{\varepsilon}$ whence $\min\left\{
\left(  x+y\right)  ^{\gamma},\frac{1}{\varepsilon}\right\}  =\left(
x+y\right)  ^{\gamma}\leq y^{\gamma}.$ Combining this inequality with
(\ref{IntStep}) we then obtain
\begin{equation}
\Phi_{\varepsilon}\left(  s,x\right)  \min\left\{  \left(  x+y\right)
^{\gamma},\frac{1}{\varepsilon}\right\}  \leq\frac{C}{s^{p}}y^{\gamma}%
\chi_{\left\{  y\leq C_{\ast} \left( \frac{1}{  \varepsilon}\right) ^{\frac{\sigma}{p}}\right\}  }\ .\label{IntStep2}%
\end{equation}

In the case 4 we can derive a similar estimate. To this end we use the fact
that since $\sigma<\frac{p}{\gamma}$ it follows, since $\gamma>0,$ that
$\Phi_{\varepsilon}\left(  s,x\right)  =0$ if $y^{\gamma}\geq\frac
{1}{\varepsilon}$ and $\varepsilon$ is sufficiently small, because then
$y\geq\frac{1}{\varepsilon^{\frac{1}{\gamma}}}\geq C_{\ast} \left( \frac{1}{  \varepsilon}\right) ^{\frac{\sigma}{p}}.$
 Then $\Phi_{\varepsilon}\left(
s,x\right)  \min\left\{  \left(  x+y\right)  ^{\gamma},\frac{1}{\varepsilon
}\right\}  \leq C\Phi_{\varepsilon}\left(  s,x\right)  \min\left\{  y^{\gamma
},\frac{1}{\varepsilon}\right\}  \leq C\Phi_{\varepsilon}\left(  s,x\right)
y^{\gamma}\leq\frac{C}{s^{p}}y^{\gamma}\chi_{\left\{  y\leq C_{\ast} \left( \frac{1}{  \varepsilon}\right) ^{\frac{\sigma}{p}}\right\}  }$ which yields
the inequality (\ref{IntStep2}) that we had obtained also in the case 3. Using
then that $p=\max\left\{  \lambda,-\left(  \gamma+\lambda\right)  \right\}
>0$ and $y\geq1$ we obtain $y^{p}\leq y^{\lambda}+y^{-\gamma-\lambda}$ whence
both \eqref{C3} and \eqref{C4} follow.

}

We can now estimate the first term on the right hand side of \eqref{eq:EstA3} in all the cases.
{ We first observe that in the cases 1, 3 and 4 (cf. \eqref{C1}, \eqref{C3}, \eqref{C4}) the region, where the integrand is non-zero, is contained in the set $V_{\gamma
,\varepsilon,M}=\left\{  y\in\mathbb{R}_{\ast}: y \geq M,\ y^{\gamma}\leq\frac
{1}{\varepsilon}\right\}  .$ This follows immediately in the cases \eqref{C1},
\eqref{C3}, since in those cases $\gamma\leq0$ and then $y^{\gamma}\leq1\leq
\frac{1}{\varepsilon}.$ In the case of \eqref{C4} we remark that due to the
presence of the characteristic function on the right of \eqref{C4} the region is restricted to to the set $\left\{  1\leq y\leq C_{\ast}
\varepsilon  ^{-\frac{\sigma}{p}}\right\}  .$ Since in this case
$\gamma>0$ it follows that this set is the same as $\left\{  1\leq y^{\gamma
}\leq C_{\ast} \varepsilon  ^{-\frac{\sigma\gamma}{p}}\right\}
.$ Using then that $0<\sigma<\frac{p}{\gamma}$ it follows that $C_{\ast
} \varepsilon  ^{-\frac{\sigma\gamma}{p}}\leq\frac
{1}{\varepsilon}$ for $\varepsilon$ sufficiently small. Then, the region of
non-zero integrand is contained in $V_{\gamma,\varepsilon,M}$ also in this case, as
claimed. We now remark that for any $y\in V_{\gamma,\varepsilon,M}$ we have
$\min\left\{  y^{\gamma},\frac{1}{\varepsilon}\right\}  =y^{\gamma}.$ 
 Then \eqref{A1} implies that 
\begin{equation}\label{eq:estimate_V}
\frac{1}{y}\int_{\left[  \frac{2y}{3},y\right]  }f_{\varepsilon}\left(
dx\right)  \leq\frac{\tilde{C}}{y^{\frac{3+\gamma}{2}}}\,, \quad\text{if
} y \in V_{\gamma,\varepsilon,M}\ .
\end{equation}

We then obtain, using \eqref{C1}, \eqref{C3}, \eqref{C4}, the following estimate for the
first term on the right hand side of \eqref{eq:EstA3} in the cases 1, 3 and 4
\begin{equation}
\sup_{x\in\operatorname*{supp}\varphi}\int_{\left\{  y\geq M\right\}  }%
\min\left\{  \left(  x+y\right)  ^{\gamma},\frac{1}{\varepsilon}\right\}
\Phi_{\varepsilon}\left(  \frac{x}{x+y},x\right)  f_{\varepsilon}\left(
dy\right)  \leq C\int_{V_{\gamma,\varepsilon,M}}\left(  y^{\gamma+\lambda
}+y^{-\lambda}\right)  f_{\varepsilon}\left(  dy\right)\ .  \label{IntEstB1}%
\end{equation}
Notice that in the cases 3 and 4, estimate (\ref{IntEstB1}) follows from
\eqref{C3}, \eqref{C4}. In the case 1 we use that $p=\max\left\{  \lambda,-\left(
\gamma+\lambda\right)  \right\}  \leq0.$ Then $\left(  \gamma+\lambda\right)
\geq0$ and $-\lambda\geq0$ and we can use then \eqref{C1} to show that
$\min\left\{  \left(  x+y\right)  ^{\gamma},\frac{1}{\varepsilon}\right\}
\Phi_{\varepsilon}\left(  \frac{x}{x+y},x\right)  \leq C\left(  y^{\gamma
+\lambda}+y^{-\lambda}\right)  $ in the region of integration. Therefore
(\ref{IntEstB1}) follows also in this case.

We can now combine \eqref{eq:estimate_V} and (\ref{IntEstB1}) with Lemma \ref{lem:bound} to obtain
\begin{align*}
& \sup_{x\in\operatorname*{supp}\varphi}\int_{\left\{  y\geq M\right\}  }%
\min\left\{  \left(  x+y\right)  ^{\gamma},\frac{1}{\varepsilon}\right\}
\Phi_{\varepsilon}\left(  \frac{x}{x+y},x\right)  f_{\varepsilon}\left(
dy\right)  \\
& \leq C\int_{  V_{\gamma,\varepsilon,M} }\left(  \frac{y^{\gamma+\lambda
}+y^{-\lambda}}{y^{\frac{3+\gamma}{2}}}\right)  dy
\leq C\int_{\left\{  y\geq
M\right\}  }\left(  y^{\frac{\gamma}{2}+\lambda-\frac{3}{2}}+y^{-\left(
\lambda+\frac{\gamma}{2}\right)  -\frac{3}{2}}\right)  dy
\leq\frac{C}{M^{b}}%
\end{align*}
with $b>0$ since $\left\vert \gamma + 2\lambda\right\vert <1.$
Thus this term can be made arbitrarily small by taking $M\rightarrow\infty.$
}

It only remains to examine in detail the case (\ref{C2}). In this case we
obtain:%
\begin{align*}
& \int_{\left\{  y\geq M\right\}  }\min\left\{  \left(  x+y\right)  ^{\gamma
},\frac{1}{\varepsilon}\right\}  \Phi_{\varepsilon}\left(  \frac{x}%
{x+y}, x\right)  f_{\varepsilon}\left(  dy\right)  \\
& \leq C\int_{\left\{  y\geq M\right\}  }\frac{\left(  y^{\gamma+\lambda
}+y^{-\lambda}\right)  }{y^{\frac{3+\gamma}{2}}}dy+\frac{C}{\varepsilon}%
\int_{\left\{  y\geq\left(  \frac{1}{\varepsilon}\right)  ^{\frac{1}{\gamma}%
}\right\}  }\frac{\left(  y^{\lambda}+y^{-\gamma-\lambda}\right)  }%
{y^{\frac{3+\gamma}{2}}}dy\ .
\end{align*}

The first integral can be estimated as $\frac{C}{M^{b}}$ with $b>0$ arguing as
before (using $\left\vert \gamma+2\lambda\right\vert <1$). It only remains to
estimate the last integral. We have $p=\max\left\{  \lambda,-\left(
\gamma+\lambda\right)  \right\}  \leq0,$ whence $\lambda\leq0$ and $-\left(
\gamma+\lambda\right)  \leq0.$ In this case we have also $\gamma>0.$
Hence, the second integral converges and it can be estimated as 
\begin{align*}
& \frac{C}{\varepsilon}\int_{\left\{  y\geq\left(  \frac{1}{\varepsilon
}\right)  ^{\frac{1}{\gamma}}\right\}  }\frac{\left(  y^{\lambda}%
+y^{-\gamma-\lambda}\right)  }{y^{\frac{3+\gamma}{2}}}dy 
\leq\frac{C}{\varepsilon}\left[  \left(  \left(  \frac{1}{\varepsilon
}\right)  ^{\frac{1}{\gamma}}\right)  ^{\lambda-\frac{1+\gamma}{2}}+\left(
\left(  \frac{1}{\varepsilon}\right)  ^{\frac{1}{\gamma}}\right)
^{-\gamma-\lambda-\frac{1+\gamma}{2}}\right]  \\
& =\frac{C}{\varepsilon}\left[  \left(  \varepsilon\right)  ^{\frac{1+\gamma
}{2\gamma}-\frac{\lambda}{\gamma}}+\left(  \varepsilon\right)  ^{1+\frac
{\lambda}{\gamma}+\frac{1+\gamma}{2\gamma}}\right]  =C\left[  \left(
\varepsilon\right)  ^{\frac{1+\gamma}{2\gamma}-\frac{\lambda}{\gamma}%
-1}+\left(  \varepsilon\right)  ^{\frac{\lambda}{\gamma}+\frac{1+\gamma
}{2\gamma}}\right]  \\
& =C\left[  \left(  \varepsilon\right)  ^{\frac{1}{2\gamma
}\left(  1-2\lambda-\gamma\right)}+\left(  \varepsilon\right)  ^{\frac{1}{2\gamma
}\left(  1+2\lambda+\gamma\right)}\right].
\end{align*}
Thus the integral converges to zero as $\varepsilon\rightarrow0$ since $\vert \gamma +2\lambda \vert <1$.

Therefore, we
can take the limit $\varepsilon _{n}\rightarrow 0$ as $n\rightarrow \infty $ 
in (\ref{WeakFormEps}) with an arbitrary large $M$.
Then $M\to \infty$ can be taken by the assumed bounds on $K$ 
and using the tail estimates (\ref{ftailboundneggamma}) or (\ref{ftailboundposgamma}).
This yields
\begin{equation}\label{eq:weaksoldef}
\int_{( 0,\infty )^{2}}K\left( x,y\right) [\varphi (x+y)-\varphi
(x)-\varphi (y)]f\left( dx\right) f\left( dy\right) +\int_{(0,\infty
)}\varphi (x)\eta \left( dx\right) =0\,,
\end{equation}%
for any $\varphi \in C_{c}\left( \mathbb{R}_{*}\right) .$ In particular, $f\neq 0$ due to $\eta \neq 0$.
Taking
the limit of (\ref{A1}) as $\varepsilon \rightarrow 0$ we arrive at%
\begin{equation*}
\frac{1}{z}\int_{[2z/3,z]}f(dx)\leq \frac{\widetilde C}{z^{3/2+\gamma /2}}\ \ \text{ for
all }z\in (0,\infty ),
\end{equation*}
which implies 
\begin{equation*}
\frac{1}{z}\int_{[2z/3,z]}x^{\mu} f(dx)\leq \widetilde C \frac{z^{\mu} }{z^{3/2+\gamma /2}}\ \ \text{ for
all }z\in (0,\infty ),
\end{equation*}
for any $\mu \in \R$.
From Lemma~\ref{lem:bound} we obtain the boundedness of the  moment of order $\mu$:
\begin{equation}\label{eq:moment_mu}
\int_{(0,\infty)}x^{\mu} f(dx)< \infty
\end{equation}
 for any $\mu$ satisfying $\mu < \frac{\gamma + 1}{2}$. In particular, since $|\gamma + 2\lambda| <1$, then the moments  $\mu = -\lambda$ and $\mu = \gamma + \lambda$ are bounded, which proves (\ref{eq:moment_cond}).
\end{proofof}

\bigskip 

\bigskip

\section{Nonexistence result: Continuous model}\label{sec:nonexistence}

{

The rationale behind the proof of Theorem \ref{thm:NonExistence} is the following. The solutions
of \eqref{eq:contStat} satisfy \eqref{eq:flux_J0} for large values of $x$ with $J\left(  x;f\right)  $ as
in \eqref{eq:flux} and $J_{0}=\int x\eta\left(  dx\right)  $. A detailed analysis of the
contributions to the integrand of the different regions, using also the
assumption \eqref{eq:moment_cond}, that is the minimal assumption required to define a
solution of \eqref{eq:contStat}, shows that $J\left(  x;f\right)  $ can be approximated for
large values of $x$ as%
\begin{equation}
\int\int_{\left\{  y+z>x,\ y\leq x\right\}  \cap\left\{  z\leq\delta
y\right\}  }yK\left(  y,z\right)  f\left(  y\right)  f\left(  z\right)
dydz\simeq J_{0}\label{ApproxNonEx}%
\end{equation}
where $\delta>0$ can be chosen arbitrarily small. By assumption $K\left(
y,z\right)  \approx   y  ^{\gamma+\lambda}z^{-\lambda}+
z  ^{\gamma+\lambda}y^{-\lambda}.$ Suppose that $\gamma+\lambda
\geq0>-\lambda,$ since the other ranges of exponents can be studied with
slight modifications of the arguments. Notice that the assumption $\left\vert
\gamma+2\lambda\right\vert \geq1$ then implies, since $\gamma+2\lambda\geq0,$
that $\gamma+2\lambda\geq1.$ Then $\gamma+\lambda\geq1-\lambda$ and \eqref{eq:moment_cond}
implies that%
\begin{equation}
\int_{1}^{\infty}f\left(  z\right)  z^{1-\lambda}dz<\infty\ .\label{IntBound}%
\end{equation}

Moreover, we can approximate (\ref{ApproxNonEx}), using the form of the region
of integration, as%
\begin{equation}
\left(  x\right)  ^{\gamma+\lambda+1}\int\int_{\left\{  y+z>x,\ y\leq
x\right\}  \cap\left\{  z\leq\delta y\right\}  }f\left(  y\right)  f\left(
z\right)  z^{-\lambda}dydz\simeq J_{0}\ .\label{ApproxFlux}%
\end{equation}

We define $F\left(  x\right)  =\int_{x}^{\infty}f\left(  y\right)  dy$. This
integral is well defined due to \eqref{eq:moment_cond} and the fact that $\gamma+\lambda
\geq0.$ We can then approximate (\ref{ApproxFlux}) for large values of $x$, as%
\begin{equation}
\int_{1}^{\delta x}\left[  F\left(  x-z\right)  -F\left(  x\right)  \right]
f\left(  z\right)  z^{-\lambda}dz\simeq\frac{J_{0}}{\left(  x\right)
^{\gamma+\lambda+1}}\ .\label{NonLocEqu}%
\end{equation}

The equation (\ref{NonLocEqu}) can be thought as a nonlocal differential
equation. Due to (\ref{IntBound}) we can approximate formally (\ref{NonLocEqu}%
) for large values of $x$ as%
\begin{equation}
-\frac{dF}{dx}\simeq\frac{J_{0}}{\int_{1}^{\infty}f\left(  z\right)
z^{1-\lambda}dz}\frac{1}{\left(  x\right)  ^{\gamma+\lambda+1}}\ .\label{DE}%
\end{equation}
Therefore, using the definition of $F$ we formally obtain that $f\left(
x\right)  \simeq\frac{C}{\left(  x\right)  ^{\gamma+\lambda+1}}$ as
$x\rightarrow\infty.$ However, this implies that $\int_{1}^{\infty}%
x^{\gamma+\lambda}f\left(  x\right)  dx=\infty$ which contradicts \eqref{eq:moment_cond}. This
argument is formal and instead of approximating $K\left(  y,z\right)  $ by
means of power laws we must use the inequalities \eqref{eq:cond_kernel2}, \eqref{eq:cond_kernel3}. The solutions
of (\ref{NonLocEqu}) can be estimated in terms of the solutions of (\ref{DE})
by means of maximum principle arguments which are described in the following Lemmas.

}

\bigskip
\begin{lemma}
\label{lem:F+estimate}
 Let  $a$ and $b$ be constants satisfying $a\geq 0$ and  $(a-b)\geq 1$. Let $F: \R_* \to \R$ be a right-continuous non-increasing function satisfying $F(R)\geq 0$, for all $R> 0$. Assume that $f \in \mathcal{M}_+(\R_*)$ satisfies $f([1,\infty))> 0$ and
\begin{equation}\label{eq:momentBound1}
\int_{[1,\infty)} x^a f(dx) < \infty \,.
\end{equation}
There exists $\delta_0\in(0,1)$ which depends only on $a$ such that the following result holds: 

If $0<\delta\le  \delta_0$, $R_0> 1/\delta $, and $C>0$ are such that 
\begin{equation}
-\int_{\left[ 1,\delta R\right] }\left[ F\left( R-y\right) -F\left( R\right) %
\right] y^{b}f\left( dy\right) \leq -\frac{C}{R^{a +1}}\,,\quad \text{for }R\geq R_{0},\label{S4E9_1}
\end{equation}
then there are $R_0' \geq R_0$ and $B>0$ which depend only on $a$, $f$, $\delta$, $R_0$, and $C$, such that
if $a>0$ then 
\begin{equation}
F\left( R\right) \geq \frac{B}{R^{a}}\, , \quad \text{for \ }
R\geq  R_{0}', \label{S5E3_1}
\end{equation}
else, if $a=0$, then 
\begin{equation}
F\left( R\right) \geq B \log(R)\, , \quad \text{for \ }
R\geq  R_{0}'.\label{S5E3_12}
\end{equation}
\end{lemma}

\begin{proof}
Since $F$ is non-increasing and right-continuous, we have
\begin{equation}
F\left( R^{-}\right) =\lim_{\rho \rightarrow R^{-}}F\left( \rho \right) \geq
\lim_{\rho \rightarrow R^{+}}F\left( \rho \right) ={F\left( R\right)}\,.
\label{S4E7b_1}
\end{equation}
For the proof, let us first point out that 
we can increase $R_0$ while keeping $\delta$ and $C$ fixed if needed.

We first consider the case of $a>0$, and prove that in this case
the choice $\delta_0 := 1-(3/4)^{1/(1+a)}\in (0,1)$ will suffice.
From now on, we assume that $\delta$ is fixed to a value such that $0<\delta\le \delta_0$. 

We use a comparison argument. To this end, we construct an auxiliary
function 
$$F_{\ast}\left( R\right) =\frac{2B}{R^{a}}$$
 with $B>0$
to be determined. We choose $B$ in order to have
\begin{equation}
-\int_{\left[ 1,\delta R\right] }\left[ F_{\ast}\left( R-y\right)
-F_{\ast}\left( R\right) \right] y^{b}f\left( dy\right) \geq -\frac{C}{%
R^{a+1}}\ \ \text{for }R\geq R_{0}\ .  \label{S5E1_1}
\end{equation}
Therefore, the goal is to impose
\begin{equation}
-\int_{\left[ 1,\delta R\right] }\left[ \frac{2B}{\left( R-y\right)
^{a}}-\frac{2B}{R^{a}}\right] y^{b}f\left(
dy\right) \geq-\frac{C}{R^{a+1}}\ \ \text{for }%
R\geq R_{0} \,. \label{S4E8_1}
\end{equation}

Since $(1-\delta)^{a+1}\ge \frac{1}{2}$, we have for any $R\ge R_0>1/\delta$ and $y\in\left[ 1,\delta R\right] $ \begin{equation*}
\frac{1}{\left( R-y\right) ^{a}}-\frac{1}{R^{a}}
\leq\frac{2 a y}{R^{a+1}} \,.
\end{equation*}
Thus,
\begin{equation*}
-\int_{\left[ 1,\delta R\right] }\left[ \frac{2B}{\left( R-y\right)
^{a}}-\frac{2B}{R^{a}}\right] y^{b}f\left(
dy\right) \geq-\frac{4aB}{R^{a+1}}%
\int_{\left[ 1,\delta R\right] }y^{1+b}f\left( dy\right).
\end{equation*}
On the other hand, then  
$$
\int_{\left[ 1,\delta R\right] }y^{1+b}f\left( dy\right) \leq D,$$
where $D=\int_{\left[ 1,\infty\right) }y^{1+b}f\left( dy\right)$ is a well-defined, strictly positive constant due to $b+1\le a$,~\eqref{eq:momentBound1} and $f\neq0.$ Therefore, choosing 
$$B=\frac{C }{4D a  },$$
 we obtain that (\ref{S4E8_1}) holds.

For the next step, we require that $f([1,\delta R_0])>0$.
If needed, this can be accomplished by increasing $R_0$ since 
the left hand side, by dominated convergence theorem, 
approaches $f([1,\infty))>0$, as $R_0\to \infty$.
  
To prove~\eqref{S5E3_1}, we argue by contradiction. Suppose that there exists $R_{1}\geq
R_{0}$ such that $F\left( R_{1}\right) <\frac{B}{\left(
R_{1}\right) ^{a }}.$ Then, using that $F\left( R\right) $ is
decreasing, we obtain that 
\begin{equation}
F\left( R\right) <\frac{B}{\left( R_{1}\right) ^{a }}\,, \quad \text{ for } R \in \left[ R_{1},\frac{R_{1}}{1-\delta }\right]\,.\label{S5E11_1}
\end{equation}
We define 
$$G\left( R\right) =F_{\ast}\left( R\right) -\frac{B}{2}\frac{1}{\left( R_{1}\right) ^{a }}-F\left( R\right).$$ 
Combining (\ref{S4E9_1}) and (\ref{S5E1_1}) we
obtain that%
\begin{equation}
-\int_{\left[ 1,\delta R\right] }\left[ G\left( R-y\right) -G\left( R\right) %
\right] y^{b }f\left( dy\right) \geq 0\ \ \text{for all }R\geq R_{0} \ .
\label{S5E2_1}
\end{equation}%
Using~\eqref{S5E11_1} we obtain
\begin{align}
G\left( R\right) & =F_{\ast }\left( R\right) -\frac{B}{2}\frac{1}{\left(
R_{1}\right) ^{a }}-F\left( R\right) >\frac{2B}{R^{a }}-\frac{B}{2}\frac{1}{\left( R_{1}\right) ^{a }}-\frac{B}{\left( R_{1}\right) ^{a }}  \notag \\
& \geq B\left( \frac{2\left( 1-\delta \right) ^{a }}{\left(
R_{1}\right) ^{a }}-\frac{3}{2}\frac{1}{\left( R_{1}\right)
^{a}}\right) >0, \ \ \ \ \ \ \ \ \text{    for    } R\in \left[ R_{1},\frac{R_{1}}{1-\delta }\right], \label{S5E2a_1}
\end{align}%
since $\delta>0$ is sufficiently small so that $\left( 1-\delta \right) ^{a }>\left( 1-\delta \right) ^{a+1}\ge \frac 3 4$. 
 Notice that since $%
F_{\ast }\left( R\right) $ and $\frac{B}{2}\frac{1}{\left( R_{1}\right)
^{a }}$ are continuous functions we have that $G$ is right continuous and  (\ref{S4E7b_1})
implies
\begin{equation}
G\left( R^{-}\right) =\lim_{\rho \rightarrow R^{-}}G\left( \rho \right)  \leq 
\lim_{\rho \rightarrow R^{+}}G\left( \rho \right) =G\left( R\right)\,.\label{S5E33_1}
\end{equation}
We define $R_{2}$ as
\begin{equation*}
R_{2}=\inf \left\{ \rho \geq R_{1}:G\left( \rho\right) \leq 0\right\} \,.
\end{equation*}

Suppose first that $R_{2}<\infty.$ 
By definition $G(R_2^+) \leq 0$. Since $G$ is right-continuous, then $G(R_2) \leq 0$.  
From~\eqref{S5E2a_1},  $G(R_2) \geq G(R_2^-) \geq 0$. Therefore, necessarily $G(R_2)=0$.
From~\eqref{S5E2a_1} we also have that $R_2>\frac{R_1}{1-\delta}$ and 
\begin{equation}\label{S5E34_1}
G(R)>0 \text{ for } R \in [R_1,R_2).
 \end{equation}
For $y \in [1,\delta R_2]$, we have that $(R_2-y) \in [R_1,R_2)$, therefore $G(R_2-y) >0$.  Since $f([1,\delta R_2]) \ge f([1,\delta R_0])>0$, this implies
\begin{equation*}
-\int_{\left[ 1,\delta R_{2}\right] }\left[ G\left( R_{2}-y\right) -G\left(
R_{2}\right) \right] y^{b}f\left( dy\right) <0
\end{equation*}%
which contradicts (\ref{S5E2_1}). Then $R_{2}=\infty $ whence $G\left(
R\right) \geq 0$ for all $R\geq R_{1}$. Therefore,
\begin{equation*}
F\left( R\right) \leq F_{\ast }\left( R\right) -\frac{B}{2}\frac{1}{\left(
R_{1}\right) ^{a }}\quad \text{ for }R\geq R_{1}\,.
\end{equation*}

However, this inequality implies that $F\left( R\right) <0$ for $R$ large
enough, but this contradicts the definition of $F$.  Therefore,
\begin{equation*}
F\left( R\right) \geq \frac{B}{R^{a}}\ \ \text{if \ }%
R\geq  R_{0},
\end{equation*}
which concludes the proof for $a>0$.  Note that $R_0$ in this formula might have been increased compared to the value in the original assumptions, hence it is denoted by $R'_0$ in the conclusions of the Lemma.

We now consider the case $a=0$. 
In this case, we prove that the choice $\delta_0 := \frac{1}{2}$ will suffice.
We assume that $\delta$ is fixed to a value such that $0<\delta\le  \delta_0$, 
and that $R_0$ is sufficiently large
so that $R_0>\frac{1}{1-\delta}$ and 
$f([1,\delta R_0])>0$, as before. 

We construct an auxiliary
function 
$$F_{\ast}\left( R\right) =- B \log(R)$$
 with $B>0$
to be determined by the requirement that
\begin{equation}
-\int_{\left[ 1,\delta R\right] }\left[ F_{\ast}\left( R-y\right)
-F_{\ast}\left( R\right) \right] y^{b}f\left( dy\right) \geq -\frac{C}{R}\ \ \text{for }R\geq R_{0} \ .  \label{S5E1_12}
\end{equation}
Therefore, we need to impose%
\begin{equation}
\int_{\left[ 1,\delta R\right] }\left[B\log(R-y)-B\log(R)\right] y^{b}f\left(
dy\right) \geq-\frac{C}{R}\ \ \text{for }%
R\geq R_{0}  \ . \label{S4E8_12}
\end{equation}
Since $0<\delta\le\frac{1}{2}$,
we have for all
$R>1/\delta$  and $y\in\left[ 1,\delta R\right] $ an estimate
\begin{equation*}
\log(R-y)-\log(R)
\geq - \frac{ 2y}{R}\,.
\end{equation*}
Thus,
\begin{equation*}
\int_{\left[ 1,\delta R\right] }\left[B\log(R-y)-B\log(R)\right] y^{b}f\left(
dy\right) \geq-\frac{2B}{R}\int_{\left[ 1,\delta R\right] }y^{1+b}f\left( dy\right).
\end{equation*}
Here, 
$$
\int_{\left[ 1,\delta R\right] }y^{1+b}f\left( dy\right) \leq D,$$
where $D=\int_{\left[ 1,\infty\right) }y^{1+b}f\left( dy\right)$ is a well-defined strictly positive constant due to $b+1\le a$,~\eqref{eq:momentBound1} and $f\neq0.$ 
Therefore, choosing 
$$B=\frac{C }{2D   },$$
 we obtain that (\ref{S4E8_12}) holds.

To prove~\eqref{S5E3_12}, we again argue by contradiction. Suppose that there exists $R_{1}\geq
R_{0}$ such that $F\left( R_{1}\right) <B\log(
R_{1}).$ Then, using that $F\left( R\right) $ is
decreasing, we obtain that 
\begin{equation}
F\left( R\right) <B\log(R_{1}) \text{ for } R \in \left[ R_{1},\frac{R_{1}}{1-\delta }\right].\label{S5E11_12}
\end{equation}
We define 
$$G\left( R\right) =F_{\ast}\left( R\right) +3B\log(
R_{1})-F\left( R\right).$$ 
Combining (\ref{S4E9_1}) and (\ref{S5E1_12}) we
obtain that%
\begin{equation}
-\int_{\left[ 1,\delta R\right] }\left[ G\left( R-y\right) -G\left( R\right) %
\right] y^{b }f\left( dy\right) \geq 0\ \ \text{for all }R\geq R_{0}\ .
\label{S5E2_12}
\end{equation}%
Using~\eqref{S5E11_12} we obtain
\begin{align}
G\left( R\right) & 
=F_{\ast }\left( R\right) +3 B\log(
R_{1}) -F\left( R\right) 
> - B\log(R)+3B\log(R_1)-B\log(R_1) \notag \\
& \geq B\left( -\log(\frac{R_1}{1-\delta})+2\log(R_1) \right) = B\log(R_1(1-\delta))>0, \quad \text{    for    } R\in \left[ R_{1},\frac{R_{1}}{1-\delta }\right]\,, \label{S5E2a_12}
\end{align}
where in the last step we used the property that $R_1\ge R_0>\frac{1}{1-\delta}$.
 Notice that since $F_{\ast }\left( R\right) $ and $3B\log( R_{1})$ are continuous functions we have that $G$ is right continuous and  (\ref{S4E7b_1})
implies
\begin{equation}
G\left( R^{-}\right) =\lim_{\rho \rightarrow R^{-}}G\left( \rho \right)  \leq 
\lim_{\rho \rightarrow R^{+}}G\left( \rho \right) =G\left( R^{+}\right) \ .\label{S5E33_12}
\end{equation}
We define $R_{2}$ as
\begin{equation*}
R_{2}=\inf \left\{ \rho \geq R_{1}:G\left( \rho\right) \leq 0\right\} \ .
\end{equation*}
Using the same reasoning as in the case $a>0$ we obtain that  
$R_{2}=\infty $, and thus $G\left(R\right) > 0$ for all $R\geq R_{1}.$ Therefore,%
\begin{equation*}
F\left( R\right) \leq F_{\ast }\left( R\right) + 3B\log(R_1)
\quad \text{ for }R\geq R_{1} \ .
\end{equation*}
However, this inequality implies that $F\left( R\right) <0$ for $R$ large
enough, but this contradicts the definition of $F$. Therefore,
\begin{equation*}
 F\left( R\right) \geq B\log(R)\,,\ \ \text{if \ }%
R\geq  R_{0} \ ,
\end{equation*}
which concludes the proof.
\end{proof}

\

\begin{lemma}
\label{lem:F-estimate}
Let  $a$ and $b$ be constants satisfying $a<0$ and  $(a-b)\geq 1$. Assume that $F: \R_* \to \R$ is a right-continuous non-decreasing function and $f \in \mathcal{M}_+(\R_*)$ 
satisfies $f([1,\infty))> 0$ and
\begin{equation}\label{eq:momentBound2}
\int_{[1,\infty)} x^a f(dx) < \infty \ .
\end{equation}
There exists $\delta_0\in(0,1)$ which depends only on $a$ such that the following result holds:

If $0<\delta\le  \delta_0$, $R_0> 1/\delta $, and $C>0$ are such that 
$F(R_0)>0$ and 
\begin{equation}
-\int_{\left[ 1,\delta R\right] }\left[ F\left( R-y\right) -F\left( R\right) 
\right] y^{b}f\left( dy\right) \geq  \frac{C}{R^{a +1}}\ \ \text{for }R\geq R_{0},\label{S4E9_2}
\end{equation}
then there are $R_0'\geq R_0$ and $B>0$ which only depend on $a$, $f$, $\delta$, $R_0$, and $C$, such that
\begin{equation}
F\left( R\right) \geq \frac{B}{R^{a}}\, , \quad \text{for \ }
R\geq  R_{0}'.\label{S5E3_2}
\end{equation}
\end{lemma}

\begin{proof}
Since $F$ is non-decreasing and right-continuous, we have
\begin{equation}
F\left( R^{-}\right) =\lim_{\rho \rightarrow R^{-}}F\left( \rho \right) \leq
\lim_{\rho \rightarrow R^{+}}F\left( \rho \right) =F\left( R^{+}\right) = F(R)\, .
\label{S4E7b_2}
\end{equation} 
We assume $\delta$ is fixed and satisfies $0<\delta\le \delta_0$. 
We will show that $\delta_0=\frac{1}{2}\in (0,1)$ works in this case. 
Note that if $R'_0\ge R_0$, then also $F(R'_0)\ge F(R_0)>0$ since $F$ is increasing.  Therefore, as in the previous proof, 
we can increase $R_0$ while keeping $\delta$ and $C$ fixed if needed.
In particular, we may assume that $f([1,\delta R_0])>0$, as before.

We again use a comparison argument. To this end, we construct an auxiliary function 
$$F_{\ast}\left( R\right) =\frac{A }{R^{a}},$$
where $A  > 0$ is a constant to be determined.  
We choose $A$ in order to have
\begin{equation}
-\int_{\left[ 1,\delta R\right] }\left[ F_{\ast}\left( R-y\right)
-F_{\ast}\left( R\right) \right] y^{b}f\left( dy\right) \leq \frac{C}{R^{a+1}}\ \ \text{for }R\geq R_{0} \ . \label{S5E1_2}
\end{equation}
Therefore, we need to impose
\begin{equation}
-\int_{\left[ 1,\delta R\right] }\left[ \frac{A  }{ (R-y)
^{a}}-\frac{A }{R^{a}}\right] y^{b}f\left(
dy\right) \leq  \frac{C}{R^{a+1}}\ \ \text{for }%
R\geq R_{0} \ . \label{S4E8_2}
\end{equation}

Let us next show that a constant $A$ for the above inequality may be found
for the values of $\delta$ considered here.
Since $0<\delta\le \delta_0=\frac{1}{2}$,
we have $(1-\delta)^{|a|-1}\le 2$ for $a\in [-1,0)$.
If $a<-1$, the function $x\mapsto x^{|a|-1}$ is increasing, and thus
Taylor's theorem implies
\begin{equation*}
{\left(\frac{1}{R^{a}}-\frac{1}{\left( R-y\right)^{a}}\right) \leq  2 \vert a \vert y \frac{1}{R^{a+1}}}
\end{equation*}
whenever $y\in\left[ 1,\delta R\right] $.  Thus,
\begin{equation*}
-\int_{\left[ 1,\delta R\right] }\left[ \frac{A  }{\left( R-y\right)^{a}}-\frac{A  }{R^a}\right] y^{b}f\left(
dy\right) \leq \frac{2 A |a|}{R^{a+1}}
\int_{\left[ 1,\delta R\right] }y^{1+b}f\left( dy\right)
\end{equation*}
for any $A>0$. For $R>1/\delta$ we obtain that 
$$
\int_{\left[ 1,\delta R\right] }y^{1+b}f\left( dy\right) \leq D,$$
where $D=\int_{\left[ 1,\infty\right) }y^{1+b}f\left( dy\right)$ is a well-defined positive constant due to $b+1\le a$,~\eqref{eq:momentBound2} and $f\neq0.$ Therefore, choosing 
\begin{equation}\label{eq:cdtA}
0<  A \leq \frac{  C }{D |a|  }
\end{equation}
 we obtain that (\ref{S4E8_2}) holds.
 
Next we will prove~\eqref{S5E3_2}.
We define 
$$G\left( R\right) =F(R) - F_{\ast}\left( R\right).$$ 
Combining (\ref{S4E9_2}) and (\ref{S5E1_2}) we
obtain that
\begin{equation}
-\int_{\left[ 1,\delta R\right] }\left[ G\left( R-y\right) -G\left( R\right) %
\right] y^{b }f\left( dy\right) \geq 0\ \ \text{for all }R\geq R_{0}.
\label{S5E2_2}
\end{equation}
Since $F$ is increasing and $F(R_0)>0$, then $F(R)\geq F(R_0)>0$ for all $R\geq R_0$. Then $G(R)\geq F(R_0)-\frac{A}{R^{a+1}}$ for any $R\geq R_0$.  Therefore, choosing $A$ sufficiently small and satisfying also \eqref{eq:cdtA}, we obtain 
\begin{equation}
G(R)>0 \quad \text{for } R \in \left[R_0,\frac{R_0}{1-\delta}\right].\label{S5E2a_2}
\end{equation}
Since $F_{\ast }\left( R\right) $ is continuous,  we have that $G$ is right continuous and  (\ref{S4E7b_2})
implies
\begin{equation}
G\left( R^{-}\right) =\lim_{\rho \rightarrow R^{-}}G\left( \rho \right)  \leq 
\lim_{\rho \rightarrow R^{+}}G\left( \rho \right) ={G\left( R\right)} \ . \label{S5E33_2}
\end{equation}
We define $R_{2}$ as
\begin{equation*}
R_{2}=\inf \left\{ \rho \geq R_0: G\left( \rho\right) \leq 0\right\} \ .
\end{equation*}

Suppose first that $R_{2}<\infty.$ 
By definition, $G(R_2^+) \leq 0$. Since $G$ is right-continuous, then $G(R_2) \leq 0$.  
From~\eqref{S5E33_2}, $G(R_2) \geq G(R_2^-) \geq 0$. Therefore, necessarily $G(R_2)=0$.
From~\eqref{S5E2a_2} we also have that $R_2>\frac{R_0}{1-\delta}$ and 
\begin{equation}\label{S5E34_2}
G(R)>0 \text{ for } R \in [R_0,R_2).
 \end{equation}
For $y \in [1,\delta R_2]$, we have that $(R_2-y) \in [R_0,R_2)$, therefore $G(R_2-y) >0$. This implies
\begin{equation*}
-\int_{\left[ 1,\delta R_{2}\right] }\left[ G\left( R_{2}-y\right) -G\left(
R_{2}\right) \right] y^{-\lambda }f\left( dy\right) <0
\end{equation*}%
which contradicts (\ref{S5E2_2}). Then $R_{2}=\infty $ whence $G\left(
R\right) > 0$ for all $R\geq R_{0}.$
 Therefore,
\begin{equation*}
F\left( R\right) \geq  F_*(R) = \frac{A}{R^{a}}\ \ \text{for \ }%
R\geq  R_{0},
\end{equation*}
which proves~\eqref{S5E3_2} with $B=A$.
\end{proof}

\

\begin{proofof}[Proof of Theorem \ref{thm:NonExistence} (non-existence)]
We argue by contradiction. Suppose that $f\in \mathcal{M}_{+}\left( \mathbb{R%
}_*\right) $ satisfies $f\left( \left( 0,1\right) \right) =0$ as well as (%
\ref{eq:moment_cond}) and it is a stationary injection solution of %
\eqref{eq:time_evol} in the sense of Definition \ref{DefFluxSol}.
Then, from Lemma \ref{lem:flux} and using also that $f\left( \left( 0,1\right) \right) =0$ we obtain
\begin{equation}
-\int_{\left[ 1,R\right] }f\left( dx\right) \int_
{\left(R-x ,\infty \right) \cap \left[1,\infty\right)}
{K\left( x,y\right) x f\left( dy\right)} +\int_{ [1,R]}x\eta \left( dx\right) =0,\ R\geq 1 \ . \label{S4E5b}
\end{equation}

Then we introduce a function $J:\mathbb{R}_{*}\rightarrow \mathbb{R}_{+}$
defined by
\begin{equation}
J\left( R\right) =\iint_{\Sigma _{R}}K\left( x,y\right) xf\left(
dx\right) f\left( dy\right)  \label{S4E5a}
\end{equation}%
where 
\begin{equation*}
\Sigma _{R}=\left\{ x\geq 1,\ y\geq 1:x+y>R,\ x\leq R\right\}\ .
\end{equation*}%
We notice that the function $J$ is constant if $R\geq L_{\eta },$ i.e. 
\begin{equation}
J\left( R\right) =J\left( L_{\eta }\right) \text{ for }R\geq L_{\eta }.
\label{S4E5}
\end{equation}
Suppose that $\eta $ is different from zero. Then (\ref{S4E5b}) implies that 
$J\left( L_{\eta }\right) =\int_{\mathbb{R}_{+}}x\eta \left( dx\right) >0.$
If $(\gamma + 2\lambda) \geq 1$,  we define $a:= \gamma+\lambda$ and $b:= -\lambda$, else, if $(\gamma + 2\lambda) \leq -1$,  we define $a:= -\lambda$ and $b:= \gamma+\lambda$. The assumption $|\gamma + 2\lambda|\geq 1$ becomes $a-b\geq 1$ in both cases.
By~\eqref{eq:moment_cond} we have
\begin{equation}
\int_{\left[ 1,\infty \right) }x^{a }f\left( dx\right) <\infty \ .
\label{S4E6}
\end{equation}

We now prove that the main contribution to the integral $J(R)$ in (\ref{S4E5a}) as $%
R\rightarrow\infty$ is due to the portion of the region of integration where 
$x$ is close to $R$ and $y$ is order one. 
To this end, let us  consider
parameters $\delta$ which satisfy $0<\delta<\delta_0$ for the value $\delta_0=\delta_0(a)$ 
given by Lemma \ref{lem:F+estimate} if $a\ge 0$, or by Lemma \ref{lem:F-estimate} if $a<0$.
We then define the domains
\begin{align*}
D_{\delta}^{\left( 1\right) } & =\left\{ x\geq1,\ y\geq1:y\leq\delta x\right\} \ , \\
D_{\delta}^{\left( 2\right) } & =\left\{ x\geq1,\ y\geq1:y>\delta x\right\} \ .
\end{align*}
We then write
\begin{align*}
J\left( R\right) & =J_{1}\left( R\right) +J_{2}\left( R\right) \text{ \ \
with} \\
J_{k}\left( R\right) & =\iint_{\Sigma _{R}\cap D_{\delta }^{\left( k\right)
}}\left[ K\left( x,y\right) x\right] f\left( dx\right) f\left( dy\right) ,\
\ k=1,2 \ .
\end{align*}

We estimate first $J_{2}\left( R\right) $ for large values of $R.$ Using %
\eqref{eq:cond_kernel3} we obtain
\begin{equation*}
0\leq J_{2}\left( R\right) \leq c_{2}\iint_{\Sigma _{R}\cap D_{\delta
}^{\left( 2\right) }}\left( x^{a }y^{b }+y^{a }x^{b }\right) xf\left( dx\right) f\left( dy\right) \ .
\end{equation*}
Using that $\left( a-b \right) >0$ we obtain that in the region $%
D_{\delta}^{\left( 2\right) }$ we have 
$x^{a}y^{b}\leq \delta^{b-a}y^{a}x^{b}$. Therefore,
\begin{equation*}
J_{2}\left( R\right) \leq C_{\delta}\iint_{\Sigma_{R}\cap D_{\delta}^{\left(
2\right) }}\left( y^{a}x^{1+b}\right) f\left(
dx\right) f\left( dy\right) . 
\end{equation*}
Notice that $\Sigma_{R}\cap D_{\delta}^{\left( 2\right) }\subset\left[ 1,R%
\right] \times\left[ \frac{\delta R}{1+\delta},\infty\right) ,$ whence
\begin{equation*}
J_{2}\left( R\right) \leq C_{\delta}\int_{\left[ 1,R\right] }x^{1+b
}f\left( dx\right) \int_{\left[ \frac{\delta R}{1+\delta},\infty\right) 
}y^{a}f\left( dy\right) \ .
\end{equation*}

Given that $\left(a-b \right) \geq 1$ we obtain, taking into
account \eqref{S4E6},
\begin{equation*}
\int_{\left[ 1,R\right] }x^{1+b }f\left( dx\right) \leq \int_{\left[
1,\infty \right) }x^{a }f\left( dx\right) <\infty \ .
\end{equation*}
Moreover, using again (\ref{S4E6}), it follows that $$
\lim_{R\rightarrow\infty }\int_{\left[ \frac{\delta R}{1+\delta},\infty\right) }y^{a }f\left( dy\right) =0.$$
 This implies that
the contribution due to $J_{2}$ vanishes in the limit $R\to\infty$, namely 
\begin{equation*}
\lim_{R\rightarrow\infty}J_{2}\left( R\right) =0.
\end{equation*}
Therefore, (\ref{S4E5}) implies that 
\begin{equation*}
\lim_{R\rightarrow\infty}J_{1}\left( R\right) =J\left( L_\eta\right) \ .
\end{equation*}

For next step, let us remark that for $(x,y) \in\Sigma _{R}\cap D_{\delta }^{\left( 1\right) }$ we
have $x>R-y \geq R-\delta R$ and therefore $(1-\delta)R< x\leq R.$ 
In this region we have also $y^{a }x^{b
}\leq \delta ^{a-b }x^{a }y^{b }.$ Combining \eqref{eq:cond_kernel3} and using the above bounds for $x$, we obtain 
\begin{equation*}
K\left( x,y\right) x\leq c_3\left( 1+\delta ^{|a-b| }\right)
R^{a +1}y^{b }\ \ ,\ \ \ \left( x,y\right) \in \Sigma
_{R}\cap D_{\delta }^{\left( 1\right) }
\end{equation*}
where $c_3>0$ can be chosen independent of $\delta $ as soon as $\delta \leq \frac{1}{2}$ which we do in the following. Then
\begin{equation*}
\liminf_{R\rightarrow \infty }\left( R^{a +1}\iint_{\Sigma
_{R}\cap D_{\delta }^{\left( 1\right) }} y^{b }f\left( dx\right)
f\left( dy\right) \right) \geq \frac{J\left( L_\eta\right) }{c_3\left( 1+\delta
^{|a-b| }\right) } \ .
\end{equation*}

Notice that, if $R>1/\delta, 1/(1-\delta)$,
\begin{equation*}
\Sigma_{R}\cap D_{\delta}^{\left( 1\right) }\subset\left\{ \left( x,y\right)
:1\leq y\leq\delta R,\ R<x+y,\ 1 \leq  x\leq R\right\}
\end{equation*}
whence
\begin{equation}
\int_{\left[ 1,\delta R\right] }y^{b}f\left( dy\right) \int_{\left(
R-y,R\right] }f\left( dx\right) \geq\frac{J\left( L_\eta\right) }{2c_3\left(
1+\delta^{|a-b|}\right) }\frac{1}{R^{a+1}}
\label{S4E7}
\end{equation}
 for $R\geq R_{0}$ with $R_{0}$ large enough.

The rest of the proof is divided into two cases: $a\geq 0$ and $a<0$. 

Suppose first that  $a\geq 0$. Due to (\ref{S4E6})  we may define the function:
\begin{equation}
F\left( R\right) =\int_{\left( R,\infty\right) }f\left( dx\right) \ \ ,\ \
R\geq1 \ . \label{S4E7a}
\end{equation}
Note that  the function $R\rightarrow F\left( R\right) $ is right
continuous, i.e. $F\left( R\right) =F\left( R^+\right) = \lim_{\rho \rightarrow R^{+}}F\left(\rho \right).$
Moreover, $F$ is non-increasing and $F(R)\geq 0$, for all $R\geq 1$.
 Using (\ref{S4E7a}) we can rewrite (\ref{S4E7}) as: 
\begin{equation*}
-\int_{\left[ 1,\delta R\right] }\left[ F\left( R-y\right) -F\left( R\right) %
\right] y^{b}f\left( dy\right) \leq -\frac{J\left( L_\eta\right) }{%
2c_3\left( 1+\delta ^{|a-b| }\right) }\frac{1}{R^{a +1}}\ \ \text{for }R\geq R_{0}.
\end{equation*}
From Lemma~\ref{lem:F+estimate}, it then follows: 
\begin{equation}
F\left( R\right) \geq \frac{B}{R^{a}}\ \ \text{if \ }
R\geq R_{0},\ \text{ for } a>0,  \label{S5E3_a}
\end{equation}
and
\begin{equation}
F\left( R\right) \geq B \log(R)\ \ \text{if \ }
R\geq R_{0},\ \text{ for } a=0,  \label{S5E3_b}
\end{equation}
for some constant $B>0$.

In the case where $a>0$, we use~\eqref{S4E6} and~\eqref{S5E3_a} to obtain: 
\begin{eqnarray}
\int_{[1,\infty)} x^{a}f(dx) &=& \int_{[1,R]}x^{a}f(dx)+ \int_{(R,\infty)}x^{a}f(dx) \nonumber \\
&\geq &  \int_{[1,R]}x^{a}f(dx)+ R^a\int_{(R,\infty)}f(dx) \nonumber \\
&\geq &  \int_{[1,R]}x^{a}f(dx)+ B  \ .\nonumber 
\end{eqnarray}
 By taking the limit $R \to \infty$ we obtain that $B\leq 0$ which leads to a contradiction.
 
In the case where $a=0$,~\eqref{S5E3_b} yields 
\begin{eqnarray}
\int_{(R,\infty)} f(dx) &\geq & B \log(R). \nonumber
\end{eqnarray}
By taking the limit $R \to \infty$ we obtain using~\eqref{S4E6} that the left-hand side converges to zero, while the right hand-side diverges, which leads to a contradiction.

Suppose now that $a<0$. We define the function $F$ by
\begin{equation}
F\left( R\right) =\int_{[1,R] }f\left( dx\right) \ \ ,\ \
R\geq1.  \label{S4E7aa}
\end{equation}
The function $R\rightarrow F\left( R\right) $ is right
continuous and non-decreasing. Since $f \neq 0$, then $F(R)> 0$, for all $R\geq R_0$, for $R_0$ large enough.
Using (\ref{S4E7aa}) we can rewrite (\ref{S4E7}) as: 
\begin{equation}
-\int_{\left[ 1,\delta R\right] }\left[ F\left( R\right) -F\left( R-y\right) %
\right] y^{b }f\left( dy\right) \leq -\frac{J\left( L_\eta\right) }{%
2c_3\left( 1+\delta ^{|a-b| }\right) }\frac{1}{R^{a +1}}\ \ \text{for }R\geq R_{0}.\nonumber
\end{equation}
From Lemma~\ref{lem:F-estimate}, it follows that there are $B>0$ and $R'_0\ge R_0$ such that 
\begin{equation}
F\left( R\right) \geq \frac{B}{R^{a}}\ \ \text{if \ }
R\geq  R'_{0} \, .  \label{S5E3a}
\end{equation}

From~\eqref{S5E3a} it follows that for all $R>M$ we have
\begin{eqnarray*}
B \leq R^a \int_{[1,R]} f(dx)  \leq  R^a \int_{[1,M]} f(dx) + \int_{[M,R]} x^a f(dx) \leq  R^a \int_{[1,M]} f(dx) + \int_{[M,\infty)} x^a f(dx). 
\end{eqnarray*}
Using that $a < 0$, we first let  $R \to \infty$ and then $M\to \infty$ to obtain that  $B \leq 0$, which leads to a contradiction.
\end{proofof}

\bigskip

\bigskip

\bigskip

\section{Existence and non-existence results: Discrete model}\label{sec:discr1d}

\subsection{Setting and main results}

We consider the following discrete coagulation equation with source: 
\begin{equation}
\partial_{t}n_{\alpha}=\frac{1}{2}\sum_{\beta<\alpha}K_{\alpha-\beta,\beta
}n_{\alpha-\beta}n_{\beta}-n_{\alpha}\sum_{\beta>0}K_{\alpha,\beta}n_{\beta
}+ s_{\alpha} \label{eq:Dtimecoag}
\end{equation}
where $\alpha\in\mathbb{N}=\{1,2,\dots\}$. 
We assume that the sequence $s=(s_{\alpha})_{\alpha\in {\mathbb{N}}}$ satisfies 
\begin{equation}  \label{eq:Dcond_s}
s_{\alpha}\geq 0 \;\; \forall \, \alpha\in {\mathbb{N}} \quad \text{and}%
\quad \supp s \subset\{1,2,\dots, L_s\}.
\end{equation}
We consider coagulation kernels $K_{\alpha,\beta}: \N^2 \to \N$ defined on the integers satisfying the same conditions as before:
\begin{equation}
K_{\alpha,\beta}\geq 0,\ \ \ K_{\alpha,\beta}=K_{\beta,\alpha},
\label{eq:Dcond_kernel1}
\end{equation}%
\begin{equation}
K_{\alpha,\beta} \geq c_{1}\left( \alpha^{\gamma +\lambda }\beta^{-\lambda
}+\beta^{\gamma +\lambda }\alpha^{-\lambda }\right)  \label{eq:Dcond_kernel2}
\end{equation}%
and 
\begin{equation}
K_{\alpha,\beta} \leq c_{2}\left( \alpha^{\gamma +\lambda }\beta^{-\lambda
}+\beta^{\gamma +\lambda }\alpha^{-\lambda }\right)  \label{eq:Dcond_kernel3}
\end{equation}%
for $(\alpha,\beta)\in \mathbb{N}^{2}$, with $0<c_{1}\leq c_{2}<\infty$.
Similarly to  the continuous case, we will try to construct steady states for the coagulation equation 
\eqref{eq:Dtimecoag} yielding the transfer of particles to infinity. More precisely, we consider stationary injection solutions to the discrete coagulation equation \eqref{eq:Dtimecoag}.

\begin{definition}
\label{Def:DFluxSol} Assume that $K:{\mathbb{N}}^{2}\rightarrow {\mathbb{R}}_{+}$ is a  function satisfying \eqref{eq:Dcond_kernel1} and 
\eqref{eq:Dcond_kernel3}. Assume further that $s=(s_{\alpha})_{\alpha=1 }^{\infty}$
is a sequence in ${\mathbb{R}}$ satisfying~\eqref{eq:Dcond_s}. We will say
that $(n_{\alpha})_{\alpha=1}^{\infty}$ satisfying
\begin{equation}
\sum_{\alpha=1 }^{\infty} \alpha^{\gamma +\lambda }n_{\alpha} + \sum_{\alpha=1 }^{\infty} \alpha^{-\lambda }n_{\alpha} <\infty
\label{eq:Dmoment_cond}
\end{equation}%
is a stationary injection solution of \eqref{eq:Dtimecoag} if the following
identity holds for any test sequence with finite support $%
(\varphi_{\alpha})_{\alpha=1}^{\infty}$: 
\begin{equation}
\frac {1}{2}\sum_{\beta}\sum_{\alpha}K_{\alpha,\beta}n_{\alpha}n_{\beta}%
\left[ \varphi_{\alpha+\beta}-\varphi_{\alpha}-\varphi_{\beta}\right] +\sum
_{ \beta }s_{\beta}\varphi_{\beta}=0 . \label{eq:Dweakform}
\end{equation}
\end{definition}

Next we prove the existence of stationary injection solutions as stated in the next theorems:

\begin{theorem}
\label{thm:Dexistence} Assume that $K:{\mathbb{N}}^{2}\rightarrow {\mathbb{R}}_{+}$ satisfies~\eqref{eq:Dcond_kernel1}--
\eqref{eq:Dcond_kernel3} and $| \gamma +2\lambda | <1.$ Let $s \neq 0 $
satisfy \eqref{eq:Dcond_s}. Then, there exists a stationary injection solution $%
(n_{\alpha})_{\alpha=1}^{\infty}$ to~\eqref{eq:Dtimecoag} in the sense of
Definition~\ref{Def:DFluxSol} satisfying $n_{\alpha}\geq 0$ for all $\alpha$.
\end{theorem}

\begin{theorem}
\label{thm:DNonExistence} Suppose that
$K:{\mathbb{N}}^{2}\rightarrow {\mathbb{R}}_{+}$ satisfies~\eqref{eq:Dcond_kernel1}--\eqref{eq:Dcond_kernel3}  and $|\gamma +2\lambda | \geq 1.$
 Let us assume also that $s \neq 0$ satisfies \eqref{eq:Dcond_s}.  Then, there is no stationary injection solution of~
\eqref{eq:Dtimecoag} in the sense of the Definition \ref{Def:DFluxSol}.
\end{theorem}

\bigskip

\subsection{Existence result}\label{ssec:discr1dEx}

We first consider equations with the form (\ref{eq:Dtimecoag}) but with $n_\alpha(t) $ and $n_\beta(t) $ supported in $I:=\{1,2,\dots,R_{\ast } \}$ for
each $t\geq 0.$ Therefore, (\ref{eq:Dtimecoag}) becomes
\begin{equation}  \label{DevolEqTrunc}
\partial_{t}n_{\alpha}=\frac{1}{2}\sum_{\beta\leq
\alpha-1}K_{\alpha-\beta,\beta
}n_{\alpha-\beta}n_{\beta}-n_{\alpha}\sum_{\beta\leq R_{\ast
}}K_{\alpha,\beta}n_{\beta }+\sum_{ \beta \leq
R_*}s_{\beta}\delta_{\alpha,\beta} \
\end{equation}
where  $\delta_{\alpha,\beta}=1$ if $\alpha=\beta$, and $\delta_{\alpha,\beta}=0$ otherwise.
Let  $(\varphi _{\alpha })_{\alpha \in I}$ be an arbitrary test function such that $\varphi _{\alpha}:[0,T]\rightarrow {\mathbb{R}}$ is continuously differentiable for any $\alpha $. 
Multiplying \eqref{DevolEqTrunc} by  $(\varphi _{\alpha })_{\alpha \in I}$ and adding up in $\alpha$ we obtain the weak formulation of \eqref{DevolEqTrunc}: 
\begin{align}
& \frac{d}{dt}\left( \sum_{\alpha \leq R_{\ast }}{n_\alpha
(t)\varphi_\alpha(t)}\right) -\sum_{\alpha \leq R_{\ast }}{n_\alpha(t)\dot{
\varphi}_\alpha(t) } \notag  \\
& =\frac{1}{2}\sum_{\beta \leq R_{\ast }}\sum_{\alpha \leq R_{\ast
}}K_{\alpha ,\beta }{n_\alpha(t)}{n_\beta(t)}\left[ {\varphi _{\alpha +\beta}(t)}
\chi _{\left\{ \alpha +\beta \leq R_{\ast }\right\} }-\varphi _{\alpha
}(t)-\varphi _{\beta
}(t)\right] +\sum_{\beta \leq R_{\ast}}s_{\beta }\varphi _{\beta
}(t),  \label{eq:Devol_eqWeak}
\end{align}
where $\dot \varphi$ denotes the time-derivative of $\varphi$ and $\chi
_{\left\{ \alpha+\beta\leq R_{\ast }\right\} } $
is the characteristic function of the set $\left\{ \alpha+\beta\leq R_{\ast}\right\}.$

The approximation \eqref{DevolEqTrunc} is known as the non-conservative
approximation of the coagulation equation \eqref{eq:Dtimecoag}. This equation and its weak formulation \eqref{eq:Devol_eqWeak} have been extensively used in the study of the mathematical properties of the coagulation equations (cf. for instance \cite{C98, L99}).

Our first goal is to prove the well-posedness for \eqref{DevolEqTrunc}.
\begin{proposition}
\label{prop:Dwellp} Assume that $1<R_{\ast }<\infty $ and that $%
K:I^{2}\rightarrow {\mathbb{R}}_{+}$ is a  function satisfying %
\eqref{eq:Dcond_kernel1} and \eqref{eq:Dcond_kernel2}. Assume further that $%
s=(s_{\alpha})_{\alpha\in I}$ satisfies~\eqref{eq:Dcond_s}. Let $%
(n_{\alpha}\left( 0\right))_{\alpha\in I}$ be the initial condition. Then,
there exists a unique solution $(n_{\alpha}\left( t\right))_{\alpha\in I}$%
, with $n_{\alpha}: (0,\infty)\to {\mathbb{R}}_{+}$ continuously differentiable
 for any $\alpha$, which solves \eqref{DevolEqTrunc} in the
classical sense.
\end{proposition}

\begin{proof}
The proof of this statement relies on classical arguments of the theory of
ordinary differential equations. We just outline the main steps. To simplify
the notation we define 
\begin{equation*}
g_{\alpha}:{\mathbb{R}}_{+}^{R_{\ast}}\rightarrow {\mathbb{R}}%
_{+}^{R_{\ast}} \quad \text{for}\;\; \alpha=1,\dots, R_{\ast}
\end{equation*}
such that 
\begin{equation*}
g_{\alpha}(\xi_1,\dots, \xi_{R_{\ast}})=\frac{1}{2}\sum_{\beta\leq
\alpha-1}K_{\alpha-\beta,\beta}\xi_{\alpha-\beta}\xi_{\beta}-\xi_{\alpha}%
\sum_{\beta\leq R_{\ast }}K_{\alpha,\beta}\xi_{\beta}+\sum_{\beta \leq
R_*}s_{\beta}\delta_{\alpha,\beta}.
\end{equation*}
Then, we can rewrite \eqref{DevolEqTrunc} as 
\begin{align*}
\partial_t n_{\alpha}=g_{\alpha}(n_1,\dots, n_{R_{\ast}}),
\end{align*}
with initial condition $n_{\alpha}\left( 0\right).$ We observe that the
functions $g_{\alpha}$ are polynomials, therefore they are locally Lipschitz
continuous functions. Thus, due to the Picard-Lindel\"of theorem there
exists a unique solution continuously differentiable $(n_{\alpha}\left(
t\right))_{\alpha\in I}$ on a maximal time interval $[0,T_{\ast})$.

Moreover, since $K_{\alpha ,\beta }\geq 0,$ $s_{\gamma }\geq 0$ by
assumption and $n_{\alpha }(0)\geq 0$ it easily follows that $n_{\alpha
}\geq 0$ in $[0,T_{\ast })$ for any $\alpha =1,\dots ,R_{\ast }$.
The fact that the solutions of \eqref{DevolEqTrunc} are globally defined in time follows from the fact that 
\begin{equation}
\partial_t\left( \sum_{\alpha=1}^{R_*} n_\alpha
\right) \leq \sum_{\alpha=1}^{R_*} s_\alpha.
\end{equation}
\end{proof}

 Next we show the existence of stationary injection solutions to \eqref{DevolEqTrunc} corresponding to time independent solutions of \eqref{DevolEqTrunc}.

\begin{proposition}
\label{thm:Dexistence_truncated} Under the assumptions of Proposition~\ref{prop:Dwellp}, there exists a stationary injection solution $(\hat
n_{\alpha})_{\alpha\in I}$ to~\eqref{DevolEqTrunc} 
satisfying $\hat n_{\alpha}\geq 0$ for any $\alpha\in I$.
\end{proposition}

\begin{proof}
We first construct an invariant region for the evolution equation~%
\eqref{DevolEqTrunc}. From Proposition~\ref{prop:Dwellp} there exists a
unique solution to~\eqref{DevolEqTrunc}, $(n_{\alpha }\left( t\right)
)_{\alpha \in I}$, with $n_{\alpha }:(0,\infty )\rightarrow {\mathbb{R}}%
_{+} $ continuously differentiable for any $\alpha $. In particular, $(n_{\alpha }\left( t\right) )_{\alpha \in I}$ satisfies~%
\eqref{eq:Devol_eqWeak}. Choosing $\varphi _{\alpha }=1$ and using the upper
bound for $\chi _{\left\{ \alpha +\beta \leq R_{\ast }\right\} }\leq 1$ and the lower bound $a_{1}=\min_{\alpha ,\beta \in
I}K_{\alpha ,\beta }$, we obtain 
\begin{equation*}
\frac{d}{dt}\sum_{\alpha \leq R_{\ast }}n_{\alpha }(t)\leq -\frac{a_{1}}{2}%
\left( \sum_{\alpha \leq R_{\ast }}n_{\alpha }(t)\right) ^{2}+c_{0}\hspace{%
1cm}
\end{equation*}%
where $c_{0}=\sum_{\beta \leq R_{\ast }}s_{\beta }\varphi _{\beta }$. Notice
that $a_{1}>0$ because we assume that (\ref{eq:Dcond_kernel2}) holds. We
then obtain the invariant region%
\begin{equation}
\mathcal{U}_{M}=\left\{ (n_{\alpha })_{\alpha \in I}\in {\mathbb{R}}%
^{R_{\ast }}:\sum_{\alpha \leq R_{\ast }}n_{\alpha }\leq M\right\}
\label{eq:Dinvariant_region}
\end{equation}%
with $M\geq \sqrt{\frac{2c_{0}}{a_{1}}}$. Moreover, $\mathcal{U}_{M}$ is
compact and convex. Consider the operator $S(t):{\mathbb{R}}^{R_{\ast }}\rightarrow {%
\mathbb{R}}^{R_{\ast }}$ defined by $n_{\alpha }(t)=S(t)n_{\alpha }(0)$.
This operator is continuous by standard continuity results on the initial
data for the solutions of ODEs (cf.~\cite{CD}). Since the functions $n_{\alpha }(t)$
solve a first order ODE, they are also continuous in time. Then, the mapping 
$t\rightarrow S(t)n_{\alpha }(0)$ is continuous.

We can now conclude the proof of Proposition \ref{thm:Dexistence_truncated}.
The operator $S(t):\mathcal{U}_{M}\rightarrow \mathcal{U}_{M}$ is continuous and 
$\mathcal{U}_{M}$ is convex and compact.
Then, Brouwer's Theorem (cf.\ \cite{Evans}) implies that for all $\delta >0$, there exists a fixed-point $\hat{n}_{\delta }$ of $S(\delta )$ in $\mathcal{U}_{M}$.    
Arguing as in the last paragraph of the proof of Theorem \ref{thm:existence} we conclude that there exists $\hat{n}\in \mathcal{U}_{M} $ such that $S(t) \hat{n} = \hat{n}$, which implies that $\hat{n}$ is a stationary injection solution to \eqref{DevolEqTrunc}. 
\end{proof}

We now prove the Theorem \ref{thm:Dexistence}.

\begin{proofof}[Proof of Theorem~\ref{thm:Dexistence} (existence)] 
We just sketch the argument since it is an adaptation of Theorem \ref{thm:existence}.  For notational convenience we rewrite the kernel $K_{\alpha,\beta}=K(\alpha ,\beta )$ in the form \eqref{B1in3} where now $x$, $y\in \N$. Throughout this proof we will also use the notation $n_\alpha=n(\alpha)$ and $n_\beta=n(\beta)$. 
The function $\Phi(s,x)$ is defined in a subset of the rational numbers contained in the interval $(0,1)$ and satisfies \eqref{eq:B1bound} in this domain of definition. We then define the kernel $K_{\ep}(x,y)$ as in \eqref{eq:1levTrunc} and $K_{\ep,R_{\ast}}(x,y)$ as in \eqref{eq:2levTrunc}. Hence, using Proposition \ref{thm:Dexistence_truncated} 
there exists a stationary injection
solution $n_{\varepsilon ,R_{\ast }}$ satisfying 

\begin{eqnarray}
&\frac{1}{2}\sum\limits_{\beta \leq R_{\ast }}\sum\limits_{\alpha \leq R_{\ast
}}K_{\varepsilon ,R_{\ast }}(\alpha ,\beta )n_{\varepsilon ,R_{\ast }}(\alpha )n_{\varepsilon ,R_{\ast }}(\beta )\left[ \varphi_{\alpha
+\beta }\chi_{\left\{ \alpha +\beta \leq R_{\ast }\right\} }-\varphi_\alpha -\varphi_\beta\right] \nonumber\\
 &+\sum_{\beta \leq R_{\ast
}}s_{\beta }\varphi_\beta =0\label{WeakD_TruncForm}
\end{eqnarray}
for any test function $\varphi: \N \to \R$  compactly supported.

 Choosing $\varphi$ of the form $\varphi _{\alpha }=\alpha \psi
_{\alpha }$ for some compactly supported function $\psi: \N \to \R$, we obtain
\begin{eqnarray*}
&&\varphi_{\alpha+\beta} \chi _{\{  \alpha+\beta \leq R_{\ast }\}} -\varphi_\alpha -\varphi_\beta \\
&=&\alpha(\psi_{ \alpha+\beta}\chi _{\{  \alpha+\beta \leq R_{\ast }\}}  -\psi_{\alpha})+ \beta(\psi_{ \alpha+\beta}\chi _{\{  \alpha+\beta \leq R_{\ast }\}}  -\psi_\beta).
\end{eqnarray*}
Symmetrizing we arrive at%
\begin{equation}
\sum_{\beta \leq R_{\ast }}\sum_{\alpha \leq R_{\ast
}}K_{\varepsilon ,R_{\ast }}(\alpha ,\beta )n_{\varepsilon ,R_{\ast }}(\alpha )n_{\varepsilon ,R_{\ast }}(\beta )\left[\alpha(\psi_{ \alpha+\beta}\chi _{\{  \alpha+\beta \leq R_{\ast }\}}  -\psi_\alpha)\right] +\sum_{\alpha \leq R_{\ast
}}\alpha\psi_\alpha s_\alpha=0.\nonumber
\end{equation}
Let us assume that 
\begin{equation}
\psi_\alpha=0\text{ \ for \ }\alpha\geq R_{\ast }. \nonumber
\end{equation}
For such test functions we have $\psi_{\alpha+\beta}\chi _{\{  \alpha+\beta \leq R_{\ast }\}}  =\psi_{\alpha+\beta}$, therefore,
\begin{equation}
\sum_{\beta \leq R_{\ast }}\sum_{\alpha \leq R_{\ast
}}K_{\varepsilon ,R_{\ast }}(\alpha ,\beta )n_{\varepsilon ,R_{\ast }}(\alpha )n_{\varepsilon ,R_{\ast }}(\beta )\left[\alpha(\psi_{ \alpha+\beta}-\psi_\alpha)\right] +\sum_{\alpha \leq R_{\ast
}}\alpha\psi_\alpha s_\alpha=0.\nonumber
\end{equation} 
Choosing a test function $\psi_\alpha = \chi_{\{\alpha \leq z\}}$ we obtain 
\begin{equation}
\sum_{\alpha  \leq  z} \alpha n_{\varepsilon ,R_{\ast }}(\alpha )\sum_{\beta >  z-\alpha}K_{\varepsilon ,R_{\ast }}(\alpha ,\beta )n_{\varepsilon ,R_{\ast }}(\beta ) = \sum_{\alpha < z} \alpha s_\alpha,\ \ z\in \left( 0,R_{\ast }\right).\nonumber
\end{equation} 

We can then argue as in the proof of \eqref{EstTruncFunction} to obtain 
\begin{equation*}
 \sum_{\alpha < R_*} n_{\varepsilon ,R_{\ast }}(\alpha ) \leq \bar{C}_{\ep}.
\end{equation*}
Therefore there exists a subsequence $R_{\ast}^{n}\to\infty$ and $(n_{\varepsilon}(\cdot ) )\in \ell^{1}(\N)$ such that 
$n_{\varepsilon, R^n_{\ast }}(\alpha ) \to n_{\ep}(\alpha)$ for any $\alpha \in \N$. Moreover, 
$\sum_{\alpha } n_{\varepsilon}(\alpha ) \leq \bar{C}_{\ep}$ and 
 for any bounded test function $\varphi: \N \to \R_+$, $n_{\varepsilon }$ satisfies 
\begin{equation}
\frac{1}{2}\sum_{\beta}\sum_{\alpha }K_{\varepsilon}(\alpha ,\beta )n_{\varepsilon }(\alpha )n_{\varepsilon}(\beta )\left[ \varphi_{\alpha
+\beta }-\varphi_\alpha -\varphi_\beta \right] 
 +\sum_{\beta}s_{\beta }\varphi_\beta =0.\label{WeakD_FormEps}
\end{equation}
Following the same reasoning as in the derivation of \eqref{A1} and \eqref{A2} in the proof of Theorem \ref{thm:existence} we then arrive at 
\begin{eqnarray}
\frac{1}{\beta}\sum_{\alpha \in \left[\frac{2\beta}{3},\beta\right]\cap \N}n_{\varepsilon }(\alpha) &\leq &\frac{c}{\beta^{3/2}}\left( 
\frac{1}{\min  \left\{
\beta^{\gamma},\frac{1}{\varepsilon }\right\} }\right) ^{1/2} \nonumber \\
&\leq & \frac{C}{ \beta^{3/2} \sqrt{\ep}},\ \ \text{ for all }\beta \in
\N. \label{IntEstFEps_D}
\end{eqnarray}
Then, taking subsequences, we obtain that there exists a limit sequence $(n(\alpha ))_{\alpha\in\N}$  such that $n_{\varepsilon_{n}}(\alpha )\to n(\alpha)$ as $n\to\infty$ with $\ep_n\to 0$ as $n\to\infty$ for any $\alpha\in \N$. 

Definition~\eqref{eq:1levTrunc} implies that $\lim_{\varepsilon \rightarrow 0}K_{\varepsilon }\left( \alpha,\beta\right) =K\left(\alpha,\beta\right)$ uniformly in compact sets. Taking now the limit as $n\to\infty$ in \eqref{WeakD_FormEps} we obtain that $n$ satisfies: 
\begin{equation}
\frac{1}{2}\sum_{\beta}\sum_{\alpha }K(\alpha ,\beta )n(\alpha )n(\beta )\left[ \varphi_{\alpha+\beta }-\varphi_\alpha -\varphi_\beta \right] 
 +\sum_{\beta}s_{\beta }\varphi_\beta =0,
\end{equation}
for every test function $\varphi$ compactly supported. The only difficulty doing that is to control the contribution due to the regions with $\beta\geq M$ with $M$ large in the sums
$$\sum_{\beta}\sum_{\alpha }K(\alpha ,\beta )n_{\ep}(\alpha )\varphi_\alpha n_{\ep}(\beta ).$$
This can be made arguing exactly as in the proof of  Theorem \ref{thm:existence} distinguishing the cases \eqref{C1}-\eqref{C4} and replacing the integrals by sums.

Moreover,  taking
the limit of (\ref{IntEstFEps_D}) as $\varepsilon \rightarrow 0$ we arrive at:%
\begin{equation*}
\frac{1}{\beta}\sum_{\alpha \in [2\beta /3,\beta]\cap \N}n(\alpha)\leq \frac{C}{\beta^{3/2+\gamma /2}}\ \ \text{ for
all }\beta\in \N,
\end{equation*}
which implies
\begin{equation*}
\frac{1}{\beta}\sum_{\alpha \in [2\beta/3,\beta]\cap \N} \alpha^{\gamma + \lambda} n(\alpha)\leq \bar C \beta^{\gamma + \lambda - 3/2-\gamma /2}\ \ \text{ for
all }\beta\in \N,
\end{equation*}
which implies (\ref{eq:Dmoment_cond})  using that $-1<\gamma +2\lambda<1$ .
\end{proofof}

\begin{remark}
We notice that in the paper \cite{RT} it has been proved that there exists a unique stationary solution of a problem that can be reformulated as a  solution of \eqref{eq:Dtimecoag} for the explicit kernel $K_{\alpha,\beta}=\alpha\beta.$
\end{remark}
\bigskip 

\subsection{Non-existence result}\label{ssec:discr1dNEx}

We first give an example of a construction of a continuous kernel $\widetilde K$ which interpolates the values of the discrete kernel $K_{\alpha,\beta}$ and satisfies \eqref{eq:cond_kernel2}-\eqref{eq:cond_kernel3}.
Let $K_{\alpha, \beta}$  satisfy \eqref{eq:Dcond_kernel1}-\eqref{eq:Dcond_kernel3}  and let $w$ denote the corresponding weight function in (\ref{eq:cond_kernel}). 
We define  the  continuous  kernel $\widetilde{K}: (\R_*)^2 \to \R_+$ by setting 
\begin{equation}\label{eq:dc_K}
\widetilde K(x,y) = \sum_{\alpha,\beta=1}^\infty  K_{\alpha,\beta} \zeta_\varepsilon(x-\alpha) \zeta_\varepsilon(y-\beta)+
c_1 \left(\zeta_\varepsilon(x) + \zeta_\varepsilon(y)\right) w(x,y)\,, \quad x,y>0\,,
\end{equation}
where   $\varepsilon<1/2$ and  $\zeta_\varepsilon$ is a continuous non-negative function satisfying $\zeta_\varepsilon (x)= 0,\ |x| \geq 1/2+\varepsilon$, $\zeta_\varepsilon (x)= 1,\ |x| \leq 1/2-\varepsilon$  and affine  in each interval $(1/2-\varepsilon,1/2+\varepsilon)$ and  $(-1/2-\varepsilon,-1/2+\varepsilon)$.  
We remark that the series in (\ref{eq:dc_K}) is convergent at any $x$ and $y$ since it contains at most $4$ non-zero terms.

The function $\widetilde K$ is  continuous, non-negative and symmetric  as it is written as a sum of functions with the same properties. 
We now show that $\widetilde K$ satisfies the growth bounds \eqref{eq:cond_kernel2}-\eqref{eq:cond_kernel3} with the same exponents  $\gamma, \lambda$  of the discrete kernel $K_{\alpha,\beta}$, although possibly for different constants $c_1$ and $c_2$.   
Let us observe first that, if $x<\frac{1}{2}$ or $y<\frac{1}{2}$,
the second term in (\ref{eq:dc_K}) is proportional to $w(x,y)$ and thus already provides a suitable lower bound.
The upper bound may also be checked to hold then,
after possibly adjusting $c_2$ from the value in (\ref{eq:Dcond_kernel3}).
Hence, we assume $x,y\ge \frac{1}{2}$ in the following.

For each  $\alpha,\beta \in \N$, we have from \eqref{eq:dc_K} that  $\widetilde K(\alpha,\beta)=K_{\alpha,\beta}$. Therefore $\widetilde K(x,y)$  satisfies \eqref{eq:cond_kernel2}-\eqref{eq:cond_kernel3} for  $x=\alpha$ and $y=\beta$.  If $x\in [\alpha-1/2, \alpha+1/2]$, $y\in [\beta-1/2, \beta+1/2]$ we have that $\frac{1}{2}K_{\alpha,\beta} \leq \widetilde K(x,y) \leq   \sum_{i,j=-1,0,1  } K_{\alpha+i,\beta+j} +
c_1 \left(\zeta_\varepsilon(x) + \zeta_\varepsilon(y)\right) w(x,y)$, where we set $K_{0,j}=K_{j,0}=0$ for $j\in \N$.  Using
the bounds \eqref{eq:Dcond_kernel2}, \eqref{eq:Dcond_kernel3},  and 
the monotonicity properties of $w$, this  
implies that there exist positive constants $c_1$ and $c_2$  such that $\widetilde K(x,y)$ satisfies \eqref{eq:cond_kernel2}-\eqref{eq:cond_kernel3}.

\begin{lemma}\label{lem:discrete_sol_continuous}
Assume that: 
\begin{itemize}
\item $K: \N^2 \to \R_+$ is a function satisfying \eqref{eq:Dcond_kernel1} and \eqref{eq:Dcond_kernel3}, 
\item  $\widetilde K: \R_*^2 \to \R_+ $ is a continuous interpolation of $K$ satisfying \eqref{eq:cond_kernel1} and \eqref{eq:cond_kernel3}, i.e., $\widetilde K \in C(\R_*^2) $ and $K_{\alpha, \beta} = \widetilde K(\alpha,\beta)$,
\item $s = (s_\alpha)_{\alpha \in \N} $ satisfies $s \neq 0$ and \eqref{eq:Dcond_s},
\item  $(n_\alpha )_{\alpha \in \N}$ is a stationary injection solution to \eqref{eq:Dtimecoag} in the sense of Definition \ref{Def:DFluxSol}.
\end{itemize}
Let $f, \eta \in \mathcal{M}_+(\R_*)$ be defined by $ f(dx)=\sum_{\alpha=1}^\infty n_\alpha  \delta_{\alpha}(dx)$ and $ \eta(dx) = \sum_{\alpha=1}^\infty s_{\alpha}\delta_{\alpha}(dx) $, where $\delta_{\alpha} \in \mathcal{M}_+(\R_*)$ is the Dirac measure at $\alpha$. 
  Then  $f$ is a stationary injection solution to the continuous coagulation equation \eqref{eq:time_evol} in the sense of Definition \ref{DefFluxSol} with the kernel $\widetilde K$ and source $\eta$.
\end{lemma}

\begin{proof}
We first  notice that $\eta$ satisfies \eqref{eq:cond_eta} with $L_\eta = L_s$ and $\widetilde K$ satisfies \eqref{eq:cond_kernel1} and \eqref{eq:cond_kernel3}. Therefore $\widetilde K$ and $\eta$ satisfy the assumptions of Definition \ref{DefFluxSol}. 

For $f \in \mathcal{M}_+(\R_+)$ such that $f=\sum_{\alpha=1}^\infty  \delta_{\alpha} n_\alpha$, we have that $f((0,1))=0$. Since $(n_\alpha )_{\alpha \in \N}$ is a stationary injection solution in the sense of Definition \ref{Def:DFluxSol}, then it satisfies \eqref{eq:Dmoment_cond}. Using \eqref{eq:Dmoment_cond} and using Fubini's theorem to exchange the sum and the integral, we obtain:
\begin{eqnarray}
\infty &>& \sum_{\alpha=1 }^{\infty} \alpha^{\gamma +\lambda }n_{\alpha} + \sum_{\alpha=1 }^{\infty} \alpha^{-\lambda }n_{\alpha} \nonumber\\
 &=& \sum_{\alpha=1}^\infty  \int_{\left(0,\infty \right) }x^{\gamma +\lambda } n_\alpha { \delta_{\alpha}(dx) } + \sum_{\alpha=1}^\infty \int_{\left(0,\infty \right) }x^{-\lambda } n_\alpha { \delta_{\alpha}(dx) } \nonumber\\
&=& \int_{\left(0,\infty \right) }x^{\gamma +\lambda }\sum_{\alpha=1}^\infty  n_\alpha { \delta_{\alpha}(dx) }+ \int_{\left(0,\infty \right) }x^{-\lambda }\sum_{\alpha=1}^\infty  n_\alpha { \delta_{\alpha}(dx) } \nonumber\\
&=&\int_{\left(0,\infty \right) }x^{\gamma +\lambda }f\left( dx\right) + \int_{\left(0,\infty \right) }x^{-\lambda }f\left( dx\right) \nonumber
\end{eqnarray}
which proves \eqref{eq:moment_cond}.
For any test function $\varphi \in C_{c}(\R_*)$ we  use  Fubini's theorem to exchange the sum and the integral yielding:
\begin{equation}\label{eq:disc_cont1}
\int_{ (0,\infty )} \varphi\left( x\right) \eta \left( dx\right)
=\int_{(0,\infty )}\varphi\left( x\right)   \sum_{\alpha=1}^\infty  s_{\alpha} { \delta_{\alpha}(dx) }
= \sum_{\alpha=1}^\infty  \int_{ (0,\infty )}\varphi\left( x\right)  s_{\alpha}{ \delta_{\alpha}(dx) }
=\sum_{ \alpha=1 }^\infty s_{\alpha}\varphi_{\alpha}.
\end{equation}
Using \eqref{eq:moment_cond} we have that for any  test function $\varphi \in C_{c}(\R_*)$,
$$\iint_{  (0,\infty )^2}\widetilde K\left( x,y\right) \left[
\varphi \left( x+y\right) -\varphi \left( x\right) -\varphi \left( y\right) %
\right] f\left( dx\right) f\left( dy\right) < \infty $$
is well-defined. 
Using again Fubini's theorem we obtain that 
\begin{eqnarray}\label{eq:disc_cont2}
&&\frac{1}{2}\iint_{(0,\infty )^2}\widetilde K\left( x,y\right) \left[
\varphi \left( x+y\right) -\varphi \left( x\right) -\varphi \left( y\right) %
\right] f\left( dx\right) f\left( dy\right) \nonumber\\
&& \hspace{2cm} = \frac {1}{2}\sum_{\beta}\sum_{\alpha}K_{\alpha,\beta}n_{\alpha}n_{\beta}\left[ \varphi_{\alpha+\beta}-\varphi_{\alpha}-\varphi_{\beta}\right].
\end{eqnarray}
Adding \eqref{eq:disc_cont1} with \eqref{eq:disc_cont2} and using \eqref{eq:Dweakform} we  obtain \eqref{eq:stationary_eq}, which concludes the proof.
\end{proof}

Let $\mathcal{C} \subset \mathcal{M}_+(\R_+)$ be the set of positive bounded Radon measures supported on the natural numbers, i.e.,
$$\mathcal{C} = \{ f \in \mathcal{M}_+(\R_*)\ | \ f = \sum_{\beta = 1}^\infty n_\beta {\delta_{\beta}},\ n_\beta \geq 0,\ \beta \in \N  \}.$$

\begin{proofof}[Proof of Theorem~\ref{thm:DNonExistence} (non-existence)]
We first recall the interpolation construction in the beginning of the Section which allows to extend to the discrete bounds \eqref{eq:Dcond_kernel2}-\eqref{eq:Dcond_kernel3} and construct a continuous interpolation kernel $\widetilde K$ such that \eqref{eq:cond_kernel2}-\eqref{eq:cond_kernel3} hold. From Theorem~\ref{thm:existence}  there is no stationary injection solution  to \eqref{eq:time_evol} in the sense of Definition~\ref{DefFluxSol}. In particular, by Lemma \ref{lem:discrete_sol_continuous} then there is no solution in the subset of discrete measures $\mathcal{C}$, which concludes the proof. 
\end{proofof}

\bigskip

\bigskip

\bigskip

\section{Estimates and regularity}\label{sec:estimates}

In order to define upper and lower estimates for the measure $f$ we need detailed estimates for the fluxes $J$ defined on the left hand side of (\ref{eq:flux_lem}) in Lemma \ref{lem:flux}. 
That is, we consider the function 
\[
 J(z) =  
\iint_{\Omega_z} K\left( x,y\right) xf \left( dx\right)
f\left( dy\right) \,, \quad z >0,  
\]
where
\begin{equation}\label{eq:Omegaz}
 \Omega_z := \{(x,y)\ |\ 0< x \leq z ,\ y > z-x  \}.
\end{equation}
Given $\delta>0$,  we introduce a partition of $(\R_+)^2= \Sigma_1(\delta) \cup \Sigma_2(\delta) \cup \Sigma_3(\delta) $  by
\begin{equation}
\label{eq:sigma}\Sigma_1(\delta) =  \{(x,y)\ |\  y >  x/\delta  \}\,,\quad \Sigma_2(\delta) =  \{(x,y)\ |\ \delta x\leq y \leq   x/\delta  \}\,,\quad \Sigma_3(\delta) = \{(x,y)\ |\ y <  \delta x \} \,,
\end{equation}
and we then define for $j=1,2,3$
\begin{equation} \label{eq:fluxeq_j}
J_j(z,\delta) = 
\iint_{\Omega_z \cap \Sigma_j(\delta)}K\left( x,y\right) xf \left( dx\right)
f\left( dy\right)  \ \text{for }z >0.  
\end{equation}
Clearly, $J(z) = \sum_{j=1}^3 J_j(z,\delta) $ for any choice of $\delta$.

The following Lemma will be used to prove that the contribution to the integral defining the fluxes due to the points contained in $ \Sigma_1(\delta)$ and $ \Sigma_3(\delta) $ are small for $\delta$ sufficiently small.

\begin{lemma}\label{lem:J}
Let $K$ satisfy~\eqref{ContinAssumpt}--\eqref{eq:cond_kernel3} and $| \gamma +2\lambda | <1.$ Suppose that $f \in \mathcal{M}_+(\R_*)$ satisfies
\begin{equation}
 \frac{1}{z}\int_{[z/2,z]} f(dx) \leq \frac{ A }{z^{(\gamma + 3)/2}}\ \ \text{ for all }   z>0. \label{eq:lemm_upperBound}
 \end{equation}
 Then for every $\varepsilon>0$ there exists a  $ \delta_\eps>0$  depending on $\varepsilon$ as well as   on $\gamma,\ \lambda$ and on the constants $c_1,\ c_2$ in  \eqref{eq:cond_kernel2}--\eqref{eq:cond_kernel3}  but independent of $A$ such that for any $\delta \leq \delta_\eps $ we have that
\begin{equation}\label{eq:flux_sigma1}
\sup_{z>0} J_1(z,\delta) \leq \varepsilon A^2 
\end{equation}
and
\begin{equation}\label{eq:flux_sigma3}
\sup_{R>0}\frac{1}{R}\int_{[R,2R]} J_3(z,\delta)dz \leq \varepsilon A^2.
\end{equation}
\end{lemma}

\begin{proof}
We set $\theta := 1/\delta>1$.
In order to estimate the contribution due to the region  $ \Sigma_1(\delta) \cap \Omega_z$ we define 
$$ D(z,\theta):= \{(x,y)\ |\ 0< x \leq z,\ \max\{\theta x ,z/2\} \leq y \}.$$ 
First suppose that $2\lambda+\gamma \geq 0 $.
Using the upper bound for $K$ given in \eqref{eq:cond_kernel3} and the fact that $\Omega_z \cap \Sigma_1 \subset D(z,\theta)$  we obtain that
\begin{equation*}
\iint_{\Omega_z \cap \Sigma_1 }K\left( x,y\right) xf \left( dx\right)
f\left( dy\right) 
\leq  2  c_2 \iint_{D(z,\theta) }  x^{1-\lambda}y^{\gamma+\lambda} f \left( dx\right)
f\left( dy\right).\
\end{equation*}

{
Assuming $0<x\le z$ we denote $a:=\max\{\theta x ,z/2\}$.  Then we can employ item \ref{it:Risinf}
of Lemma \ref{lem:bound} and the upper bound \eqref{eq:lemm_upperBound} to estimate
\[
\int_{[a,\infty)} y^{\gamma+\lambda}
f\left( dy\right) \le A \frac{2^{|\gamma+\lambda|}}{\nu \ln 2} a^{-\nu}\,,
\]
where $\nu:=\frac{1-|2\lambda+\gamma|}{2}=\frac{1-(2\lambda+\gamma)}{2}>0$.
Therefore, by Fubini's theorem, we can now conclude that
\begin{align}\label{eq:Omzcap1}
\iint_{\Omega_z \cap \Sigma_1} K\left( x,y\right) x
f \left( dx\right)
f\left( dy\right)
 \leq  C A  \int_{(0,z]} \max\{\theta x ,z/2\}^{-\nu}
x^{1-\lambda} f \left( dx\right)\,.
\end{align}
Denoting $\varphi(x)=\max\{\theta x ,z/2\}^{-\nu}
x^{1-\lambda} $, it follows from 
\eqref{eq:lemm_upperBound} that 
\[
  \frac{1}{y}\int_{[y/2,y]} \varphi(x) f(dx) \leq
  2^\nu \max\{\theta y ,z\}^{-\nu} 2^{|1-\lambda|} 
  A y^{\nu-1}\,, \quad \text{ for all } y>0.
\]
Thus by item \ref{it:Rlessinf} of Lemma \ref{lem:bound}, we find that
for any $a'\in (0,z]$,
\[
 \int_{[a',z]} \varphi(x) f \left( dx\right)\le 
 2^{|1-\lambda|+\nu}  A \int_{[a',z]} \max\{\theta y ,z\}^{-\nu}  y^{\nu-1} dy 
 + 2^{|1-\lambda|+\nu}  A \max\{\theta ,1\}^{-\nu} \,.
\]
The limit $a'\to 0$ of the right hand side is finite, and thus we can use monotone convergence theorem to conclude that 
\[
 \int_{(0,z]} \varphi(x) f \left( dx\right)\le 
 2^{|1-\lambda|+\nu}  A \left(\int_{0}^z \max\{\theta y ,z\}^{-\nu}  y^{\nu-1} dy 
 +  \theta^{-\nu}\right) \,.
\]
Evaluating the remaining integral and inserting the result in (\ref{eq:Omzcap1}) yields
\begin{eqnarray*}
\iint_{\Omega_z \cap \Sigma_1} K\left( x,y\right) x
f \left( dx\right)
f\left( dy\right) \nonumber
& \leq &  C A^2  \left(1+ \frac{1}{\nu} + \ln \theta \right)\theta^{-\nu} \,.
 \label{eq:flux1}
\end{eqnarray*}
Since $\nu >0$, the factor multiplying $A^2$ converges to $0$ as $\theta \to \infty$, i.e., also when $\delta\to 0$.  Therefore, to any $\vep>0$ there is $\delta_\vep>0$ such that 
\eqref{eq:flux_sigma1} holds for all $0<\delta\le \delta_\vep$.

In the case where $\gamma +2\lambda <0$ we have $\nu=\frac{1+2\lambda+\gamma}{2}>0$.
The above steps can then be repeated simply by 
exchanging the exponents $\gamma+\lambda$ and $-\lambda$ therein. We find
\begin{align*}
& \iint_{\Omega_z \cap \Sigma_1 }K\left( x,y\right) xf \left( dx\right)
f\left( dy\right) 
\leq  2  c_2 \iint_{D(z,\theta) }  x^{1+\gamma+\lambda}y^{-\lambda} f \left( dx\right)
f\left( dy\right)
\\ & \quad \le  C A^2  \left(1+ \frac{1}{\nu} + \ln \theta \right)\theta^{-\nu}\,.
\end{align*}
Thus also in this case \eqref{eq:flux_sigma1} holds for all sufficiently small $\delta$.

To study the region $\Sigma_3(\delta)$, 
we return to the case $\gamma +2\lambda \ge 0$ and assume also $\delta\le \frac{1}{4}$.
Then we have that 
$$\Omega_z \cap \Sigma_3(\delta) \subset \{(x,y)\ |\ 0 < y \leq \delta z,\ z-y \leq x \leq z \}\,.$$ 
In particular, if $(x,y)\in \Omega_z \cap \Sigma_3(\delta)$, 
we have $x\ge (1-\delta)z\ge (\delta^{-1}-1)y> y$.
We integrate \eqref{eq:fluxeq_j} 
in $z$ from $R$ to $2R$, and using \eqref{eq:cond_kernel3} we obtain
\begin{eqnarray*}
I_3 &:=& \int_{[R,2R]}\iint_{\Omega_z \cap \Sigma_3} K\left( x,y\right) xf \left( dx\right)
f\left( dy\right)dz
 \\
& \leq & 2 c_2 \int_{[R,2R]} \int_{(0,\delta z]} \int_{[z-y,z]}  x^{1+\gamma+\lambda}y^{-\lambda}f \left( dx\right)f\left( dy\right)dz .
\end{eqnarray*}
Notice that in the region of integration we have $ R/2 \leq x \leq z\leq 2R$ since $\delta \le \frac{1}{2}$. Therefore, $(0,\delta z] \subset (0,2\delta R]$.
Thus there exists a  constant $C>0$ independent of $\delta$ and $R$  such that 
\begin{eqnarray*}
I_3
 &\leq & C R^{1+\gamma+\lambda} \int_{[R,2R]}  \int_{(0,2 \delta R]}  \int_{[z-y,z]} y^{-\lambda}f \left( dx\right)f\left( dy\right)dz.
\end{eqnarray*}

Using now Fubini's theorem, as well as the fact that 
 $\{(x,z)\ |\ z-y\leq x \leq z,\ R \leq z \leq 2R \} \subset \{(x,z)\ |\ R/2 \leq x \leq 2R, \ x\leq z\leq x+y \}$ if $0<y\le R/2$, we obtain
\begin{eqnarray*}
I_3 &\leq & C R^{1+\gamma+\lambda}\int_{(0,2\delta R]} y^{-\lambda} f\left( dy\right) \int_{[R/2,2R]}f\left( dx\right)  \int_{[x,x+y]} dz   \\
&=&   C R^{1+\gamma+\lambda}\int_{(0,2\delta R]} y^{1-\lambda}f\left( dy\right) \int_{[R/2,2R]} f\left( dx\right). 
\end{eqnarray*}
Then using the bound \eqref{eq:lemm_upperBound}, the assumption $|\gamma +2\lambda|<1$, and item \ref{it:polcase} of
Lemma \ref{lem:bound} applied separately to both of the remaining integrals and with a regularization by $a'\to 0$ in the first integral, we obtain
\[
I_3 \leq  CA^2 R^{1+\gamma + \lambda} \delta^\nu  R^{\nu} R^{-\frac{\gamma+1}{2}}= 
C A^2 \delta^\nu   R, 
\]
where $\nu = \frac{1-|2\lambda+\gamma|}{2}>0$, as before, and $C$ is an adjusted constant independent of $R$ and $\delta$.  Thus the prefactor of $A^2 R$ goes to 
zero as $\delta\to 0$, and we may conclude that to any $\vep>0$ there is $\delta_\vep>0$ such that 
\eqref{eq:flux_sigma3} holds for all $0<\delta\le \delta_\vep$.  Taking smaller of
the cutoffs $\delta_\vep$ obtained for \eqref{eq:flux_sigma1}
and \eqref{eq:flux_sigma3}, we find a value such that both inequalities are
valid whenever $0<\delta\le \delta_\vep$.

In the case where $\gamma +2\lambda <0$, also \eqref{eq:flux_sigma3} can be checked as in the first case, by  exchanging the exponents $\gamma+\lambda$ and $-\lambda$ in the above.
}
\end{proof}

In this Section we use  the assumptions of  Theorem~\ref{thm:existence} as stated next which guarantee the existence of a stationary injection solution $f$ in the sense of Definition~\ref{DefFluxSol}.

\begin{assumption}\label{assump:estimates}
Let $K$ satisfy~\eqref{ContinAssumpt}--\eqref{eq:cond_kernel3}
and suppose $| \gamma +2\lambda | <1.$ Let $\eta \neq 0  $ satisfy \eqref{eq:cond_eta}.
Let $f \in \mathcal{M}_+(\R_*)$, $f\neq 0$ be a stationary injection solution to~\eqref{eq:time_evol} in the sense of Definition~\ref{DefFluxSol} with $f((0,a))=0$, for some $a>0$  (cf.\ Remark \ref{rem:defFluxSol}).
\end{assumption}

Under Assumption \ref{assump:estimates} we obtain, from Lemma \ref{lem:flux}, that $f$ satisfies:
\begin{equation}\label{eq:fluxeq}
\iint_{\Omega_z }K\left( x,y\right) xf \left( dx\right)
f\left( dy\right) = \int_{(0,z]}x\eta \left(dx\right) \ \text{for }z > 0.
\end{equation}

Notice that  Lemma \ref{lem:J} shows that the contributions of the regions $\Sigma_j(\delta) \cap \Omega_z$ with $j=1,3$ to the fluxes defined in \eqref{eq:flux} are  small for $\delta $   sufficiently small.  This shows that  the flux of particles in the size space is due to region $\Sigma_2(\delta) \cap \Omega_z$, which yields the contribution of the collisions between  particles  of comparable size.

\begin{proposition}\label{prop:estimate}
Suppose that Assumption \ref{assump:estimates} holds. Let  $J$ be the constant  $J= \int_{(0,L_\eta]} x \eta (dx)$. Then:
\begin{equation}
 \frac{1}{z}\int_{[z/2,z]} f(dx) \leq \frac{ C_1\sqrt{J} }{z^{(\gamma + 3)/2}}\ \ \text{ for all }   z>0. \label{eq:integral_bounds1}
 \end{equation}
Moreover,  there exists a constant $b$, with $0<b<1$ and depending  on $\gamma,\ \lambda$ and on the constants $c_1,\ c_2$  in \eqref{eq:cond_kernel2}-\eqref{eq:cond_kernel3} such that 
 \begin{equation}
 \frac{1}{z}\int_{ (bz,z]} f(dx)  \geq  \frac{ C_2\sqrt{J} }{z^{(\gamma + 3)/2}}\ \ \text{ for all }   z\geq \frac{L_\eta}{\sqrt{b}} \ .\label{eq:integral_bounds2}
 \end{equation}
The constants $C_1, \ C_2$ that appear in \eqref{eq:integral_bounds1} and \eqref{eq:integral_bounds2} depend   on $\gamma,\ \lambda$ and on the constants $c_1,\ c_2$  in \eqref{eq:cond_kernel2}-\eqref{eq:cond_kernel3}. 
\end{proposition}

\begin{proof}
Using Lemma \ref{lem:flux} we obtain that \eqref{eq:fluxeq} holds.
 We first prove the upper bound \eqref{eq:integral_bounds1}.
Using that  $[2z/3,z]^2 \subset \Omega_z$, where  $\Omega_z$ is as in \eqref{eq:Omegaz},
 we obtain
\begin{equation*}
\iint_{\left[ 2z/3,z\right]^2} K\left( x,y\right) xf \left( dx\right)
f\left( dy\right) \leq J.
\end{equation*}
Using the lower bound~\eqref{eq:cond_kernel2} for $K$ and the fact that $x$ and $y$ are of the same order of $z$ in the domain of integration we obtain
\begin{equation*}
z^{\gamma+1} \left( \int_{\left[ 2z/3,z\right] } f \left( dx\right)\right)^2
 \leq  C^2 J
\end{equation*}
for some positive constant $ C$ which depends only on $K$. Equivalently,
\begin{equation}
\frac{1}{z}\int_{\left[ 2z/3,z\right] } f \left( dx\right)
 \leq \frac{ C \sqrt{J}}{z^{(\gamma+3)/2}},\ \text{for }z\in (0,\infty), \label{eq:upperEstimate}
\end{equation}
which proves the upper estimate using that $[z/2,z]\subset [4z/9,2z/3] \cup [2z/3,z] $.

Using $J= \int_{(0,L_\eta]} x \eta (dx)$ in \eqref{eq:fluxeq} as well as the definition of $J_j$ in \eqref{eq:Omegaz}--\eqref{eq:fluxeq_j}  we obtain
\begin{equation}\label{eq:sumJ}
J = \sum_{j=1}^3 J_j(z,\delta),\  \ z\geq L_\eta.
\end{equation} 
Integrating \eqref{eq:sumJ} with respect to $z$ in $[R,2R]$, using the upper estimate \eqref{eq:integral_bounds1} 
as well as Lemma \ref{lem:J}
we obtain that for $\delta >0$ sufficiently small depending only  on $\gamma,\ \lambda$ and on the constants $c_1,\ c_2$  in \eqref{eq:cond_kernel2}--\eqref{eq:cond_kernel3}, the following chain of inequalities holds with $A:=C \sqrt{J}$ and $\varepsilon \leq \frac{1}{4A^2}$ 
\begin{eqnarray*}
\frac{JR}{2} &\leq & J(1-2\varepsilon A^2) R \leq \int_{[R,2R]} J_2(z,\delta) dz \\ &\leq & \int_{[R,2R]}\iint_{\Omega_z \cap \Sigma_2(\delta)} K\left( x,y\right) xf \left( dx\right)
f\left( dy\right)dz, \quad R \geq L_\eta \ . 
\end{eqnarray*}

A simple geometrical argument shows that there exists a constant $b$, $0<b<1$, depending only on $\delta$ (and therefore on $\gamma,\ \lambda,\ c_1$ and $c_2$) such that
$\underset{z \in [R,2R]}{\bigcup} (\Omega_z \cap \Sigma_2(\delta)) \subset (\sqrt{b}R, R/\sqrt{b}]^2 $. 
(For a fixed $z$ and $(x,y)\in \Omega_z \cap \Sigma_2(\delta)$ one finds
$(\delta^{-1}+1)^{-1} z<x,y\le \delta^{-1} z$; thus for example 
$b = \frac{\delta^2}{4}$ would suffice.) 
Moreover, for every $(x,y) \in (\sqrt{b}R, R/\sqrt{b}]^2 $ we have
$x K(x,y) \leq C R^{\gamma+1}$, with $C$ depending only on $\gamma,\ \lambda,\ c_1$ and $c_2$.
Then 
\begin{equation*}
\frac{JR}{2} \leq  C R R^{\gamma+1} 
\left( 
\int_{(\sqrt{b}R, R/\sqrt{b}]}f\left( dx\right) 
\right)^2,
\end{equation*}
whence $1/R \int_{(\sqrt{b}R,R/{\sqrt{b}}]} f(dx) \geq C \sqrt{J} R^{-(\gamma+3)/2}$ for $R\geq L_\eta$.
Thus \eqref{eq:integral_bounds2} follows after substituting $R/\sqrt{b}$ by $z$.
\end{proof}

In the next Corollary we obtain the moment estimates for a stationary injection solution, when it exists.

\begin{corollary}\label{cor:moments}
Suppose that Assumption \ref{assump:estimates} holds. Then we have the following moment estimates:
\begin{enumerate}
\item[a)] 
$
\int_{\mathbb{R}_{*}}x^{\mu}f\left( dx\right) < \infty \quad \text{for }\quad \mu < \frac{\gamma  + 1}{2}
$ , 
\item[b)] 
$\int_{\mathbb{R}_{*}}x^{\frac{\gamma+1}{2}}f\left( dx\right) = \infty$ .
\end{enumerate}
\end{corollary}
\begin{proof}
a)  The boundedness of moments of order $\mu$ for $\mu < \frac{\gamma+1}{2}$ has  already been obtained  in the proof of Theorem~\ref{thm:existence} in Section~\ref{sec:existence} equation~\eqref{eq:moment_mu}. Notice also that $a)$ is an easy consequence of \eqref{eq:integral_bounds1} and Lemma \ref{lem:bound}. 

b) Using the lower bound~\eqref{eq:integral_bounds2} and multiplying by $z^{(\gamma+3)/2}$ we obtain 
\begin{equation*}
C_2\sqrt{J} \leq z^{(\gamma+1)/2} \int_{(bz,z]}f(dx) \leq   C \int_{(bz,z]}x^{(\gamma+1)/2} f(dx),\quad z \geq \frac{L_\eta}{\sqrt{b}}
\end{equation*}
for some constant $C>0$. In particular, for any natural number $n$ satisfying $b^{-n}\ge L_\eta / \sqrt{b}$ and for  $z=b^{-n}$ we have that  $C_2\sqrt{J} \leq C \int_{(b^{1-n},b^{-n}]}x^{(\gamma+1)/2} f(dx)$.
 Summing in $n$ we finally obtain the result $b)$.
\end{proof}

\begin{remark}
Notice that for $\gamma>1$ Corollary \ref{cor:moments} a) implies that the first moment $\int_{\R_*} xf(dx)$ is finite. Therefore the stationary injection solutions can be interpreted in this case as solutions having a finite number of monomers for which the source of monomers $\eta(x)$ is balanced with the flux of monomers towards infinity. This is closely related to the phenomenon of gelation, which takes place for $\gamma>1$, in which it is possible to have  solutions with a finite number of monomers having a flux of monomers towards infinity. Notice that for $\gamma<1$ we have that $\int_{\R_*} xf(dx)$ is infinite. We further observe that the existence or non existence of stationary injection solutions is independent of the corresponding kernels yielding mass conservation or gelation. 
\end{remark}

\begin{remark}
We observe that for $\gamma>-1$, Corollary \ref{cor:moments} implies that the number of clusters associated to the stationary injection solutions $\int_{\R_*} f(dx)$ is finite  
and Proposition \ref{prop:estimate} together with Lemma \ref{lem:bound} yields
the following  integral estimates:
$$  \frac{ C_1\sqrt{J}}{z^{(\gamma + 1)/2}}  \leq \int_{[z,\infty)} f(dx) \leq
\frac{ C_2\sqrt{J}}{z^{(\gamma + 1)/2}} \quad \text{for } z \geq L_{\eta}
$$
where $J= \int_{(0,L_\eta]} x \eta (dx)$  and  $0< C_1 \leq C_2 $.
\end{remark}

\begin{remark}
The result in Corollary \ref{cor:moments} has been obtained in \cite{Dub} in the case of bounded kernels.
\end{remark}

The next Corollary contains the estimates for a stationary injection solution in the discrete case.

\begin{corollary}\label{prop:Destimate}
Assume that $K:{\mathbb{N}}^{2}\rightarrow {\mathbb{R}}_{+}$ satisfies~\eqref{eq:Dcond_kernel1}-%
\eqref{eq:Dcond_kernel3} and $| \gamma +2\lambda | <1.$ Let $s \neq 0 $
satisfy \eqref{eq:Dcond_s}. Let $(n_\alpha)_{\alpha=1}^{\infty}$ be a stationary injection solution  to~\eqref{eq:Dtimecoag} in the sense of
Definition~\ref{Def:DFluxSol}.
 Then:
\begin{equation}
 \frac{ C_1 \sqrt{J}}{z^{(\gamma + 3)/2}} \leq 
 \frac{1}{z}\sum_{\alpha \in \N \cap [z/2,z]} n_\alpha \leq \frac{ C_2\sqrt{J}}{z^{(\gamma + 3)/2}}\ \  \text{ for all }  z\geq L_{s} \label{eq:Dintegral_bounds}
 \end{equation}
where $J=\sum_\alpha s_\alpha$  and the constants $0< C_1 \leq C_2 $  are independent of $s$. 
\end{corollary}
\begin{proof}
The results follow directly from Lemma \ref{lem:discrete_sol_continuous} and Proposition \ref{prop:estimate}.
\end{proof}

Finally we obtain that the solutions to the continuous problem when they exist are measures in $C^k(\R_+)$ provided that the source $\eta$ and the kernel $K$ are functions in $ C^k(\R_+)$ and that the derivatives of $K$ satisfy some growth conditions.

\begin{lemma}\label{lem:estimate_x0}
Suppose that Assumption \ref{assump:estimates} holds with  $\eta \in L^\infty((0,\infty))$.  
Let $L_0 \geq \frac{1}{2}$ and $0<\rho < \frac{1}{8}$. Assume that 
there exists $A >  0$ such that 
\begin{equation}
\int_{[x_0-r,x_0+r]}f(dx) \leq Ar\label{eq:reg_1}
\end{equation}
for all $r\leq \rho$ and for all  $x_0 \in [\frac{1}{4}, L_0]$.
Then there exists a constant $B>0$ that depends on $L_0,\ \eta$ and $A$, but it is independent of $r$ such that
\begin{equation}
\int_{[x_0-r,x_0+r]}f(dx) \leq B( A^2+\| \eta \|_{L^\infty}) r,\label{eq:reg_2}
\end{equation}
for any $x_0 \in [\frac{1}{4},L_0+1]$ and $r \leq \rho/2$.
\end{lemma}
\begin{proof}
Using \eqref{eq:stationary_eq}
we obtain for all $\varphi \in C_c(\R_*)$
\begin{equation}
\int_{\R_*} \varphi(x) \alpha(x) f(dx) =
\int_{\R_*}\varphi
\left( x\right) \eta \left( dx\right) + \frac{1}{2}\int_{\R_*}\int_{\R_*} K\left( x,y\right) 
\varphi \left( x+y\right)  f\left( dx\right) f\left( dy\right), \label{eq:regu_stat}
\end{equation}
where 
\begin{equation} \label{eq:reg_3}
\alpha(x) = \int_{\R_*} K(x,y) f(dy).
\end{equation}
The continuity and the lower estimate for the kernel $K$ (cf.\ \eqref{eq:cond_kernel2} and \eqref{ContinAssumpt}) 
imply that $\alpha(x) \geq \alpha_{L_0}>0, $ for all $x \in [\frac{1}{8},L_0+1]$.
Using an approximation argument as in Lemma \ref{lem:flux}, we may use in \eqref{eq:regu_stat} a test function $\varphi(x) = \chi_{[x_0-r, x_0+r]}(x)$. Using the boundedness of $\eta$ we obtain
\begin{equation}
\int_{[x_0-r, x_0+r]}f(dx) \leq \frac{1}{\alpha_{L_0}} \left( 2 \| \eta\|_{L^\infty} r + \frac{1}{2}\int_{\R_*}\int_{\R_*} K\left( x,y\right) 
\chi_{[x_0-r, x_0+r]}(x+y)  f\left( dx\right) f\left( dy\right) \right) \ .
\end{equation}
We now use a geometrical argument to show that for every $x_0 \in [\frac{1}{4}, L_0+1] $ and $r < \frac{\rho}{2} $ there exists a set $\{ \xi_\ell \}_{\ell \in { I}} \subset \R_+$ such that $\# {I} \leq \frac{L_0+1}{r}$ and 
$$
\{(x,y) \ |\ |x+y- x_0| \leq r \} \subset \bigcup_{\ell \in {I}} Q_\ell 
$$
with $Q_\ell=[\xi_\ell - 2r,\xi_\ell +2r] \times [x_0-\xi_\ell - 2r,x_0- \xi_\ell +2r]$ and $\xi_\ell \leq x_0 $ for all $\ell \in {I}$. 

This can be seen just locating points along the segment $\left\{  \left(  x,y\right)  :x+y=x_{0},\ x\geq0,\ y\geq0\right\}  $ given by $\left\{  \left(  \xi_{\ell},x_{0}-\xi_{\ell}\right)  \right\}  _{\ell\in  I}$ and such that $\operatorname*{dist}\left(  \xi_{\ell},\left\{  \xi_{j}\right\}  _{j\in I}\backslash\left\{  \xi_{\ell}\right\}  \right)  =r.$ Then, the union of the
cubes $Q_{\ell}$ cover the strip $\left\{  \left(  x,y\right)  :\left\vert
x+y-x_{0}\right\vert \leq r,\ x\geq0,\ y\geq0\right\}  .$ 
 Using the boundedness of $K$ for $x\geq 1, y \geq 1$ and $x+y \leq L_0 + 1 + \frac{\rho}{2}$ as well as the fact that $f((0,1))=0$ by assumption, we obtain 
\begin{equation*}
\int_{[x_0-r, x_0+r]}f(dx) \leq \frac{1}{\alpha_{L_0}} \left( 2 \| \eta\|_{L^\infty} r +
 C \sum_{\ell \in { I}} \iint_{Q_\ell}   f\left( dx\right) f\left( dy\right) \right)
\end{equation*}
where $C$ depends on $K$ and $L_0$. 
Using \eqref{eq:reg_1} it follows that  
 $\iint_{Q_\ell}   f\left( dx\right) f\left( dy\right)  \leq 4A^2 r^2
$. Then, since  $\#  I \leq \frac{L_0+1}{r}$, we get
\begin{equation*}
\int_{[x_0-r, x_0+r]}f(dx) \leq \frac{1}{\alpha_{L_0}} \left( 2 \| \eta\|_{L^\infty} r +
 4 A^2 C (L_0+1)r \right).
\end{equation*}
Hence \eqref{eq:reg_2} follows.
\end{proof}

\begin{proposition}
Suppose that Assumption \ref{assump:estimates} holds with  $\eta \in C((0,\infty))$.  Then $f \in  C((0,\infty))$.

In addition, suppose that for some $k\geq 1$  we have that $\eta \in C^k((0,\infty))$, $ K  \in C^k((0,\infty)^2)$   and that for every $P>1$ there exists a constant $C_P$ such that 
\begin{equation}\label{eq:reg_5}
\left|\frac{\partial^\ell K}{\partial x^\ell} (x,y)\right| \leq C_P[y^{-\lambda} + y^{\gamma+\lambda}],\ \ \  \forall x \in [1,P],\ y\in (0,\infty), \  1 \leq \ell \leq k.
\end{equation}
  Then $f \in C^k((0,\infty))$.
\end{proposition}
\begin{proof}
Suppose that $\eta \in C((0,\infty))$. 
Using that $f((0,1))=0$ it follows that $\int_{[x_0-r, x_0+r]}f(dx) =0$ for all $x_0 \in [ 1/8, 1/2]$ and $r \leq \rho =1/8$.
 Given any $M>1/8$, it then follows from Lemma \ref{lem:estimate_x0} that 
$\int_{[x_0-r, x_0+r]}f(dx) \leq C_M r$ for any $x_0 \in [ 1/8, M]$ and $r \leq \rho_M$ and $\rho_M>0$ sufficiently small. Then, since every null set can be covered by a countable union of intervals with arbitrary small lengths, we have that  $f$ is absolutely continuous with respect to the Lebesgue measure.  Thus $f(dx)=f dx$ for some $f\in L^1_{loc}(\R_{+})$. Moreover, $f(x_0)=\lim_{r\to 0} \frac 1 r \int_{[x_0-r,x_0+r]}f(dx),$ a.e. $x_0\in\R_
{+}$ whence $f(x_0)\leq C_{M}$ a.e. $x_0\in\R_
{+}$.  
Hence  $f \in L_{\text{loc}}^\infty(\R_+).$
Using also the weak formulation \eqref{eq:regu_stat}
it follows that
\begin{eqnarray*}
f(x) &=& \frac{1}{\alpha(x)}[\eta(x)+\frac{1}{2} \int_0^x K(x-y,y) f(x-y)f(y)dy)]\,,\\
&=&  \frac{1}{\alpha(x)}[\eta(x)+ \int_0^{x/2} K(x-y,y) f(x-y)f(y)dy)]\,,
\end{eqnarray*}
with $\alpha$ given in \eqref{eq:reg_3}.
Then $f \in C((0,\infty))$ can be obtained by induction, taking as starting point the fact that $f(x)=0$ for $0 \leq x \leq 1$.   
The fact that $f \in C^k((0,\infty))$ if $\eta \in  C^k((0,\infty)) $ and  \eqref{eq:reg_5} follows in a similar manner.
\end{proof}

\bigskip

\bigskip

\bigskip

\section{Convergence of discrete to continuous model}\label{sec:discr_cont}

We start by  defining constant flux solution (cf.\ Section \ref{sec:types_sol}).

\begin{definition}
\label{DefFluxSol2} Assume that $K:{\mathbb{R}}_{*}^{2}\rightarrow {\mathbb{R}%
}_{+}$ is a continuous function satisfying \eqref{eq:cond_kernel1} and %
\eqref{eq:cond_kernel3}.  
 We will say that $f\in 
 \mathcal{M}_+\left( 0,\infty \right) ,$ satisfying:
\begin{equation}
\int_{(0,\infty) }x^{\gamma +\lambda }f\left( dx\right) + \int_{(0,\infty ) }x^{-\lambda }f\left( dx\right) <\infty
\label{eq:moment_cond2}
\end{equation}
is a constant flux solution of \eqref{eq:time_evol} with $\eta \equiv 0$ if the following
identity holds for  some constant $J\geq 0$ and for any $z>0$: 
\begin{equation}
\int_{(0,z]}\int_{(z-y, \infty)}y K\left( x,y\right) f\left( dx\right) f\left( dy\right) =J.  \label{eq:stationary_eq2}
\end{equation}
\end{definition}

\begin{remark}
Note that in Definition \ref{DefFluxSol2} we use measures  $f\in 
 \mathcal{M}_+\left( 0,\infty \right)$ and therefore the measure can be unbounded in any interval of the form $(0,a)$ for any $a>0$.
\end{remark}

Our goal is to prove that for a large class of kernels $K_{\alpha,\beta}$ satisfying \eqref{eq:Dcond_kernel1}-\eqref{eq:Dcond_kernel3}, the stationary injection solutions to the discrete problem \eqref{eq:Dtimecoag} can be approximated for large cluster sizes by constant flux solutions of the continuous problem \eqref{eq:time_evol} in the sense of Definition \ref{DefFluxSol2}.
Since we proved in Theorems \ref{thm:Dexistence}-\ref{thm:DNonExistence} that stationary injection solutions to \eqref{eq:Dtimecoag} exist if and only if $|\gamma+2\lambda|<1$,  we will assume this condition in the rest of this Section.  To this end, for each $R>0$ we  construct stationary injection solutions $f_R$ to \eqref{eq:time_evol} with some suitable kernel $K_R$ and $\eta_R$ satisfying $ \supp \eta_R \subseteq [1/R,  L_\eta/R]$ (cf.\ Remark \ref{rem:defFluxSol}). 

Let $K_{\alpha, \beta}$  satisfy \eqref{eq:Dcond_kernel1}-\eqref{eq:Dcond_kernel3}  with  $\left\vert \gamma+2\lambda
\right\vert <1$ and $s$ satisfy \eqref{eq:Dcond_s}. Let $(n_\alpha )_{\alpha \in \N}$ be a discrete stationary injection solution to \eqref{eq:Dtimecoag} in the sense of Definition \ref{Def:DFluxSol}. 
For each $R>0$, we  define the  measure $f_R \in \mathcal{M}(\R_*) $ by 
\begin{equation}\label{eq:dc_fR}
 f_R(dx) =  R^{(3+\gamma)/2}\sum_{\alpha=1}^\infty n_{\alpha}  \delta_{\alpha/R}(dx) ,
\end{equation}
and the  continuous  kernel $K_R: (\R_*)^2 \to \R_+$ by 
\begin{equation}\label{eq:dc_KR}
K_R(x,y) = R^{-\gamma} \sum_{\alpha,\beta=1}^\infty  K_{\alpha,\beta} \zeta_\varepsilon(Rx-\alpha) \zeta_\varepsilon(Ry-\beta)
{+ 
c_1 \left(\zeta_\varepsilon(Rx) + \zeta_\varepsilon(Ry)\right) w(x,y)\,,}
\end{equation}
where  $w$ denotes the weight function in (\ref{eq:cond_kernel}),  $\varepsilon<1/2$, and  $\zeta_\varepsilon$ is a continuous non negative function satisfying $\zeta_\varepsilon (x)= 0,\ |x| \geq 1/2+\varepsilon$, $\zeta_\varepsilon (x)= 1,\ |x| \leq 1/2-\varepsilon$  and affine  in each interval $(1/2-\varepsilon,1/2+\varepsilon)$ and  $(-1/2-\varepsilon,-1/2+\varepsilon)$.  
Moreover, we define the source $\eta_R \in \mathcal{M}(\R_*) $ with  $\supp \eta_R \subseteq [1/R,L_\eta/R]$ by
\begin{equation}\label{eq:dc_etaR}
\eta_R(dx) =  R^{2}\sum_{\alpha=1}^\infty s_{\alpha}  \delta_{\alpha/R}(dx).
\end{equation}

\begin{lemma}\label{lem:fR}
The kernel $K_R$ satisfies \eqref{ContinAssumpt}-\eqref{eq:cond_kernel3} with  $\left\vert \gamma+2\lambda
\right\vert <1$, uniformly in $R$. 
The  measure $f_R$ defined as in \eqref{eq:dc_fR} is a stationary injection solution to \eqref{eq:time_evol} in the sense of Definition \ref{DefFluxSol} (cf.\ Remark \ref{rem:defFluxSol}) satisfying \eqref{eq:moment_cond} and \eqref{eq:stationary_eq} with the kernel $K_R$ and the source $\eta_R$ given by \eqref{eq:dc_KR} and \eqref{eq:dc_etaR} respectively.
\end{lemma}

\begin{proof}
{We first notice that the function $K_R$ is  continuous, non-negative and symmetric  as it is written as a sum of functions with the same properties. Next we will show that $K_R$ satisfies the growth bounds \eqref{eq:cond_kernel2}-\eqref{eq:cond_kernel3} with the same exponents  $\gamma, \lambda$  of the discrete kernel $K_{\alpha,\beta}$. In particular these exponents satisfy $\left\vert \gamma+2\lambda\right\vert <1$. 
If $x <\frac{1}{2}$ or $y < \frac{1}{2}$, the second term in \eqref{eq:dc_KR} is proportional to $w(x,y)$ and thus it provides a suitable lower bound. The upper bound also holds after possibly adjusting $c_2$ in \eqref{eq:Dcond_kernel3}. Hence, we may assume $x,y \geq \frac{1}{2}$ in the following.
For each  $\alpha,\beta \in \N$, we have from \eqref{eq:dc_KR} that  $K_R(\alpha/R,\beta/R)=R^{-\gamma}K_{\alpha,\beta}$. Therefore $K_R(x,y)$  satisfies \eqref{eq:cond_kernel2}-\eqref{eq:cond_kernel3} for  $x=\alpha/R$ and $y=\beta/R$   uniformly in $R$ due to the assumption on $K_{\alpha,\beta}$ \eqref{eq:Dcond_kernel2}-\eqref{eq:Dcond_kernel3}.  For $x\in [(\alpha-1/2)/R, (\alpha+1/2)/R]$, $y\in [(\beta-1/2)/R, (\beta+1/2)/R]$ we have that $\frac{1}{2} R^{-\gamma}K_{\alpha,\beta} \leq K_R(x,y) \leq  R^{-\gamma} \sum_{i,j=-1,0,1  } K_{\alpha+i,\beta+j}$, where we set $K_{0,j}=K_{j,0}=0,$ for $j\in \N$. This implies together with the bounds \eqref{eq:Dcond_kernel2}-\eqref{eq:Dcond_kernel3} and the monotonicity properties of $w$,  that there exist positive constants $c_1$ and $c_2$  that are independent of $R$ and such that $K_R(x,y)$ satisfies \eqref{eq:cond_kernel2}-\eqref{eq:cond_kernel3}, which concludes the first part of the Lemma.}

Next we  substitute the expressions for $f_R$, $K_R$ and $\eta_R$ in the weak formulation \eqref{eq:stationary_eq} and perform a change of variables $\xi = Rx$  and $\theta = Ry$. We then obtain an expression where all the terms are multiplied by the same factor $R$. Using then that for $m \in \N $, $\zeta_{\varepsilon}(m) = 1, \ m=0$ and $\zeta_{\varepsilon}(m) = 0, \ m\neq 0$ we obtain that the weak formulation of the continuous problem \eqref{eq:stationary_eq} reduces to the weak formulation of the discrete problem \eqref{eq:Dweakform}.
\end{proof}

\begin{theorem}\label{thm:convergence}
Let $(n_\alpha )_\alpha$ be a solution of the stationary problem 
\eqref{eq:Dtimecoag} in the sense of Definition \ref{Def:DFluxSol}.   Let $f_R, K_R$ and $\eta_R$ be as in \eqref{eq:dc_fR}, \eqref{eq:dc_KR} and \eqref{eq:dc_etaR}, respectively.  Assume that there exists  $K \in C((\R_*)^2)$ such that  $K_R \to K $ as $R \to \infty$ uniformly on compact sets of $(0,\infty)^2$.
Consider the family of stationary injection solutions defined above $(f_R)_{R>0}$. Then for any sequence $( R_n)_{n\in \N}$ such that $\lim_{n \to \infty} R_{n} = \infty$ there exists a subsequence $( R_{n_k})_{k\in \N}$  and  $f \in \mathcal{M}(0,\infty)$ (that might depend on the subsequence) such that
\begin{equation}\label{eq:dc_weak_conv}
\forall \varphi \in  C_c(0,\infty),\ \int_{\R_*} \varphi(x)  f_{R_{n_k}}(dx) \to \int_{\R_*} \varphi(x) f(dx) \text{ as } k \to \infty
\end{equation}
and $f$ is a constant flux solution to \eqref{eq:time_evol} in the sense of Definition \ref{DefFluxSol2} with $J = \sum_{\alpha=1}^\infty \alpha s_\alpha$.
\end{theorem}
\begin{remark}
Note that a priori we may expect that the only constant flux solutions in the sense of Definition \ref{DefFluxSol2} are power laws.
{ We will see in \cite{FLNVprep} that there are homogeneous kernels $K$ that satisfy the upper and lower bounds \eqref{eq:cond_kernel2}--\eqref{eq:cond_kernel3} for which this is not true.} Therefore the limit measure $f$ can be different for different subsequences $(f_{n_k})_k$  in \eqref{eq:dc_weak_conv}.
\end{remark}

\begin{remark}
The assumption   $K_R \to K $ as $R \to \infty$ means that the discrete kernel $K_{\alpha,\beta}$ behaves like  the continuous kernel $K$ for large values of $\alpha,\ \beta$.   For instance, in the case of the kernel $K_{\alpha,\beta}=\frac{\alpha^{\gamma+\lambda}}{\beta^{\lambda}}%
+\frac{\beta^{\gamma+\lambda}}{\alpha^{\lambda}},$ the function $K_{R}$
defined by means of \eqref{eq:dc_KR} converges to $K\left(  x,y\right)  =\frac
{x^{\gamma+\lambda}}{y^{\lambda}}+\frac{y^{\gamma+\lambda}}{x^{\lambda}}$ as $R\rightarrow\infty.$ A large class of kernels $K_{\alpha,\beta}$ for which the convergence $K_{R}\rightarrow K$ as $R\rightarrow\infty$ takes place can be obtained just restricting a continuous homogeneous kernel $K=K\left(x,y\right)  $ to integer values, i.e $K_{\alpha,\beta}=K\left(  \alpha,\beta\right)  $ for $\alpha,\beta\in\mathbb{N}$.
\end{remark}

\begin{proofof}[Proof of Theorem \ref{thm:convergence}] 
Using the expression \eqref{eq:dc_fR} for $f_R$ and the upper estimate in Corollary \ref{prop:Destimate} we obtain 
\begin{equation}
\frac{1}{z}\int_{[z/2,z]} f_R(dx) = C \frac{R^{(\gamma+3)/2}}{Rz}\sum_{\alpha \in [R z/2,Rz]} n_\alpha \ \leq \ \frac{C \sqrt{J}}{z^{(\gamma+3)/2}},\ z>0\label{eq:dc_fR_3}
\end{equation}
for some positive constant $C$ independent of $R$. Note that this estimate is valid for $Rz \geq 1$ and for $0<Rz< 1$ is automatic because the sum is empty. 
Therefore $\{ {f_R}_{|_K} \}_{R>0}$ is precompact in $\mathcal{M}_+(K)$ for any $K \subset (0,\infty)$ compact, where $|_K$ denotes the restriction to $K$.
Given a sequence of compact sets in $(0,\infty)$, $I_n = [2^{-n},2^n]$ we then obtain using a diagonal sequence argument,  that there is a subsequence of measures $( f_{R_{n_k}})_{k\in \N}$ and a measure $f\in  \mathcal{M}_+(\R_+)$ such that  \eqref{eq:dc_weak_conv} holds. Moreover
\begin{equation}
\frac{1}{z}\int_{[z/2,z]} f(dx)  \leq \ \frac{C \sqrt{J}}{z^{(\gamma+3)/2}},\ z>0.\label{eq:dc_fR_2}
\end{equation}

Now we prove that $f $ is a constant flux solution in the sense of Definition \ref{DefFluxSol2}. 
 Indeed, since $f$ satisfies \eqref{eq:dc_fR_2} and using that $\left\vert\gamma+2\lambda\right\vert <1$ and Lemma \ref{lem:bound}, it follows that the moment condition \eqref{eq:moment_cond2}  holds. 
 
For any test function $\varphi \in C_c(0,\infty)$, since $f_R$ is a stationary injection solution, we have from Lemma \ref{lem:flux} that $f_R$ satisfies
\begin{equation}\label{eq:dc_2}
\int_{(L_\eta/R,\infty)}dz \varphi(z)\int_{ (0,z]}\int_{(z-x,\infty)} K_R(x,y)x f_R(dx) f_R(dy)  = J \int_{(L_\eta/R,\infty)}dz \varphi(z),
\end{equation} 
where $J=\int_{(0,\infty)} x \eta_R(dx) = \sum_{\alpha=1}^\infty \alpha s_\alpha >0$ is independent of $R$. We now rewrite using the domain of integration $\Omega_z$ defined in \eqref{eq:Omegaz} as well as the domains $\Sigma_1(\delta),\ \Sigma_2(\delta)$ and $\Sigma_3(\delta)$ for $\delta >0$ defined in \eqref{eq:sigma}. We use also the partial fluxes $J_j,\ j=1,2,3$ defined in \eqref{eq:fluxeq_j}. In order to make explicit in these fluxes the dependence on the kernel $K$ and the measure $f$, we will rewrite them as $J_j(z,\delta; K,f)$ in the rest of this proof.
Therefore \eqref{eq:dc_2} becomes
\begin{equation}\label{eq:dc_3}
\sum\limits_{j=1}^3\int_{(L_\eta/R,\infty)}dz \varphi(z)  J_j(z,\delta; K_R,f_R)  = J \int_{ (L_\eta/R,\infty)} dz \varphi(z)
\end{equation}
for any $\varphi \in C_c(0,\infty)$. 

Let $\varepsilon >0$ arbitrarily small. Since the kernels $K_R,\ R\geq 1$ satisfy the Assumption \ref{assump:estimates} with $c_1,\ c_2$  in \eqref{eq:cond_kernel2}-\eqref{eq:cond_kernel3}  independent of $R$,
we can  apply Lemma \ref{lem:J} combined with \eqref{eq:dc_fR_3} to obtain
\begin{equation}\label{eq:dc_4}
\left| \sum_{j \in \{1,3\}} \int_{(L_\eta/R,\infty)} dz \varphi(z) J_j(z,\delta;K_R,f_R) \right| \leq C\varepsilon J \| \varphi \|_{L^\infty(0,\infty)} 
\end{equation}
where $C>0$ is independent of $R$. 

For every compact set $K\subset (0,\infty)$ we have that $\underset{z \in K}{\bigcup} (\Sigma_2(\delta) \cap \Omega_z)$ is bounded.
Then using \eqref{eq:dc_weak_conv} and 
 using that $\lim_{R\to \infty} K_R = K$ uniformly on compact sets of $(0,\infty)$ and that $\varphi$ is compactly supported, we obtain 
\begin{equation*}
\lim_{n\to \infty } \int_{(0,\infty)}dz \varphi(z) J_2(z,\delta ; K_{R_n},f_{R_n}) = \int_{(0,\infty)}dz \varphi(z) J_2(z,\delta;K,f)
\end{equation*}
 for any test function $\varphi \in C_c(0,\infty)$. Then using \eqref{eq:dc_3}-\eqref{eq:dc_4} we arrive at 
 \begin{equation*}
 \left| \int_{(0,\infty)}dz \varphi(z) J_2(z,\delta;K,f) - J\int_{(0,\infty)}dz \varphi(z) \right| \leq C\varepsilon J \| \varphi\|_{L^\infty(0,\infty)}.
 \end{equation*}
 Using again Lemma \ref{lem:J} and \eqref{eq:dc_fR_2} we deduce that
  \begin{equation*}
\left| \sum_{j \in \{1,3\}} \int_{(0,\infty)} dz \varphi(z) J_j(z,\delta;K,f) \right| \leq C\varepsilon J \| \varphi \|_{L^\infty(0,\infty)} 
 \end{equation*} 
 whence 
\begin{equation*}
\left|\int_{(0,\infty)}dz \varphi(z)\int_{ (0,z]}\int_{(z-x,\infty)} K(x,y)x f(dx) f(dy)  - J \int_{(0,\infty)}dz \varphi(z)\right| \leq C\varepsilon J \| \varphi \|_{L^\infty(0,\infty)} ,
\end{equation*}  
 for any $\varphi \in  C_c(0,\infty)$. Then since $\varepsilon$ is arbitrary small and $\varphi$ is an arbitrary test function, $f$ is a flux solution in the sense of Definition \ref{DefFluxSol2} and the result follows. 
\end{proofof}

\bigskip

\noindent \textbf{Acknowledgements.}
The authors are grateful to  P.\ Laurencot and R.\ L.\ Pego for interesting discussions about the content of this paper. We also thank  the anonymous reviewers for their careful reading of the manuscript and for valuable suggestions of improvements.  JL also thanks A.\  Vuoksenmaa for his help with the details of the existence proof in Section \ref{sec:existence}.
The authors gratefully acknowledge the support of the Hausdorff Research Institute for Mathematics
(Bonn), through the {\it Junior Trimester Program on Kinetic Theory},
of the CRC 1060 {\it The mathematics of emergent effects} at the University of Bonn funded through the German Science Foundation (DFG), 
and of the {\it Atmospheric Mathematics} (AtMath) collaboration of the Faculty of Science of University of Helsinki, as well as of  
the Academy of Finland via the {\it Centre of Excellence in Analysis and Dynamics Research} (project No. 307333).
The funders had no role in study design, analysis, decision to
publish, or preparation of the manuscript.

\bigskip

\noindent\textbf{Compliance with ethical standards} \smallskip

\noindent \textbf{Conflict of interest} The authors declare that they have no conflict of interest.

\bigskip 

 \bigskip
 
 \def\adresse{
\begin{description}

\item[M.~A. Ferreira] 
{Department of Mathematics and Statistics, University of Helsinki,\\ P.O. Box 68, FI-00014 Helsingin yliopisto, Finland \\
E-mail:  \texttt{marina.ferreira@helsinki.fi}}

\item[J. Lukkarinen]{ Department of Mathematics and Statistics, University of Helsinki, \\ P.O. Box 68, FI-00014 Helsingin yliopisto, Finland \\
E-mail: \texttt{jani.lukkarinen@helsinki.fi}}

\item[A. Nota:] { Institute for Applied Mathematics, University of Bonn, \\ Endenicher Allee 60, D-53115 Bonn, Germany\\
E-mail: \texttt{nota@iam.uni-bonn.de}}

\item[J.~J.~L. Vel\'azquez] { Institute for Applied Mathematics, University of Bonn, \\ Endenicher Allee 60, D-53115 Bonn, Germany\\
E-mail: \texttt{velazquez@iam.uni-bonn.de}}

\end{description}
}

\adresse
 
 \bigskip

\end{document}